\documentclass[ssy,preprint]{imsart}

\usepackage{amsthm,amsmath}
\RequirePackage[colorlinks,citecolor=blue,urlcolor=blue]{hyperref}

\arxiv{arXiv:1410.4231}

\startlocaldefs
\usepackage{geometry}                
\usepackage{graphicx}
\usepackage{xr}
\usepackage{hyperref}
\usepackage{amssymb,amsfonts}
\usepackage{dsfont}  
\usepackage{bbm}
\usepackage{epstopdf}
\usepackage{xargs}
\usepackage{ifthen}
\usepackage{aliascnt}
\usepackage{stmaryrd}
\usepackage{version}
\usepackage[numbers]{natbib}
\usepackage{algorithm2e}
\usepackage{enumerate}
\usepackage{color}

\usepackage[disable]{todonotes}   

\newtheorem{theorem}{Theorem}
\newaliascnt{proposition}{theorem}

\aliascntresetthe{proposition}

\newaliascnt{assumption}{theorem}

\aliascntresetthe{assumption}

\newaliascnt{lemma}{theorem}
\newtheorem{lemma}[lemma]{Lemma}
\aliascntresetthe{lemma}

\newaliascnt{corollary}{theorem}
\newtheorem{corollary}[corollary]{Corollary}
\aliascntresetthe{corollary}

\newaliascnt{definition}{theorem}
\newtheorem{definition}[definition]{Definition}
\aliascntresetthe{definition}

\newaliascnt{example}{theorem}
\newtheorem{example}[example]{Example}
\aliascntresetthe{example}

\newaliascnt{remark}{theorem}
\newtheorem{remark}[remark]{Remark}
\aliascntresetthe{remark}

\def\nset{\mathbb{N}}
\def\zset{\mathbb{Z}}

\newcommand{\PP}{\mathbb{P}}

\newcommand{\RR}{\mathbb{R}}

\newcommand{\ZZ}{\mathbb{Z}}

\def\et{\eta}
\def \e{\epsilon}

\def\e{\varepsilon}

\newcounter{hypA}
\newenvironment{hypA}{\refstepcounter{hypA}\begin{itemize}
  \item[{\bf (A\arabic{hypA})}]}{\end{itemize}}

\newcounter{hypB}
\newenvironment{hypB}{\refstepcounter{hypB}\begin{itemize}
  \item[{\bf (B\arabic{hypB})}]}{\end{itemize}}

\newcounter{hypAN}
\newenvironment{hypAN}{\refstepcounter{hypAN}\begin{itemize}
  \item[{\bf (AN\arabic{hypAN})}]}{\end{itemize}}

\newcounter{hypC}
\newenvironment{hypC}{\refstepcounter{hypC}\begin{itemize}
  \item[{\bf (C\arabic{hypC})}]}{\end{itemize}}

\newcounter{hypD}
\newenvironment{hypD}{\refstepcounter{hypD}\begin{itemize}
  \item[{\bf (D)}]}{\end{itemize}}

\newcounter{hypMG}
\newenvironment{hypMG}{\refstepcounter{hypMG}\begin{itemize}
  \item[{\bf (M)}]}{\end{itemize}}

\newcounter{hypLD}
\newenvironment{hypLD}{\refstepcounter{hypLD}\begin{itemize}
  \item[{\bf (L)}]}{\end{itemize}}

\newcounter{hypR}
\newenvironment{hypR}{\refstepcounter{hypR}\begin{itemize}
  \item[{\bf (S)}]}{\end{itemize}}

\newcommandx\LD[2][1=]{
\ifthenelse{\equal{#1}{}}
{\hspace{-1mm}(\textbf{L})\hspace{-1mm}}
{\hspace{-1mm}(\textbf{L\ref{#1}--\ref{#2}})\hspace{-1mm}}
}

\newcommandx\MG[2][1=]{
\ifthenelse{\equal{#1}{}}
{\hspace{-1mm}(\textbf{M})\hspace{-1mm}}
{\hspace{-1mm}(\textbf{M\ref{#1}--\ref{#2}})\hspace{-1mm}}
}

\newcommandx\A[2][1=]{
\ifthenelse{\equal{#1}{}}
{\hspace{-1mm}(\textbf{A\ref{#2}})\hspace{-1mm}}
{\hspace{-1mm}(\textbf{A\ref{#1}--\ref{#2}})\hspace{-1mm}}
}
\newcommandx\AN[2][1=]{
\ifthenelse{\equal{#1}{}}
{\hspace{-1mm}(\textbf{AN\ref{#2}})\hspace{-1mm}}
{\hspace{-1mm}(\textbf{AN\ref{#1}--\ref{#2}})\hspace{-1mm}}
}
\newcommandx\C[2][1=]{
\ifthenelse{\equal{#1}{}}
{\hspace{-1mm}(\textbf{C\ref{#2}})\hspace{-1mm}}
{\hspace{-1mm}(\textbf{C\ref{#1}--\ref{#2}})\hspace{-1mm}}
}

\newcommandx\D[2][1=]{
\ifthenelse{\equal{#1}{}}
{\hspace{-1mm}(\textbf{D})\hspace{-1mm}}
{\hspace{-1mm}(\textbf{D\ref{#1}--\ref{#2}})\hspace{-1mm}}
}
\newcommandx\R[2][1=]{
\ifthenelse{\equal{#1}{}}
{\hspace{-1mm}(\textbf{S})\hspace{-1mm}}
{\hspace{-1mm}(\textbf{S\ref{#1}--\ref{#2}})\hspace{-1mm}}
}
\newcommandx\CI[2][1=]{
\ifthenelse{\equal{#1}{}}
{\hspace{-1mm}(\textbf{D\ref{#2}})\hspace{-1mm}}
{\hspace{-1mm}(\textbf{D\ref{#1}--\ref{#2}})\hspace{-1mm}}
}

\newcommandx{\arr}[2][1=]{
\ifthenelse{\equal{#1}{}}
{U_{\N}({#2})}
{U_{\N}^{#1}({#2})}
}
\newcommandx{\asvar}[1][1=]{
\ifthenelse{\equal{#1}{}}
{\sigma^2}
{\sigma_{#1}^2}
}
\newcommandx{\ascovar}[1][1=]{
\ifthenelse{\equal{#1}{}}
{\Gamma}
{\Gamma_{#1}}
}

\newcommandx{\asvartd}[1][1=]{
\ifthenelse{\equal{#1}{}}
{\tilde{\sigma}^2}
{\tilde{\sigma}^2_{#1}}
}

\newcommandx\B[2][1=]{
\ifthenelse{\equal{#1}{}}
{\hspace{-1mm}(\textbf{B\ref{#2}})\hspace{-1mm}}
{\hspace{-1mm}(\textbf{B\ref{#1}--\ref{#2}})\hspace{-1mm}}
}

\newcommandx{\M}[1][1=]{
\ifthenelse{\equal{#1}{}}
{M}
{M_{\N}}
}

\newcommand{\bd}{Y_N}
\newcommand{\bdconst}{X_N}
\newcommand{\bmf}[1]{\mathsf{F}_{\mathrm{b}}(#1)}

\newcommandx\condmeas[2][1=]{
\ifthenelse{\equal{#1}{}}
{\mu_{#2}}
{\mu^{(#1)}_{#2}}
}

\newcommandx\asvarANone[1][1=]{
\ifthenelse{\equal{#1}{}}
{\nu^{2}}
{\nu_{#1}^{2}}
}

\newcommandx\asvartdANone[1][1=]{
\ifthenelse{\equal{#1}{}}
{\tilde{\nu}^{2}}
{\tilde{\nu}_{#1}^{2}}
}

\newcommandx\condmeastd[2][1=]{
\ifthenelse{\equal{#1}{}}
{\tilde{\mu}_{#2}}
{\tilde{\mu}_{#2}^{(#1)}}
}

\newcommandx{\cexp}[3][1=]{
\ifthenelse{\equal{#1}{}}
{\mathbb{E}\left[ #2 \mid #3 \right]} 
{\mathbb{E}[ #2 \mid #3 ]}
}

\newcommandx{\cov}[3][1=]{
\ifthenelse{\equal{#1}{}}
{\sigma(#2,#3)}
{\sigma_{#1}(#2,#3)}
}

\newcommandx{\covtd}[3][1=]{
\ifthenelse{\equal{#1}{}}
{\tilde{\sigma}(#2,#3)}
{\tilde{\sigma}_{#1}(#2,#3)}
}
\newcommand{\CV}[1][1=]{\operatorname{CV}^2_{\N}}
\newcommand{\CVthres}{\tau}

\newcommand{\DB}{B$^2$}
\newcommand{\DBAIS}{B$^2$\!ASIL}
\newcommandx{\der}[1][1=]{
\ifthenelse{\equal{#1}{}}
{w}
{w_{#1}}
}

\newcommandx{\Nder}[1][1=]{
\ifthenelse{\equal{#1}{}}
{\bold{w}}
{\bold{w}_{#1}}
}

\newcommand{\diff}[2]{\delta_{\N}(#1,#2)}
\newcommand{\difftd}[2]{\tilde{\delta}_{\N}(#1,#2)}
\newcommandx{\dlim}[1][1=]{
\ifthenelse{\equal{#1}{}}
{\stackrel{\mathcal D}{\longrightarrow}}
{\stackrel{\mathcal D}{\longrightarrow}}
}
\newcommand{\disc}{{\sf P}}

\newcommandx\epart[3][1=]{
\ifthenelse{\equal{#1}{}}
{\xi_\N(#2,#3)}
{\xi^{(#1)}_\N(#2,#3)}
}
\newcommandx\eparttd[3][1=]{
\ifthenelse{\equal{#1}{}}
{\tilde{\xi}_\N(#2,#3)}
{\tilde{\xi}^{(#1)}_\N(#2,#3)}
}
\newcommandx\epartck[3][1=]{
\ifthenelse{\equal{#1}{}}
{\check{\xi}_\N(#2,#3)}
{\check{\xi}^{(#1)}_\N(#2,#3)}
}

\newcommand{\eqsp}{} 
\newcommand{\eqdef}{\triangleq}
\newcommand{\E}{\mathbb{E}}
\newcommandx{\evar}[1][1=]{
\ifthenelse{\equal{#1}{}}
{Z_\infty}
{Z_{#1}}
}

\newcommand{\fd}{\mathcal{X}}

\newcommandx\fest[1][1=]{
\ifthenelse{\equal{#1}{}}
{\phi^\N}
{\phi_{#1}^\N}
}
\newcommand{\field}[1]{\mathcal{#1}}
\newcommandx\filt[1][1=]{
\ifthenelse{\equal{#1}{}}
{\field{F}_\N}
{\field{F}{(#1)}_\N}
}
\newcommandx\filttd[1][1=]{
\ifthenelse{\equal{#1}{}}
{\tilde{\field{F}}^\N}
{\tilde{\field{F}}_{#1}^\N}
}
\newcommandx\ffilt[1][1=]{
\ifthenelse{\equal{#1}{}}
{\field{G}^\N}
{\field{G}_{#1}^\N}
}

\newcommandx{\gaussfield}[1][1=]{
\ifthenelse{\equal{#1}{}}
{W^{N}}
{W_{#1}^{N}}
}
\newcommandx{\gaussfieldtd}[1][1=]{
\ifthenelse{\equal{#1}{}}
{\widetilde{W}^{N}}
{\widetilde{W_{#1}}^{N}}
}
\newcommandx{\cgfield}[1][1=]{
\ifthenelse{\equal{#1}{}}
{A^{N}}
{A_{#1}^{N}}
}

\newcommandx{\ngfield}[1]{\mathcal{W}_{#1}^{N}}
\newcommandx{\ungfield}[1]{\mathcal{Y}_{#1}^{N}}
\newcommandx{\gfield}[1]{\mathcal{V}_{#1}^{N}}

\newcommand{\genexp}{\nu}
\newcommandx{\genfd}[1][1=]{
\ifthenelse{\equal{#1}{}}
{\mathcal{F}}
{\mathcal{F}_{\N}}
}

\newcommand{\Hlw}{\varrho} 
\newcommand{\Hlwtd}{\tilde{\varrho}}
\newcommand{\Hc}{c_1}  
\newcommandx{\Hcons}[1][1=]{
\ifthenelse{\equal{#1}{}}
{c}
{c_{#1}}
}

\newcommand{\Hconstck}[1]{\check{c}_{#1}}
\newcommand{\Hconstd}[1]{\tilde{c}_{#1}}
\newcommand{\He}{c_2} 

\newcommandx{\inind}[3][1=]{
\ifthenelse{\equal{#1}{}}
{J_{\N}(#2,#3)}
{J_{\N}_{(#1)}(#2,#3)}
}
\newcommandx{\isind}[2][1=]{
\ifthenelse{\equal{#1}{}}
{I_{N}(#2)}
{I_{N}^{(#1)}(#2)}
}
\newcommand{\intvect}[2]{\llbracket #1, #2 \rrbracket}
\newcommandx\inwgt[3][1=]{
\ifthenelse{\equal{#1}{}}
{\omega_\N(#2,#3)}
{\omega^{(#1)}_\N(#2,#3)}
}
\newcommandx\inwgtbar[3][1=]{
\ifthenelse{\equal{#1}{}}
{\check{\omega}_\N(#2,#3)}
{\check{\omega}^{(#1)}_\N(#2,#3)}
}
\newcommand{\inwgtbd}{|\omega|_{\infty}}
\newcommandx\inwgttd[3][1=]{
\ifthenelse{\equal{#1}{}}
{\tilde{\omega}_\N(#2,#3)}
{\tilde{\omega}^{(#1)}_\N(#2,#3)}
}
\newcommand{\inwgttdbd}{|\tilde{\omega}|_{\infty}}
\newcommandx\iswgt[2][1=]{
\ifthenelse{\equal{#1}{}}
{\Omega_\N(#2)}
{\Omega^{(#1)}_\N(#2)}
}
\newcommandx\iswgtbar[2][1=]{
\ifthenelse{\equal{#1}{}}
{\check{\Omega}_\N(#2)}
{\check{\Omega}^{(#1)}_\N(#2)}
}
\newcommandx\iswgttd[2][1=]{
\ifthenelse{\equal{#1}{}}
{\tilde{\Omega}_\N(#2)}
{\tilde{\Omega}^{(#1)}_\N(#2)}
}

\newcommand{\ubmarkov}{\sigma_+}
\newcommand{\lbmarkov}{\sigma_-}
\newcommand{\lbop}{c_{-}}
\newcommandx{\ulim}[1][1=]{
\ifthenelse{\equal{#1}{}}
{\stackrel{}{\longrightarrow}}
{\xrightarrow[{#1} \rightarrow +\infty]{}}
}

\newcommandx{\Mk}[1][1=]{
\ifthenelse{\equal{#1}{}}
{M}
{M_{#1}}
}

\newcommandx{\mbd}[1][1=]{
\ifthenelse{\equal{#1}{}}
{V_{\N}}
{V^{#1}_{\N}}
}
\newcommand{\meas}[1]{\set{M}(#1)}
\newcommand{\mf}[1]{\set{F}(#1)}
\newcommand{\minmeas}{\varphi}

\newcommand{\mutation}[2]{\mathsf{Mut}\langle#1\rangle}

\SetCommentSty{mycommfont}

\newcommand{\N}{N}
\newcommandx\Nin[1][1=]{
\ifthenelse{\equal{#1}{}}
{N_2}
{N_2(#1)}
}
\newcommandx\Nis[1][1=]{
\ifthenelse{\equal{#1}{}}
{N_1}
{N_1(#1)}
}
\newcommand{\nn}{\mathbb{N}}
\newcommand{\nsetpos}{\mathbb{N}^*}
\newcommand{\normdist}{{\sf N}}

\newcommand{\1}[1]{\mathbbm{1}_{#1}}
\newcommandx{\osc}[2][1=]{
\ifthenelse{\equal{#1}{}}
{\operatorname{osc}(#2)}
{\operatorname{osc}^2(#2)}
}

\newcommandx{\op}[1][1=]{
\ifthenelse{\equal{#1}{}}
{Q}
{Q_{#1}}
}
\newcommandx{\optd}[1][1=]{
\ifthenelse{\equal{#1}{}}
{\tilde{Q}}
{\tilde{Q}_{#1}}
}

\newcommand{\pind}[2][]{I_{#1}\ifthenelse{\equal{#1}{}}{}{^{#2}}}
\newcommandx{\plim}[1][1=]{
\ifthenelse{\equal{#1}{}}
{\stackrel{\prob}{\longrightarrow}}
{\stackrel{\prob}{\longrightarrow}}
}
\newcommandx{\pot}[1][1=]{
\ifthenelse{\equal{#1}{}}
{g}
{g_{#1}}
}
\newcommand{\ppart}[2][]{X_{#1}\ifthenelse{\equal{#1}{}}{}{^{#2}}}
\newcommand{\prob}{\mathbb{P}}
\newcommand{\probmeas}[1]{\set{M}_1(#1)}
\newcommandx{\prop}[1][1=]{
\ifthenelse{\equal{#1}{}}
{R}
{R_{#1}}
}

\newcommandx{\pvar}[1][1=]{
\ifthenelse{\equal{#1}{}}
{X_\infty}
{X_{#1}}
}
\newcommandx{\Nprop}[1][1=]{
\ifthenelse{\equal{#1}{}}
{\bold{R}}
{\bold{R}_{#1}}
}

\newcommandx{\NNprop}[1][1=]{
\ifthenelse{\equal{#1}{}}
{\mathcal{R}}
{\mathcal{R}_{#1}}
}

\newcommand{\randpar}[1]{Z_{#1}}
\newcommand{\rmd}{\mathrm{d}}
\newcommand{\rset}{\mathbb{R}}
\newcommand{\rsetnonneg}{\mathbb{R}_+}
\newcommand{\rsetpos}{\mathbb{R}^\ast_+}

\newcommand{\set}[1]{\mathsf{#1}}
\newcommand{\SIL}{\mathsf{SIL}}
\newcommand{\SiL}{\mathsf{SiL}}

\newcommandx{\supn}[2][1=]{
\ifthenelse{\equal{#1}{}}
{\left \| #2 \right \|_\infty}
{ \| #2 \|_\infty}
}
\newcommand{\stsp}{\set{X}}

\newcommand{\stsptd}{\tilde{\set{X}}}
\newcommand{\stfd}{\mathcal{X}}
\newcommand{\stfdtd}{\tilde{\mathcal{X}}}

\newcommand{\targ}[1][]{
\ifthenelse{\equal{#1}{}}
{\eta}
{\eta_{#1}}
}
\newcommand{\untarg}[1][]{
\ifthenelse{\equal{#1}{}}
{\gamma}
{\gamma_{#1}}
}
\newcommand{\targtd}[1][]{
\ifthenelse{\equal{#1}{}}
{\tilde{\eta}}
{\tilde{\eta}_{#1}}
}

\newcommandx{\var}[1][1=]{
\ifthenelse{\equal{#1}{}}
{\varsigma^2}
{\varsigma_{#1}^2}
}
\newcommandx\varmat[1]{D_{#1}}

\newcommand{\WA}{archipelago}

\newcommand{\Xinit}{\chi}
\newcommandx{\X}[2][1=]{
\ifthenelse{\equal{#1}{}}
{X_{\N}(#2)}
{X^{#1}_{\N}(#2)}
}

\newcommand{\Yset}{\set{Y}}
\newcommand{\Yfd}{\mathcal{Y}}

\newcommand{\Zset}{\set{Z}}
\newcommand{\Zfd}{\mathcal{Z}}

\newcommand{\separt}[2][]{
\ifthenelse{\equal{#1}{}}
{\xi^{\N}(#2)}
{\xi_{#1}^{\N}(#2)}
}
\newcommand{\Nepart}[2][]{
\ifthenelse{\equal{#1}{}}
{\bold{\Xi}^{\N}(#2)}
{\bold{\Xi}_{#1}^{\N}(#2)}
}

\newcommand{\separttd}[2][]{
\ifthenelse{\equal{#1}{}}
{\tilde{\xi}^{\N}(#2)}
{\tilde{\xi}_{#1}^{\N}(#2)}
}
\newcommand{\sinwgt}[2][]{
\ifthenelse{\equal{#1}{}}
{\omega^{\N}(#2)}
{\omega_{#1}^{\N}(#2)}
}
\newcommand{\sinwgttd}[2][]{
\ifthenelse{\equal{#1}{}}
{\tilde{\omega}^{\N}(#2)}
{\tilde{\omega}_{#1}^{\N}(#2)}
}
\newcommand{\siswgt}[2][]{
\ifthenelse{\equal{#1}{}}
{\Cmega^{\N}(#2)}
{\Cmega_{#1}^{\N}(#2)}
}

\newcommand{\kac}[1][]{\eta\ifthenelse{\equal{#1}{}}{}{_{#1}}}
\newcommand{\kactd}[1][]{
\tilde{\eta}
\ifthenelse{\equal{#1}{}}{}{_{#1}}
}

\newcommand{\unkac}[1][]{\gamma\ifthenelse{\equal{#1}{}}{}{_{#1}}}

\numberwithin{equation}{section}
\numberwithin{theorem}{section}
\numberwithin{lemma}{section}
\numberwithin{definition}{section}
\numberwithin{remark}{section}
\numberwithin{corollary}{section}

\endlocaldefs

\begin{document}

\begin{frontmatter}

\title{Convergence properties of weighted particle islands with application to the double bootstrap algorithm}
\runtitle{Convergence properties of weighted particle islands}

\author{\fnms{Pierre} \snm{Del Moral}\thanksref{e2}\ead[label=e2,mark]{p.del-moral@unsw.edu.au}},

\author{\fnms{Eric} \snm{Moulines}\thanksref{e3}\ead[label=e3,mark]{eric.moulines@polytechnique.edu}},

\author{\fnms{Jimmy} \snm{Olsson}\thanksref{e4,t4} \ead[label=e4,mark]{jimmyol@kth.se}}
\thankstext{t4}{J.~Olsson is supported by the Swedish Research Council, Grant 2011-5577.}

\author{\fnms{Christelle} \snm{Verg\'e}\corref{}\thanksref{e1,t1}\ead[label=e1,mark]{christelle.verge@onera.fr}}
\thankstext{t1}{C.~Verg\'e is supported by CNES (Centre National d'\'Etudes Spatiales) and ONERA, The French Aerospace Lab.},

\runauthor{P. Del Moral et al.}
\affiliation{University of New South Wales\thanksref{e2},
\'Ecole Polytechnique, INRIA XPOP\thanksref{e3},
KTH Royal Institute of Technology\thanksref{e4},
ONERA/CNES\thanksref{e1}
}

\address[e2]{School of Mathematics and Statistics \\
High Street, Kensington \\
NSW 2052 Australia \\
\printead{e2}}

\address[e3]{ Ecole Polytechnique, \\ 
Centre de Math\'ematiques Appliqu\'ee, INRIA XPOP\\
Route de Saclay  \\
91128 Palaiseau, Cedex \\
France \\
\printead{e3}}

\address[e4]{Department of Mathematics \\
KTH \\
SE-100 44 Stockholm \\
Sweden \\
\printead{e4}}

\address[e1]{ONERA - The French Aerospace Lab \\
91761 Palaiseau \\
France \\
\printead{e1}}

\begin{abstract}
    Particle island models \cite{verge:dubarry:delmoral:moulines:2013} provide a means of parallelization of sequential Monte Carlo methods, and in this paper we present novel convergence results for algorithms of this sort. In particular we establish a central limit theorem---as the number of islands and the common size of the islands tend jointly to infinity---of the double bootstrap algorithm with possibly adaptive selection on the island level. For this purpose we introduce a notion of archipelagos of weighted islands and find conditions under which a set of convergence properties are preserved by different operations on such archipelagos. This theory allows arbitrary compositions of these operations to be straightforwardly analyzed, providing a very flexible framework covering the double bootstrap algorithm as a special case. Finally, we establish the long-term numerical stability of the double bootstrap algorithm by bounding its asymptotic variance under weak and easily checked assumptions satisfied typically for models with non-compact state space.
\end{abstract}


\end{frontmatter}

\section{Introduction}
\label{Introduction}
This paper discusses approaches to parallelization of \emph{sequential Monte Carlo} (SMC) methods (or \emph{particle filters}) approximating normalized \emph{Feynman-Kac distribution flows}. At present, SMC methods are used successfully for online sampling from sequences of complex distributions in a wide range of applications, including nonlinear filtering, signal processing, data assimilation \cite[see, e.g.,][and the references therein]{doucet:defreitas:gordon:2001, chopin:2002, ristic:arulampalam:gordon:2004, cappe:moulines:2005, crisan:rozovskii:2011}, and rare event analysis \cite{delmoral:garnier:2005, cerou:delmoral:furon:guyader:2012}. These algorithms evolve, recursively and randomly in time, a sample of random draws, \emph{particles}, with associated \emph{importance weights}.
The particle cloud is updated through \emph{selection} and \emph{mutation} operations, where the former duplicates or eliminates, through resampling, particles with large or small importance weights, respectively, while the latter disseminates randomly the particles over the state space and updates accordingly the importance weights for further selection.

SMC methods are computationally intensive, which may be critical in online applications. In particular, since the particle interaction enforced by the selection operation is of ``global'' nature (as it draws, with replacement, each particle from the entire particle population rather than from a subset of the same), running SMC methods in parallel on multicore processors is not straightforward. A natural ideal, which is the basis also for the present paper, is to parallelize the algorithm by, instead of considering a single batch of $\N$ particles, simply dividing the particle population into $\Nis$ batches of each $\Nin$ particles (i.e., $\N = \Nis \Nin$), where each batch is referred to as a \emph{particle island} (or simply an \emph{island}).

Parallel implementation of SMC was first proposed in \cite{bolic:djuric:hong:2005} in the form of an algorithm referred to as the \emph{local exchange particle filter} (LEPF), in which groups of particles are spread across computational units.  This algorithm was later improved in \cite{balasingam:bolic:djuric:miguez:2011} (see also \cite{heine:whiteley:2015} where a detailed convergence analysis of the LEPF is carried out). As indicated by the almost $300$ Google Scholar citations at the time of writing, the LEPF has triggered a substantial interest in parallelization of SMC. Most notably, variations of the LEPF are found in the contexts of multitarget tracking \cite{sutharsan:kirubaraja:lang:2012}, optical tracking \cite{sankaranarayanan:srivastava:chellappa:2008}, and state estimation \cite{rosen:medvedev:2013}.

In the present paper we consider an algorithm suggested in \cite{verge:dubarry:delmoral:moulines:2013}, which may be viewed as a variant of the LEPF algorithm.  In this framework, each island evolves according to the standard SMC scheme subjecting alternatingly the subpopulation to selection and mutation. Unfortunately, the division of the particle population introduces additional bias which may be of note for moderate island sizes $\Nin$. Thus, in \cite{verge:dubarry:delmoral:moulines:2013} it is proposed to reduce this bias by performing additional selection also on the \emph{island level} by resampling multinomially, when needed, the islands according to probabilities proportional to the weight averages over the different subpopulations. Selection on the island level may be performed systematically, as in the \emph{double bootstrap} (\DB) \emph{algorithm} (in the present paper we have chosen to denote the algorithm ``{\DB}'' rather than ``2B'', as we consider it more correctly described as a ``square bootstrap'' rather than a ``double bootstrap''; nevertheless, the algorithm must not be confused with the \emph{SMC square} (SMC$^2$) \emph{algorithm} proposed in \cite{chopin:jacob:papaspiliopoulos:2012}, which is, if still of a related form, of a different nature) or may be activated adaptively by some criterion measuring the skewness of the island weights. The latter approach will be referred to by us as the \emph{double bootstrap algorithm with adaptive selection on the island level} (\DBAIS). At the end of the day, a sequence of Monte Carlo estimators is obtained by weighing up, using the island weights, the self-normalized empirical measures associated with the different particle islands.

Needless to say, the theoretical analysis of \DB-type algorithms is challenging due to the intricate dependence structure imposed by the island selection operation and the ``double asymptotics'' introduced by the two sample sizes $\Nis$Ê and $\Nin$. The authors of \cite{verge:dubarry:delmoral:moulines:2013}, who base their theoretical analysis on a reformulation of the particle island model as an extended Feynman-Kac model on an augmented space of dimension $\Nin$, detour the latter difficulty by letting first the number $\Nis$ of islands and then the number $\Nin$ of individuals of each island tend to infinity. By \emph{separating} the asymptotics in this manner, the analysis can, not surprisingly, at least in the case of the {\DB} algorithm, be handled using classical techniques from SMC analysis, and in this way the authors establish convergence of bias and variance when these quantities are scaled with the size $\N$Ê of the system. However, working with this somewhat synthetic mode of convergence (with separated limits), the authors fail to supplement their consistency results with a central limit theorem (CLT). Moreover, they do not provide any convergence results for the {\DBAIS} algorithm.

Nevertheless, even though the islands are allowed to interact through selection, any two individuals of the system should become more and more statistically independent as the number of islands as well as the size of the islands grow (cf. the \emph{propagation of chaos} property of standard SMC methods \cite{delmoral:2004}). Thus, we may expect a law of large number as well as a CLT to hold when $\Nis$ and $\Nin$ tend \emph{jointly} to infinity. Moreover, in analogy with similar result for standard, single batch SMC methods \cite[see][]{delmoral:guionnet:1999, chopin:2004, kuensch:2005, douc:moulines:2008}, we may expect the rate of such a CLT to be $\sqrt{\N}$.

The aim of the present paper is to improve the existing theoretical analysis of particle island models by establishing results of the previous type. For this purpose we will introduce a notion of \emph{archipelagos of weighted islands} that generalizes the particle models studied in \cite{verge:dubarry:delmoral:moulines:2013} and consider three kinds of convergence properties of such archipelagos, namely \emph{consistency} (convergence in probability), \emph{asymptotic normality} (convergence in distribution in terms of a CLT with rate $\sqrt{N}$), and \emph{large deviation} (an exponential inequality of Hoeffing-type that holds uniformly over all islands). After this, we perform single-step analyses of three kinds of operations on archipelagos, namely \emph{selection on the island level}, \emph{selection on the individual level}, and \emph{mutation}, and show how these operations preserve the convergence properties under consideration. As a consequence, we are able to establish that the convergence properties in question are preserved through an \emph{arbitrary} composition of the mentioned operations, including the {\DB} algorithm as a special case, and to provide explicit expressions of the associated asymptotic variance. Moreover, the flexibility of our results, which generalize those obtained in \cite{douc:moulines:2008} for standard, single batch SMC methods, makes these well-suited for analyzing particle island algorithms with adaptive resampling strategies such as the {\DBAIS} scheme, for which we provide a detailed analysis (including a CLT). In our proofs, which rely on limit theorems for triangular arrays obtained in \cite{douc:moulines:2008}, the working process is highly inductive. Since the intricate dependence structures of the particle model force us to define triangular arrays on the island level, we will often, when establishing the preservation of a certain convergence property of a certain operation, face a situation where the only way of obtaining some critical limit or bound is to add the same to the list of induction hypotheses. After this, one establishes that the operation in question preserves also this additional property (limit or bound), by possibly adding, if needed, further assumptions to the list, and so on. At the end of the day, we have obtained a more or less minimal set (a hexad in the case of asymptotic normality) of properties that need to be checked at each induction step. In this machinery, the large deviation property is a critical component, since it provides, as a consequence of the distribution-free character of Hoeffding-type inequalities, uniform control of the deviation of the empirical measures associated with the different islands from their common mean.

As a last contribution, we establish the numerical stability of the {\DB} algorithm by bounding uniformly the asymptotic variance of its output. We carry through this analysis under a \emph{strong mixing condition} as well a \emph{local Doeblin condition} (see \autoref{proof_bound:var} for details), where the latter is considerably weaker than the former and easily verified for a large variety of models with possibly non-compact state space. When operating under the local Doeblin condition, we let the Feynman-Kac model be indexed by a strictly stationary sequence of random parameters (corresponding, e.g., to random observations in the case of optimal filtering in \emph{hidden Markov models}) and show, using novel results in \cite{douc:moulines:olsson:2014}, that the sequence of asymptotic variances is stochastically bounded (tight) in this setting. On the other hand, imposing the strong mixing assumption, which is classical in the literature of SMC analysis \cite{delmoral:guionnet:2001,delmoral:2004}, allows an explicit, deterministic uniform variance bound to be obtained using standard methods.

To sum up, the contribution of the present paper is threefold, since we
\begin{itemize}
    \item introduce a general theory of archipelagos of weighted particle islands and analyze thoroughly the convergence properties, as the number $\Nis$Ê of islands and the common size $\Nin$ of the islands tend jointly to infinity, of such objects when subjected to certain operations. For this purpose, we develop a machinery that allows triangular arrays defined on the island level to be analyzed and which may be used for handling double asymptotics appearing in other kinds of island-type particle algorithms.
    \item apply the previous theoretical results to the {\DB} and {\DBAIS} algorithms, yielding laws of large numbers and CLTs for these schemes.
    \item establish the long-term stability of the {\DB} algorithm under weak and easily checked assumptions.
\end{itemize}

The paper is organized as follows. In \autoref{sec:preliminaries} we introduce, after some prefatory notation, the concept of archipelagos of weighted islands, and define the three different convergence properties of such archipelagos. Our main results are, along with the three different operations under consideration, presented in \autoref{sec:main:results}, and \autoref{sec:applications} discusses the application of these results to the {\DBAIS} algorithm. In particular, in \autoref{coro:ESS:bootstrap} we establish the asymptotic normality of this algorithm, which implies the asymptotic normality of the {\DB} algorithm as a special case (see \autoref{coro:systematic:SIL}), and provide a formula for the asymptotic variance; moreover, in \autoref{long-term:stability:compact:case} establish the long-term stability of the algorithm by showing that the asymptotic variance of the {\DB} algorithm may, under suitable assumptions, be bounded uniformly. The most significative proofs are gathered in \autoref{sec:proofs}, and in order to avoid repetition we have put some additional proofs using similar techniques in the supplementary paper \cite{verge:delmoral:moulines:olsson:2014:supp}. Finally, \autoref{appenbdix:technical:lemmas} provides some technical results that are used frequently in \autoref{sec:proofs}.


\section{Preliminaries}
\label{sec:preliminaries}

\subsection{Some notation}
\label{sec:notation}
For $(m, n) \in \zset^2$ such that $m \leq n$ we denote $\intvect{m}{n} \eqdef \{m, m + 1, \ldots, n \} \subset \ZZ$. Moreover, we denote by and $\rsetnonneg$ and $\rsetpos$ the sets of nonnegative and positive real numbers, respectively, and by $\nsetpos$ the set of positive integers. For any quantities $\{Êa_\ellÊ\}_{\ell = m}^n$ we will use the vector notation $a_{m:n} = (a_m, \ldots, a_n)$ with the convention $a_{m:n} = \varnothing$ if $m > n$.

In the sequel we assume that all random variables are defined on a common probability space $(\Omega,\mathcal{F},\PP)$. For some given measurable space $(\stsp, \stfd)$ we denote by $\meas{\stfd}$ and $\probmeas{\stfd} \subset \meas{\stfd}$ the set of measures and probability measures on $(\stsp, \stfd)$, respectively. In addition, we denote by $\mf{\stfd}$ the set of real-valued measurable functions on $(\stsp, \stfd)$ and by $\bmf{\stfd} \subset \mf{\stfd}$ the set of bounded such functions. For $h \in \bmf{\stfd}$ we denote the sup norm $\supn[txt]{h} \eqdef \sup_{x \in \stsp} |h(x)|$ and the oscillator norm $\osc{h} \eqdef \sup_{(x, x') \in \stsp^2} |h(x) - h(x')|$. For any $\nu \in \meas{\stfd}$ and $f \in \mf{\stfd}$ we denote by $\nu f \eqdef \int f(x) \, \nu(\rmd x)$  the Lebesgue integral of $f$ under $\nu$ whenever this is well-defined. Now, given also some other $(\Yset, \Yfd)$ measurable space, an \emph{unnormalized transition kernel} $K$ from $(\stsp, \stfd)$ to $(\Yset, \Yfd)$ is a mapping from $\stsp \times \Yfd$ to $\rsetnonneg$ such that for all $\set{A} \in \Yfd$, $x \mapsto K(x, \set{A})$ is a nonnegative measurable function on $\stsp$ and for all $x \in \stsp$, $\set{A} \mapsto K(x, \set{A})$ is a measure on $(\Yset, \Yfd)$. If $K(x, \Yset) = 1$ for all $x \in \stsp$, then $K$ is called a \emph{transition kernel} (or simply a \emph{kernel}). The kernel $K$ induces two integral operators, one acting on functions and the other on measures. More specifically, let $f \in \mf{\stfd}$ and $\nu \in \meas{\stfd}$ and define the measurable function
$$
	K f : \stsp \ni x \mapsto \int f(y) \, K(x, \rmd y)
$$
and the measure
$$
	\nu K : \Yfd \ni \set{A} \mapsto \int K(x, \set{A}) \, \nu(\rmd x)
$$
whenever these quantities are well-defined. Finally, let $K$ be as above and let $L$ be another unnormalized transition kernels from $(\Yset, \Yfd)$ to some third measurable space $(\Zset, \Zfd)$; then we define the \emph{product} of $K$ and $L$ as the unnormalized transition kernel
$$
K L : \stsp \times \Zfd \ni (x, \set{A}) \mapsto \int K(x, \rmd y) \, L(y, \set{A})
$$
from $(\stsp, \stfd)$ to $(\Zset, \Zfd)$ whenever this is well-defined.

\subsection{Weighted particle islands and archipelagos}
\label{sec:weighted:archipelagos}
Let $\{ \Nis[\N] \}_{\N \in \nsetpos}$ and $\{ \Nin[\N] \}_{\N \in \nsetpos}$ be sequences of positive integers such that $\Nis[\N] \Nin[\N] = \N$ for all $\N \in \nsetpos$ and $\Nis[\N] \rightarrow \infty$ and $\Nin[\N] \rightarrow \infty$ as $\N \rightarrow \infty$. For lucidity we will often omit the index $\N$ from the notation and write simply $\Nis$ and $\Nin$. In the following, let $\{ (\epart{i}{j}, \inwgt{i}{j}) ; (i, j) \in \intvect{1}{\Nis} \times \intvect{1}{\Nin} \}$ be an array of $\stsp$-valued random variables (the $\xi_\N$) with associated nonnegative (possibly unnormalized) weights (the $\omega_\N$). For each $i \in \intvect{1}{\Nis}$, the subset $\{ (\epart{i}{j}, \inwgt{i}{j}) \}_{j = 1}^{\Nin}$ of the array will be referred to as an \emph{island}. With this terminology, a random variable $\epart{i}{j}$ in the array will be referred to as an \emph{individual} or a \emph{particle}. Finally, we associate each island $\{ (\epart{i}{j}, \inwgt{i}{j}) \}_{j = 1}^{\Nin}$ with a nonnegative (possibly unnormalized) weight $\iswgt{i}$. In the following, the set $\{ (\iswgt{i}, \{ (\epart{i}{j}, \inwgt{i}{j}) \}_{j = 1}^{\Nin}) \}_{i = 1}^{\Nis}$ of islands with associated weights will be referred to as an \emph{\WA} on $(\stsp, \stfd)$. We will always require the island weights to be positive and the particle weights to be positive and uniformly bounded, i.e., there exists some constant $\inwgtbd$ such that $0 < |\inwgt{i}{j}| \leq \inwgtbd$ for all $(i, j) \in \intvect{1}{\Nis[\N]} \times \intvect{1}{\Nin[\N]}$ and $\N \in \nsetpos$.


\subsection{Convergence properties of archipelagos}
\label{sec:properties:WAs}

In the following, any limit ($\longrightarrow$), limit in probability ($\plim$), and limit in distribution ($\dlim$) is supposed to hold as $\N \rightarrow \infty$ if not specified differently.

\begin{definition}[consistency] 
    An archipelago $\{(\iswgt{i}, \{(\epart{i}{j}, \inwgt{i}{j}) \}_{j = 1}^{\Nin}) \}_{i = 1}^{\Nis}$ on $(\stsp, \stfd)$ is said to be \emph{consistent for} $\targ \in \probmeas{\stfd}$ if for all $h \in \bmf{\stfd}$,
    \begin{hypC} \label{def:consistency}
    $
    	\displaystyle \quad \sum_{i = 1}^{\Nis} \frac{\iswgt{i}}{\sum_{i' = 1}^{\Nis} \iswgt{i'}} \sum_{j = 1}^{\Nin} \frac{\inwgt{i}{j}}{\sum_{j' = 1}^{\Nin} \inwgt{i}{j'}} h(\epart{i}{j}) \plim[\N] \targ h \eqsp,
    $
    \end{hypC}
    \begin{hypC} \label{def:consistency:max}
    $
    	\displaystyle \quad \max_{i \in \intvect{1}{\Nis}} \frac{\iswgt{i}}{\sum_{i' = 1}^{\Nis} \iswgt{i'}} \plim[N] 0 \eqsp.
    $
    \end{hypC}
\end{definition}
Note that the estimator in \C{def:consistency} assigns the weight $\iswgt{i} / \sum_{i' = 1}^{\Nis} \iswgt{i'}$ to the self-normalized importance sampling estimator $\sum_{j = 1}^{\Nin} \inwgt{i}{j} h(\epart{i}{j}) / \sum_{j' = 1}^{\Nin} \inwgt{i}{j'}$ associated with island $i \in \intvect{1}{\Nis}$, and the smallness condition \C{def:consistency:max} formalizes the fact that this weight, and thus the contribution of each island to the estimator associated with the archipelago as a whole, should vanish asymptotically as $N \rightarrow \infty$.


\begin{definition}[exponential deviation] \label{def:deviation:inequalities}
        In the following, let $\targ \in \probmeas{\stfd}$ and $\Hlw$ and $\{c_\ell \}_{\ell = 1}^2$ be positive constants. An archipelago $\{(\iswgt{i}, \{(\epart{i}{j}, \inwgt{i}{j}) \}_{j = 1}^{\Nin}) \}_{i = 1}^{\Nis}$  on $(\stsp, \stfd)$ is said to satisfy \emph{exponential deviation for} $(\targ, \Hlw, \{ c_{\ell} \}_{\ell = 1}^2)$ if for all $h \in \bmf{\stfd}$, $\Nis \in \nsetpos$, $\Nin \in \nsetpos$, and $\varepsilon > 0$,
	\begin{hypD} \label{def:deviation:inequality:individual}
	$
    		\displaystyle \prob \left( \max_{i \in \intvect{1}{\Nis}} \left| \frac{1}{\Nin} \sum_{j = 1}^{\Nin} \inwgt{i}{j} h(\epart{i}{j}) - \Hlw \times \targ h \right| \geq \varepsilon \right) \leq \Hc \Nis \exp \left( - \He \Nin \frac{\varepsilon^2 }{\supn{h}^2} \right) \eqsp.
	 $
	\end{hypD}
\end{definition}
The exponential deviation inequality in \D{def:deviation:inequality:individual} provides uniform control on the deviations of the unnormalized importance sampling estimators $\sum_{j = 1}^{\Nin} \inwgt{i}{j} h(\epart{i}{j}) / \Nin$, $i \in \intvect{1}{\Nis}$, associated with the different islands from their common mean level $\Hlw \times \targ h$. The factor $\Nis$ on the right hand side of the equality is required to compensate for the maximum with respect to the island index. Assumption~\D{def:deviation:inequality:individual} implies, by a straightforward extension of the generalized Hoeffding inequality derived in \cite[Lemma~4]{douc:garivier:moulines:olsson:2010}, that also the deviations of the properly normalized importance sampling estimators associated with the different islands from the expectations targeted by the archipelago can be uniformly controlled as follows.
\begin{lemma} \label{lem:deviation:properly:normalized}
    Assume that \D{def:deviation:inequality:individual} holds for $(\targ, \Hlw, \{ c_{\ell} \}_{\ell = 1}^2)$. Then for all $h \in \bmf{\fd}$, $\Nis \in \nsetpos$, $\Nin \in \nsetpos$, and $\varepsilon > 0$,
    \begin{equation} \label{eq:deviation:normalized:estimator}
        \prob \left( \max_{i \in \intvect{1}{\Nis}} \left| \sum_{j = 1}^{\Nin} \frac{\inwgt{i}{j}}{\sum_{j' = 1}^{\Nin} \inwgt{i}{j'}} \{h(\epart{i}{j}) - \targ h \} \right| \geq \varepsilon \right)
        \leq 2 \Hc \Nis \exp \left( - \He \Nin \frac{\varepsilon^2 \Hlw^2}{4 \supn{h}^2} \right) \eqsp.
    \end{equation}
\end{lemma}

Finally, we introduce a third convergence property describing weak convergence in the sense of a CLT. Let $\normdist$ denote the Gaussian distribution.

\begin{definition}[asymptotic normality] \label{def:asymptotic:normality}
	In the following, let
	\begin{itemize}
		\item $\asvar : \bmf{\stfd} \rightarrow \rset_+^*$ and $\asvarANone : \bmf{\stfd} \rightarrow \rset_+^*$ be functionals.
		\item $\targ \in \probmeas{\stfd}$ and $\{Ê\condmeas{\ell} \}_{\ell = 1}^3 \subset \meas{\stfd}$ be measures.
	\end{itemize}
	An archipelago $\{(\iswgt{i}, \{(\epart{i}{j}, \inwgt{i}{j}) \}_{j = 1}^{\Nin}) \}_{i = 1}^{\Nis}$
	on $(\stsp, \stfd)$ is said to be \emph{asymptotically normal for} $(\targ, \asvar, \asvarANone, \{\condmeas{\ell} \}_{\ell = 1}^3)$ if for all $h \in \bmf{\stfd}$,
	\begin{hypAN} \label{ass:AN:CLT}
	$
		\displaystyle \quad \sqrt{\N} \sum_{i = 1}^{\Nis} \frac{\iswgt{i}}{\sum_{i' = 1}^{\Nis} \iswgt{i'}} \sum_{j = 1}^{\Nin} \frac{\inwgt{i}{j}}{\sum_{j' = 1}^{\Nin} \inwgt{i}{j'}} \{ h(\epart{i}{j})  - \targ h \} \dlim[\N] \normdist(0, \asvar(h))
	$
	\end{hypAN}
	and, in addition,
	\begin{hypAN} \label{ass:AN:squared:islands}
	$
		\displaystyle \quad \Nin \sum_{i = 1}^{\Nis} \frac{\iswgt{i}}{\sum_{i' = 1}^{\Nis} \iswgt{i'}} \left( \sum_{j = 1}^{\Nin} \frac{\inwgt{i}{j}}{\sum_{j' = 1}^{\Nin} \inwgt{i}{j'}} \{Êh(\epart{i}{j}) - \targ h \} \right)^2 \plim[N] \asvarANone(h) \eqsp,
	$
	\end{hypAN}
	\begin{hypAN} \label{ass:AN:squared:iswgts}
	$
		\displaystyle \quad \Nis \sum_{i = 1}^{\Nis} \left( \frac{\iswgt{i}}{\sum_{i' = 1}^{\Nis} \iswgt{i'}} \right)^2 \sum_{j = 1}^{\Nin} \frac{\inwgt{i}{j}}{\sum_{j' = 1}^{\Nin} \inwgt{i}{j'}} h(\epart{i}{j}) \plim[N] \condmeas{1} h \eqsp,
	$
	\end{hypAN}
	\begin{hypAN} \label{ass:AN:mut-1}
	$
		\displaystyle \quad \N \sum_{i = 1}^{\Nis} \left( \frac{\iswgt{i}}{\sum_{i' = 1}^{\Nis} \iswgt{i'}} \right)^2 \sum_{j = 1}^{\Nin} \left( \frac{\inwgt{i}{j}}{\sum_{j' = 1}^{\Nin} \inwgt{i}{j'}} \right)^2 h(\epart{i}{j}) \plim[N] \condmeas{2} h \eqsp,
	$
	\end{hypAN}
	\begin{hypAN} \label{ass:AN:mut-2}
	$
		\displaystyle \quad \Nin \sum_{i = 1}^{\Nis} \frac{\iswgt{i}}{\sum_{i' = 1}^{\Nis} \iswgt{i'}} \sum_{j = 1}^{\Nin} \left( \frac{\inwgt{i}{j}}{\sum_{j' = 1}^{\Nin} \inwgt{i}{j'}} \right)^2 h(\epart{i}{j}) \plim[N] \condmeas{3} h \eqsp,
$
	\end{hypAN}
	\begin{hypAN} \label{def:tightness:island}
	$
		\displaystyle \quad \lim_{\lambda \rightarrow \infty} \sup_{N \in \nsetpos} \prob \left(\max_{i \in \intvect{1}{\Nis}} \Nis \frac{\iswgt{i}}{\sum_{i'=1}^{\Nis}\iswgt{i'}} \geq \lambda \right) = 0 \eqsp.
$
	\end{hypAN}
\end{definition}
Here \AN{ass:AN:CLT} corresponds to a CLT and implies straightforwardly \C{def:consistency}. In addition, since \AN{def:tightness:island} implies immediately \C{def:consistency:max} we may conclude that asymptotic normality is stronger than consistency. Assumptions~\AN[ass:AN:squared:islands]{def:tightness:island} guarantee the existence of asymptotic variance; see \autoref{rem:AN:conditions} for further comments.


\section{Main results}
\label{sec:main:results}
\subsection{Operations on weighted archipelagos}
In the following we let $\disc(\{ a(i) \}_{i = 1}^M)$ denote the categorical probability distribution induced by a set $\{ a(i) \}_{i = 1}^M$ of positive (possibly unnormalized) numbers; thus, writing $V \sim \disc(\{ a(i) \}_{i = 1}^M)$ means that the random variable $V$ takes the value $i \in \intvect{1}{M}$ with probability $a(i) / \sum_{i' = 1}^M a(i')$.

\subsubsection{Selection on the island level}
The first operation, described in Algorithm~\ref{alg:island:selection}, is referred to as \emph{multinomial selection on the island level} (SIL). This operation consists in converting an archipelago $\{(\iswgt{i}, \{(\epart{i}{j}, \inwgt{i}{j}) \}_{j = 1}^{\Nin}) \}_{i = 1}^{\Nis}$ targeting some distribution $\targ$ into an archipelago $\{Ê(1, \{Ê(\eparttd{i}{j}, \inwgttd{i}{j}) \}_{j = 1}^{\Nin}) \}_{i = 1}^{\Nis}$ with uniform island weights targeting the same distribution $\targ.$ This step allows islands with small/large weights to be eliminated/duplicated, respectively. More precisely, a new family of islands is generated from the existing ones by selecting, conditionally independently given the input archipelago, new islands according to probabilities proportional to the island weights $\{\iswgt{i}\}_{i = 1}^{\Nis}$. After this, the weights and the particles of the selected islands are copied deterministically (which of course implies that the particle weights of the new archipelago are bounded by the same constant $\inwgtbd$ as the ancestor archipelago).

\begin{algorithm}[H] \label{alg:island:selection}
	\KwData{$\{Ê(\iswgt{i}, \{Ê(\epart{i}{j}, \inwgt{i}{j}) \}_{j = 1}^{\Nin}) \}_{i = 1}^{\Nis}$}
	\KwResult{$\{Ê(1, \{Ê(\eparttd{i}{j}, \inwgttd{i}{j}) \}_{j = 1}^{\Nin}) \}_{i = 1}^{\Nis}$}
	\For{$i \gets 1$ \KwTo $\Nis$}{
	draw $\isind{i} \sim \disc(\{ \iswgt{i'} \}_{i' = 1}^{\Nis})$\;
		\For{$j \gets 1$ \KwTo $\Nin$}{
			set $\eparttd{i}{j} \gets \epart{{\isind{i}}}{j}$\;
			set $\inwgttd{i}{j} \gets \inwgt{{\isind{i}}}{j}$\;
		}
	}
	\caption{Multinomial selection on the island level (SIL)}
\end{algorithm}
\bigskip
In the following we will abbreviate Algorithm~\autoref{alg:island:selection} by writing
$$
    \mbox{``}\{Ê(1, \{Ê(\eparttd{i}{j}, \inwgttd{i}{j}) \}_{j = 1}^{\Nin}) \}_{i = 1}^{\Nis} \gets \SIL \left( \{Ê(\iswgt{i}, \{Ê(\epart{i}{j}, \inwgt{i}{j}) \}_{j = 1}^{\Nin}) \}_{i = 1}^{\Nis} \right) \mbox{''} \eqsp.
$$

The following theorems state conditions under which SIL preserves consistency, exponential deviation, and asymptotic normality.
The input and output in Algorithm~\autoref{alg:island:selection}
are respectively denoted by $\{(\iswgt{i}, \{(\epart{i}{j}, \inwgt{i}{j}) \}_{j = 1}^{\Nin}) \}_{i = 1}^{\Nis}$ and $\{(1, \{(\eparttd{i}{j}, \inwgttd{i}{j}) \}_{j = 1}^{\Nin}) \}_{i = 1}^{\Nis}$ and all proofs are found in \autoref{sec:proofs}.
\begin{theorem} \label{thm:consistency:island:selection}
    Assume that $\{(\iswgt{i}, \{(\epart{i}{j}, \inwgt{i}{j}) \}_{j = 1}^{\Nin}) \}_{i = 1}^{\Nis}$ is consistent for $\targ$. Then also $\{Ê(1, \{Ê(\eparttd{i}{j}, \inwgttd{i}{j}) \}_{j = 1}^{\Nin}) \}_{i = 1}^{\Nis}$ is consistent for $\targ$.
\end{theorem}
\begin{theorem} \label{thm:deviation:inequalities:island:selection}
    Assume that $\{(\iswgt{i}, \{(\epart{i}{j}, \inwgt{i}{j})\}_{j = 1}^{\Nin}) \}_{i = 1}^{\Nis}$ satisfies exponential deviation for $(\targ, \Hlw, \{ c_{\ell} \}_{\ell = 1}^2)$. Then also $\{(1, \{(\eparttd{i}{j}, \inwgttd{i}{j}) \}_{j = 1}^{\Nin}) \}_{i = 1}^{\Nis}$ satisfies exponential deviation for $(\targ, \Hlw, \{ c_{\ell} \}_{\ell = 1}^2)$.
\end{theorem}


We impose the following assumption, guaranteeing that $\Nis$ grows only subexponentially fast with respect to $\Nin$.
\begin{hypR} \label{sub-geometric:growth}
For all $\beta > 0$, $\Nis \exp(- \beta \Nin) \rightarrow 0$ as $\N \rightarrow \infty$.
\end{hypR}

\begin{theorem} \label{thm:normality:island:selection}
    Assume \R{sub-geometric:growth} and that $\{( \iswgt{i}, \{(\epart{i}{j}, \inwgt{i}{j}) \}_{j = 1}^{\Nin} ) \}_{i = 1}^{\Nis}$ satisfies exponential deviation for $( \targ, \Hlw, \{c_{\ell}\}_{\ell=1}^2 )$ and is asymptotically normal for $(\targ, \asvar, \asvarANone, \{\condmeas{\ell} \}_{\ell = 1}^3)$. Then also $\{ (1, \{(\eparttd{i}{j}, \inwgttd{i}{j}) \}_{j = 1}^{\Nin}) \}_{i = 1}^{\Nis}$ is asymptotically normal for $(\targ, \asvartd, \asvarANone, \{ \targ, \condmeas{3}, \condmeas{3} \})$, where for all $h \in \bmf{\stfd}$,
    $$
    	\asvartd(h) = \asvar(h) + \asvarANone(h)
    $$
    (i.e. the SIL operation modifies only $\asvar$, $\condmeas{1}$, and $\condmeas{2}$).
\end{theorem}


\subsubsection{Selection on the individual level}
A second operation, described in Algorithm~\autoref{alg:individual:selection}, is referred to as \emph{multinomial selection on the individual level} (SiL), and consists in converting a weighted archipelago $\{Ê(\iswgt{i}, \{Ê(\epart{i}{j}, \inwgt{i}{j}) \}_{j = 1}^{\Nin}) \}_{i = 1}^{\Nis}$ targeting some distribution $\targ$ into an archipelago $\{Ê\iswgt{i}, \{Ê(\eparttd{i}{j}, 1) \}_{j = 1}^{\Nin} \}_{i = 1}^{\Nis}$ with uniform particle weights targeting the same distribution $\targ$. This step allows particles with large/small weights to be duplicated/eliminated, respectively. Note that the island weights remain unaffected.

\begin{algorithm}[H] \label{alg:individual:selection}
	\KwData{$\{Ê(\iswgt{i}, \{Ê(\epart{i}{j}, \inwgt{i}{j}) \}_{j = 1}^{\Nin}) \}_{i = 1}^{\Nis}$}
	\KwResult{$\{Ê(\iswgt{i}, \{Ê(\eparttd{i}{j}, 1) \}_{j = 1}^{\Nin} ) \}_{i = 1}^{\Nis}$}
	\For{$i \gets 1$ \KwTo $\Nis$}{
		\For{$j \gets 1$ \KwTo $\Nin$}{
			draw $\inind{i}{j} \sim \disc(\{ \inwgt{i}{j'} \}_{j' = 1}^{\Nin})$\;
			set $\eparttd{i}{j} \gets \epart{i}{{\inind{i}{j}}}$\;
		}
	}
	\caption{Multinomial selection on the individual level (SiL)}
\end{algorithm}
\bigskip
Trivially, the particle weights are bounded by $\inwgtbd = 1$ in this case. As for the SIL operation, we will express Algorithm~\autoref{alg:individual:selection} in a compact form by writing
$$
    \mbox{``}\{(\iswgt{i}, \{Ê(\eparttd{i}{j}, 1) \}_{j = 1}^{\Nin})\}_{i = 1}^{\Nis} \gets \SiL \left( \{Ê(\iswgt{i}, \{Ê(\epart{i}{j}, \inwgt{i}{j}) \}_{j = 1}^{\Nin}) \}_{i = 1}^{\Nis} \right) \mbox{''} \eqsp.
$$

The following theorems state conditions under which SiL preserves consistency, exponential deviation inequality, and asymptotic normality. Here,
$\{(\iswgt{i}, \{(\epart{i}{j}, \inwgt{i}{j}) \}_{j = 1}^{\Nin}) \}_{i = 1}^{\Nis}$ and
$\{Ê(\iswgt{i}, \{Ê(\eparttd{i}{j}, 1) \}_{j = 1}^{\Nin}) \}_{i = 1}^{\Nis}$ denote the input and output, respectively, of Algorithm~\ref{alg:individual:selection}.

\begin{theorem} \label{thm:consistency:individual:selection}
    Assume that $\{(\iswgt{i}, \{(\epart{i}{j}, \inwgt{i}{j}) \}_{j = 1}^{\Nin}) \}_{i = 1}^{\Nis}$ is consistent for $\targ$. Then also $\{Ê(\iswgt{i}, \{Ê(\eparttd{i}{j}, 1) \}_{j = 1}^{\Nin}) \}_{i = 1}^{\Nis}$ is consistent for $\targ$.
\end{theorem}
\begin{theorem} \label{thm:deviation:inequalities:individual:selection}
    Assume that $\{(\iswgt{i}, \{(\epart{i}{j}, \inwgt{i}{j})\}_{j = 1}^{\Nin}) \}_{i = 1}^{\Nis}$ satisfies exponential deviation for $(\targ, \Hlw, \{ c_{\ell} \}_{\ell = 1}^2)$. Then also $\{(\iswgt{i}, \{(\eparttd{i}{j}, 1) \}_{j = 1}^{\Nin})\}_{i = 1}^{\Nis}$ satisfies exponential deviation for \\
    $(\targ, 1, \{ \tilde{c}_{\ell} \}_{\ell = 1}^2)$, where $\Hconstd{1} = 4 (1 \vee \Hcons[1])$ and $\Hconstd{2} = (1 \wedge (\Hcons[2] \Hlw^2/2))/8$.
\end{theorem}
\begin{theorem} \label{thm:normality:individual:selection}
    Assume that $\{Ê(\iswgt{i}, \{Ê(\epart{i}{j}, \inwgt{i}{j}) \}_{j = 1}^{\Nin}) \}_{i = 1}^{\Nis}$ satisfies exponential deviation for $(\targ, \Hlw, \{c_{\ell}\}_{\ell=1}^2)$ and is asymptotically normal for $(\targ, \asvar, \asvarANone, \{\condmeas{\ell} \}_{\ell = 1}^3)$.\\ Then also $\{(\iswgt{i}, \{(\eparttd{i}{j}, 1) \}_{j = 1}^{\Nin}) \}_{i = 1}^{\Nis}$ is asymptotically normal for $(\targ, \asvartd, \asvartdANone, \{\condmeas{1}, \condmeas{1}, \targ \})$, where for all $h \in \bmf{\stfd}$,
    $$
        \begin{cases}
            \asvartd(h) = \asvar(h) + \condmeas{1}\{(h - \targ h)^2 \} \eqsp, \\
            \asvartdANone(h) =\asvarANone(h) + \targ \{(h - \targ h)^2 \} \eqsp.
        \end{cases}
    $$
\end{theorem}
Again, proofs are found in \autoref{sec:proofs}.

\subsubsection{Mutation}
\label{sec:mutation}
The last operation we consider is \emph{Mutation}, described in Algorithm~\ref{alg:mutation}. This operation converts, using importance sampling on the individual level, an archipelago $\{(\iswgt{i}, \{(\epart{i}{j}, \inwgt{i}{j}) \}_{j = 1}^{\Nin}) \}_{i = 1}^{\Nis}$ targeting $\targ \in \probmeas{\stfd}$ into another archipelago \\
$\{\iswgttd{i}, \{(\eparttd{i}{j}, \inwgttd{i}{j}) \}_{j = 1}^{\Nin} \}_{i = 1}^{\Nis}$ targeting some \emph{other} probability distribution $\targtd$, defined on another state space $(\stsptd, \stfdtd)$. The distribution $\targtd$ is related to
$\targ$ through the identity
\begin{equation} \label{eq:one-step-F-K}
\targtd h = \frac{\targ \op h}{\targ \op \1{\stsptd}} \quad (h \in \bmf{\stfdtd})
\end{equation}
where $\op : \stsp \times \stfdtd \rightarrow \rsetnonneg$ is a possibly unnormalized transition kernel. In the algorithm that follows, let $\prop : \stsp \times \stfdtd \rightarrow \rsetnonneg$ be a (normalized) transition kernel such that $\op(x, \cdot) \ll \prop(x, \cdot)$ for all $x \in \stsp$, and denote the corresponding Radon-Nikodym derivatives by
$$
	\der(x, \tilde{x}) \eqdef \frac{\rmd \op(x, \cdot)}{\rmd \prop(x, \cdot)}(\tilde{x}) \quad ((x, \tilde{x}) \in \stsp \times \stsptd) \eqsp.
$$
In the sequel, we will refer to the mapping $\der$ as the \emph{importance weight function} and assume that $\der \in \bmf{\stfd \varotimes \stfdtd}$ and $\op \1{\stsptd} \in \bmf{\stfd}$.

\begin{algorithm}[H] \label{alg:mutation}
	\KwData{$\{Ê(\iswgt{i}, \{Ê(\epart{i}{j}, \inwgt{i}{j}) \}_{j = 1}^{\Nin}) \}_{i = 1}^{\Nis}$, $\op$, $\prop$}
	\KwResult{$\{Ê(\iswgttd{i}, \{Ê(\eparttd{i}{j}, \inwgttd{i}{j}) \}_{j = 1}^{\Nin}) \}_{i = 1}^{\Nis}$}
	\For{$i \gets 1$ \KwTo $\Nis$}{
		\For{$j \gets 1$ \KwTo $\Nin$}{
			draw $\eparttd{i}{j} \sim \prop(\epart{i}{j}, \cdot)$\;
			set $\inwgttd{i}{j} \gets \der(\epart{i}{j}, \eparttd{i}{j}) \inwgt{i}{j}$\;
		}
		set $\displaystyle \iswgttd{i} \gets \iswgt{i} \frac{\sum_{j' = 1}^{\Nin} \inwgttd{i}{j'}}{\sum_{j'' = 1}^{\Nin} \inwgt{i}{j''}}$\;
	}
	\caption{Mutation}
\end{algorithm}
\bigskip
As before, we will abbreviate Algorithm~\autoref{alg:mutation} by writing
\begin{multline*}
\mbox{``} \{Ê(\iswgttd{i}, \{Ê(\eparttd{i}{j}, \inwgttd{i}{j}) \}_{j = 1}^{\Nin}) \}_{i = 1}^{\Nis} \\Ê\gets \mutation{\op}{\prop}\left( \{Ê(\iswgt{i}, \{Ê(\epart{i}{j}, \inwgt{i}{j}) \}_{j = 1}^{\Nin}) \}_{i = 1}^{\Nis}, \prop \right) \mbox{''} \eqsp,
\end{multline*}
where the kernel $\op$ is included in the notation for the sake of completeness. Note that the Mutation operation forms indeed a proper weighted archipelago with $\inwgttdbd = \inwgtbd \supn{\der}$. In conformity with the SIL and SiL operations, the Mutation operation preserves consistency, exponential deviation, and asymptotic normality. This is established below, where $\{Ê(\iswgt{i}, \{Ê(\epart{i}{j}, \inwgt{i}{j}) \}_{j = 1}^{\Nin}) \}_{i = 1}^{\Nis}$Êand $\{Ê(\iswgttd{i}, \{Ê(\eparttd{i}{j}, \inwgttd{i}{j}) \}_{j = 1}^{\Nin}) \}_{i = 1}^{\Nis}$ denote consequently the input and output of Algorithm~\ref{alg:mutation}, respectively.
\begin{theorem} \label{thm:consistency:mutation}
    Assume that $\{(\iswgt{i}, \{(\epart{i}{j}, \inwgt{i}{j}) \}_{j = 1}^{\Nin}) \}_{i = 1}^{\Nis}$ is consistent for $\targ$. \\
    Then $\{(\iswgttd{i}, \{(\eparttd{i}{j}, \inwgttd{i}{j}) \}_{j = 1}^{\Nin}) \}_{i = 1}^{\Nis}$ is consistent for $\targtd$ defined in \eqref{eq:one-step-F-K}.
\end{theorem}
\begin{theorem} \label{thm:deviation:inequalities:mutation}
    Assume that $\{(\iswgt{i}, \{(\epart{i}{j}, \inwgt{i}{j})\}_{j = 1}^{\Nin}) \}_{i = 1}^{\Nis}$ satisfies exponential deviation for  $(\targ, \Hlw, \{ c_{\ell} \}_{\ell = 1}^2)$. Then $\{(\iswgttd{i}, \{(\eparttd{i}{j}, \inwgttd{i}{j}) \}_{j = 1}^{\Nin}) \}_{i = 1}^{\Nis}$ satisfies exponential deviation for $(\tilde{\targ}, \Hlwtd, \{ \tilde{c}_{\ell} \}_{\ell = 1}^2)$, where $\Hlwtd = \Hlw \times \targ \op \1{\stsptd}$,
    $\Hconstd{1} = 2 (2 \vee \Hcons[1])$, and
    $$
    	\Hconstd{2} = \frac{1}{2}\left(\frac{1}{\delta^2} \wedge \frac{\Hcons[2]}{2 \supn[txt]{\op \1{\stsptd}}^2} \right) \eqsp,
    $$
    with $\delta \eqdef \inwgttdbd + \inwgtbd \supn[txt]{\op \1{\stsptd}}$.
\end{theorem}

\begin{theorem} \label{thm:normality:mutation}
    Assume that $\{(\iswgt{i}, \{(\epart{i}{j}, \inwgt{i}{j}) \}_{j = 1}^{\Nin}) \}_{i = 1}^{\Nis}$ satisfies exponential deviation for $(\targ, \Hlw, \{c_{\ell} \}_{\ell = 1}^2)$ and is asymptotically normal for $(\targ, \asvar, \asvarANone, \{\condmeas{\ell} \}_{\ell = 1}^3)$. Then the mutated archipelago $\{(\iswgttd{i}, \{(\eparttd{i}{j}, \inwgttd{i}{j}) \}_{j = 1}^{\Nin}) \}_{i = 1}^{\Nis}$ is asymptotically normal for $(\targtd, \asvartd, \asvartdANone, \{\condmeastd{\ell} \}_{\ell = 1}^3)$, where $\targtd$ is defined in  \eqref{eq:one-step-F-K} and Êfor all $h \in \bmf{\stfdtd}$,
       $$
        \begin{cases}
            \asvartd(h) = \left( \asvar \{\op (h-\targtd h) \} + \condmeas{2} \prop  \{ \der^2(h - \targtd h)^2\} - \condmeas{2} \{Ê\op^2 (h - \targtd h) \} \right) \left/ (\targ \op \1{\stsptd})^2 \right. \eqsp, \\
            \asvartdANone(h) = \left( \asvarANone \{\op (h - \targtd h) \} + \condmeas{3} \prop \{Ê\der^2(h - \targtd h)^2 \} - \condmeas{3} \{Ê\op^2(h - \targtd h) \} \right) \left/ (\targ \op \1{\stsptd})^2 \right. \eqsp, \\
            \condmeastd{1} h = \condmeas{1} \op h / \targ \op \1{\stsptd}  \eqsp, \\
            \condmeastd{2} h = \condmeas{2} \prop(\der^2 h) / (\targ \op \1{\stsptd})^2 \eqsp, \\
            \condmeastd{3} h = \condmeas{3} \prop(\der^2 h) / (\targ \op \1{\stsptd})^2 \eqsp,
        \end{cases}
    $$
    (where $\op^2 h(x) \eqdef \{Ê\op h(x) \}^2$ and $\prop(\der^2 h)(x) \eqdef \int \der^2(x, x') h(x') \, \prop(x, \rmd x')$ for all $x \in \stsp$ and $h \in \bmf{\stfd}$).
\end{theorem}

\begin{remark}
Note that \autoref{thm:normality:individual:selection} and \autoref{thm:normality:mutation} hold true regardless of the intermutual rates by which $\Nis$ and $\Nin$ tend to infinity with $\N$. In particular, these results do not, on the contrary to \autoref{thm:normality:island:selection}, require the condition \R{sub-geometric:growth}. This is in line with what we expect, as the SiL and Mutation operations do not involve any island interaction.
\end{remark}

\begin{remark} \label{rem:AN:conditions}
As clear from the previous, the SiL, SIL, and Mutation operations modify a given archipelago by means of either resampling of islands, local, island-wise resampling of individuals or random mutation of all the individuals of the archipelago. Assumptions \AN[ass:AN:squared:islands]{ass:AN:mut-1} regulate the increase of asymptotic variance brought forth by subjecting the archipelago to each of these operations, respectively. Thus, when the archipelago is subjected to a given operation, only one of these conditions plays the active role for the propagation of the CLT in \AN{ass:AN:CLT}; however, since we want to be able to analyze arbitrary, possibly random (as in the {\DBAIS} algorithm in \autoref{sec:BASIL:algorithm}) compositions of the operations, we are required to keep a record of the incremental variances disengaged by each one. Still, the conditions \AN[ass:AN:squared:islands]{def:tightness:island} are nested intricately in the sense that for a given operation, one or several conditions play active roles for the propagation of another. In this way, the condition \AN{ass:AN:mut-2}, which does not regulate directly the increase of asymptotic variance for any of the operations, bridges the mutation and island resampling operations in the sense that it regulates the limit \AN{ass:AN:squared:islands} in the case of mutation.
\end{remark}


\section{Applications}
\label{sec:applications}
\subsection{Feynman-Kac models}
\label{sec:Feynman-Kac:models}

For a sequence of unnormalized transition kernels $\{ \op[n] \}_{n \in \nset}$ defined on some common measurable space $(\stsp, \stfd)$ and some probability distribution $\kac[0] \in \probmeas{\stfd}$, a sequence $\{ \kac[n] \}_{n \in \nset}$ of \emph{Feynman-Kac measures} is defined by
\begin{equation} \label{eq:FKmeasures}
    \kac[n] h \eqdef \frac{\unkac[n] h}{\unkac[n] \1{\stsp}} \eqsp, \quad n \in \nset, \ h \in \bmf{\stfd} \eqsp,
\end{equation}
where
$$
	\unkac[n] h \eqdef \idotsint h(x_n) \, \kac[0](\rmd x_0) \prod_{p = 0}^{n - 1} \op[p](x_p, \rmd x_{p+1}) \quad (h \in \bmf{\stfd}) \eqsp
$$
(with usual convention $\prod_{p = m}^n a_p = 1$ when $m > n$). We may express recursively the sequences of unnormalized and normalized Feynman-Kac measures as, for $h \in \bmf{\stfd}$ and $(m, n) \in \nset $ with $m \leq n$,
$$
    \unkac[n] h = \unkac[m] \op[m] \cdots \op[n - 1] h \quad \mbox{and} \quad \kac[n] h = \frac{\unkac[m] \op[m] \cdots \op[n - 1] h}{\unkac[m] \op[m] \cdots \op[n - 1] \1{\stsp}} = \frac{\kac[m] \op[m] \cdots \op[n - 1] h}{\kac[m] \op[m] \cdots \op[n - 1] \1{\stsp}} \eqsp,
$$
respectively, with the convention $\op[m] \cdots \op[\ell](x, h) = h(x)$ if $m > \ell$. In particular,
\begin{equation} \label{eq:F-K:recursion}
    \kac[n + 1] h = \frac{\kac[n] \op[n] h}{\kac[n] \op[n] \1{\stsp}} \quad (h \in \bmf{\stfd}, n \in \nset) \eqsp,
\end{equation}
which means that we may cast the model into the framework considered in  \autoref{sec:mutation}.
\begin{example} \label{ex:Feynman-Kac}
    A special instance of the previous framework is formed naturally by specifying, first, a sequence $\{ \Mk[n] \}_{n \in \nset}$ of normalized (Markov) transition kernels on $(\stsp, \stfd)$Ê with an associated initial distribution $\chi$ and, second, \emph{potential functions} $\{Ê\pot[n] \}_{n \in \nsetpos}$, where $\pot[n] : \stsp \rightarrow \rsetpos$ for all $n \in \nsetpos$, and letting $\op[n]h(x) \eqdef \Mk[n](\pot[n + 1] h)(x)$, $n \in \nsetpos$,  $x \in \stsp$, and $h \in \bmf{\stfd}$. In addition, $\kac[0] \eqdef \chi$. This setup covers a large variety of important models in probability and statistics, such as \emph{optimal filtering} in hidden Markov models (or \emph{state-space models}; see, e.g., \cite{cappe:moulines:ryden:2005}) and models for the analysis of \emph{rare events} \cite{delmoral:garnier:2005,cerou:delmoral:furon:guyader:2012}. We will return to this setting in \autoref{long-term:stability:compact:case}.
\end{example}
Using a Feynman-Kac model in practice is typically non-trivial as neither the distribution flow $\{ \unkac[n] \}_{n \in \nset}$Ê nor $\{ \kac[n] \}_{n \in \nset}$Ê can be computed in a closed-form in general (with the exception of the very specific cases of optimal filtering in linear state-space models, in which case the solution is provided by the \emph{Kalman filter}, or hidden Markov models with finite state space).

 \subsection{The double bootstrap algorithm with adaptive selection}
 \label{sec:BASIL:algorithm}
In this section, our aim is to form online a sequence of archipelagos targeting the Feynman-Kac flow $\{ \kac[n] \}_{n \in \nset}$ by using sequentially the operations described in \autoref{sec:main:results}. A special feature of the approach that we consider is that the SIL operation is not performed systematically at every iteration of the algorithm, but only when the island weights fail to satisfy some appropriately defined skewness criterion. In this way we avoid adding unnecessary variance to the estimator. More specifically, we will analyze an algorithm proposed in \cite[Algorithm~3]{verge:dubarry:delmoral:moulines:2013}, where SIL is executed on the basis of the so-called \emph{coefficient of variation} (CV; see \cite{kong:liu:wong:1994} and \cite{liu:2001}) given by
$\CV(\{Ê\iswgt{i} \}_{i = 1}^{\Nis})$, where
\begin{equation} \label{def:ESScriterion}
 \CV :  (\rsetpos)^{\Nis} \ni \{Êa(i) \}_{i = 1}^{\Nis} \mapsto \Nis \sum_{i=1}^{\Nis} \left(\frac{a(i)}{\sum_{i'=1}^{\Nis} a(i')}\right)^2 - 1 \eqsp.
\end{equation}
The CV is closely related to the \emph{efficient sample size} (ESS, proposed in \cite{liu:1996}), which is the criterion used in \cite{verge:dubarry:delmoral:moulines:2013}; nevertheless, since the ESS can be expressed as
$\Nis / [1 + \CV(\{Êa(i) \}_{i = 1}^{\Nis})]$, the two criteria are equivalent. Note that the CV is minimal (zero) when all island weights are perfectly equal and maximal ($\Nis - 1$) in the situation of maximal skewness, i.e., when the total mass of the system is carried by a single island (a situation which is however not possible in our framework, as we always assume the island weights to be strictly positive). More specifically, as long as the CV stays below a specified threshold $\CVthres > 0$, we let the $\Nis$ Êislands evolve without interaction according to mutation and selection on the individual level. However, when the island weights get too dispersed as measured by the CV criterion, the islands are rejuvenated by SIL. The scheme, referred to by us as the \emph{double bootstrap with adaptive selection on the island level} (\DBAIS), is described in Algorithm~\autoref{alg:ESS:Bootstrap}, where we have added the iteration index $p$ to the weighted archipelagos returned by the algorithm.


\begin{algorithm}[t] \label{alg:ESS:Bootstrap}
	\KwData{$\{Ê\prop[p] \}_{p = 0}^{n - 1}$, $\tau$}
	\KwResult{$\{( \iswgt[p]{i}, \{(\epart[p]{i}{j}, \inwgt[p]{i}{j}) \}_{j = 1}^{\Nin}) \}_{i = 1}^{\Nis}$, $p \in \intvect{0}{n}$}
	 \tcc{initialization}
	 \For{$i \gets 1$ \KwTo $\Nis$}{
          	\For{$j \gets 1$ \KwTo $\Nin$}{
                    	$\epart[0]{i}{j} \sim \targ[0]$\;
                    	$\inwgt[0]{i}{j} \gets 1$\;
        		}
		$\iswgt[0]{i} \gets 1$\;
        }
        $\{( \iswgt[1]{i}, \{(\epart[1]{i}{j}, \inwgt[1]{i}{j}) \}_{j = 1}^{\Nin}) \}_{i = 1}^{\Nis} \gets \mutation{\op[0]}{\prop[0]}\left( \{Ê(1, \{Ê(\epart[0]{i}{j}, 1) \}_{j = 1}^{\Nin}) \}_{i = 1}^{\Nis}, \prop[0] \right)$\;
        \tcc{main loop}
        \For{$p \gets 1$ \KwTo $n-1$}{
            \tcc{checking island weight skewness}
            \eIf{$\CV( \{Ê\iswgt[p]{i} \}_{i = 1}^{\Nis}) > \CVthres$}{
                \tcc{selection on the island level}
                $
                		\{Ê(\iswgttd[p]{i}, \{Ê(\eparttd[p]{i}{j}, \inwgttd[p]{i}{j}) \}_{j = 1}^{\Nin}) \}_{i = 1}^{\Nis} \gets \SIL\left( \{Ê(\iswgt[p]{i}, \{Ê(\epart[p]{i}{j}, \inwgt[p]{i}{j}) \}_{j = 1}^{\Nin}) \}_{i = 1}^{\Nis} \right)
                $\;
            }{
            	\tcc{no action}
                $
                		\{Ê(\iswgttd[p]{i}, \{Ê(\eparttd[p]{i}{j}, \inwgttd[p]{i}{j}) \}_{j = 1}^{\Nin}) \}_{i = 1}^{\Nis} \gets \{Ê(\iswgt[p]{i}, \{Ê(\epart[p]{i}{j}, \inwgt[p]{i}{j}) \}_{j = 1}^{\Nin}) \}_{i = 1}^{\Nis}
                $\;
            }
            \tcc{selection on the individual level}
            $
            		\{Ê(\iswgttd[p]{i}, \{Ê(\epartck[p]{i}{j}, 1) \}_{j = 1}^{\Nin}) \}_{i = 1}^{\Nis} \gets \SiL \left( \{Ê(\iswgttd[p]{i}, \{Ê(\eparttd[p]{i}{j}, \inwgttd[p]{i}{j}) \}_{j = 1}^{\Nin}) \}_{i = 1}^{\Nis} \right)
            $\;
            \tcc{mutation}
            $
            		\{Ê(\iswgt[p + 1]{i}, \{Ê(\epart[p + 1]{i}{j}, \inwgt[p + 1]{i}{j}) \}_{j = 1}^{\Nin} \}_{i = 1}^{\Nis} \gets \mutation{\op[p]}{\prop[p]}\left( \{Ê(\iswgttd[p]{i}, \{Ê(\epartck[p]{i}{j}, 1) \}_{j = 1}^{\Nin}) \}_{i = 1}^{\Nis}, \prop[p] \right)
            $\;
        }
        \caption{The {\DBAIS} algorithm}
\end{algorithm}

Using the theoretical results obtained in \autoref{sec:main:results} we may prove the following result, establishing that exponential deviation and asymptotic normality are preserved through one iteration of the {\DBAIS} algorithm. As a by product we obtain the incremental asymptotic variance caused by an iteration. Since focus is set on asymptotic normality, we provide recursive formulas describing precisely the evolution of the functionals and measures involved in \AN[ass:AN:CLT]{ass:AN:mut-2}, while leaving the derivation of the analogous formulas for the constants of the exponential deviation bound \D{def:deviation:inequality:individual} to the reader. The proof of this result provides a nice illustration of the efficiency by which the theoretical results obtained in \autoref{sec:main:results}, despite appearing somewhat involved at a first sight, can be applied for analyzing sequences of archipelagos produced by executing alternatingly the SIL, SiL, and Mutation operations in an arbitrary order.

\begin{theorem} \label{theo:bootstrap:within:ESS}
    Assume~\R{sub-geometric:growth} and that $\{( \iswgt[n]{i}, \{(\epart[n]{i}{j}, \inwgt[n]{i}{j}) \}_{j = 1}^{\Nin} ) \}_{i = 1}^{\Nis}$ satisfies exponential deviation and is asymptotically normal for $(\kac[n], \asvar[n], \asvarANone[n], \{ \condmeas[n]{\ell} \}_{\ell = 1}^3)$, $n \in \nsetpos$. Then the archipelago $\{( \iswgt[n + 1]{i}, \{(\epart[n + 1]{i}{j}, \inwgt[n + 1]{i}{j}) \}_{j = 1}^{\Nin} ) \}_{i = 1}^{\Nis}$ generated through one iteration of Algorithm~\autoref{alg:ESS:Bootstrap} is asymptotically normal for $(\kac[n + 1], \asvar[n + 1], \asvarANone[n + 1], \{ \condmeas[n + 1]{\ell} \}_{\ell = 1}^3)$, where $\kac[n + 1]$ is given by \eqref{eq:F-K:recursion} and for all $ h \in \bmf{\stfd}$,
    $$
        \begin{cases}
            \asvar[n + 1](h)
            = \displaystyle
            \frac{\asvar[n] \{ \op[n](h - \kac[n + 1] h) \}
            + \varepsilon_n \asvarANone[n] \{ \op[n](h - \kac[n + 1] h) \}}{(\kac[n] \op[n] \1{\stsp})^2} \\
            \displaystyle \hspace{29mm}  + \frac{\varepsilon_n \kac[n] \prop[n]\{Ê\der^2_n(h - \kac[n + 1] h)^2 \}
             + (1 - \varepsilon_n) \condmeas[n]{1} \prop[n]\{Ê\der^2_n(h - \kac[n + 1] h)^2 \}}{(\kac[n] \op[n] \1{\stsp})^2} \eqsp, \\
            \displaystyle \asvarANone[n + 1](h) = \frac{\asvarANone[n]\{Ê\op[n](h - \kac[n + 1] h) \} + \kac[n] \prop[n]\{Ê\der^2_n(h - \kac[n + 1] h)^2 \}}{(\kac[n] \op[n] \1{\stsp})^2} \eqsp, \\
            \condmeas[n + 1]{1} h = \displaystyle (1 - \varepsilon_n) \frac{\condmeas[n]{1} \op[n] h }{\kac[n] \op[n] \1{\stsp}} + \varepsilon_n \kac[n + 1] h \eqsp, \\
            \displaystyle \condmeas[n + 1]{2} h = (1 - \varepsilon_n) \frac{\condmeas[n]{1} \prop[n](\der_n^2 h)}{(\kac[n] \op[n] \1{\stsp})^2} + \varepsilon_n \condmeas[n + 1]{3} h \eqsp, \\
            \displaystyle \condmeas[n + 1]{3} h = \frac{\kac[n] \prop[n](\der_n^2 h)}{(\kac[n] \op[n] \1{\stsp})^2} \eqsp,
        \end{cases}
    $$
    where $\varepsilon_n \eqdef \1{\{Ê\condmeas[n]{1} \1{\stsp} > \CVthres + 1\}}$.
\end{theorem}

\begin{proof}
    First, note that since the input archipelago satisfies \AN{ass:AN:squared:iswgts}, it holds that
    $$
    	\CV( \{Ê\iswgt[n]{i} \}_{i = 1}^{\Nis} ) \plim \condmeas[n]{1} \1{\stsp} - 1 \eqsp,
    $$
    which implies
    $$
    	\1{\{\CV( \{Ê\iswgt[n]{i} \}_{i = 1}^{\Nis} ) > \CVthres \}} \plim[N] \varepsilon_n \eqsp,
    $$
    where $\varepsilon_n$ is defined in the statement of the theorem. Consequently, after the \textbf{if}-\textbf{else} statement in Algorithm~\autoref{alg:ESS:Bootstrap}, the resulting archipelago $\{Ê(\iswgttd[p]{i}, \{Ê(\eparttd[p]{i}{j}, \inwgttd[p]{i}{j}) \}_{j = 1}^{\Nin}) \}_{i = 1}^{\Nis}$ satisfies, by \autoref{thm:deviation:inequalities:island:selection} and \autoref{thm:normality:island:selection}, exponential deviation and asymptotic normality, the latter for
    $$
        \begin{cases}
            (\kac[n], \asvar[n], \asvarANone[n], \{ \condmeas[n]{\ell} \}_{\ell = 1}^3) & \mbox{if } \varepsilon_n = 0 \eqsp, \\
            (\kac[n], \asvar[n] + \asvarANone[n]  , \asvarANone[n], \kac[n], \condmeas[n]{3}, \condmeas[n]{3}) & \mbox{if } \varepsilon_n = 1 \eqsp.
        \end{cases}
    $$
    Thus, the archipelago $\{Ê(\iswgttd[p]{i}, \{Ê(\epartck[p]{i}{j}, 1) \}_{j = 1}^{\Nin}) \}_{i = 1}^{\Nis}$ obtained after additional SiL satisfies, by \autoref{thm:deviation:inequalities:individual:selection} and \autoref{thm:normality:individual:selection}, exponential deviation as well as asymptotic normality, the latter for
    $$
        \begin{cases}
            (\kac[n], \asvar[n](\cdot) + \condmeas[n]{1}\{Ê(\cdot - \kac[n] \cdot)^2 \}, \asvarANone[n](\cdot) + \kac[n]\{Ê(\cdot - \kac[n] \cdot)^2 \}, \condmeas[n]{1}, \condmeas[n]{1}, \kac[n]) & \mbox{if } \varepsilon_n = 0 \eqsp, \\
            (\kac[n], \asvar[n](\cdot) + \asvarANone[n](\cdot) + \kac[n]\{Ê(\cdot - \kac[n] \cdot)^2 \},  \asvarANone[n](\cdot) + \kac[n] \{Ê(\cdot - \kac[n] \cdot)^2 \}, \kac[n], \kac[n], \kac[n]) & \mbox{if } \varepsilon_n = 1 \eqsp.
        \end{cases}
    $$
    Finally, considering also the final Mutation operation in Algorithm~\autoref{alg:ESS:Bootstrap}, and propagating, for the two different values of $\varepsilon_n$, the quantities of the previous display through the updating formulas of \autoref{thm:normality:mutation}, establishes, together with \autoref{thm:deviation:inequalities:mutation}, the statement of the theorem.
\end{proof}

\begin{corollary} \label{coro:ESS:bootstrap}
    Assume \R{sub-geometric:growth}. Then all archipelagos $\{Ê(\iswgt[n]{i}, \{Ê(\epart[n]{i}{j}, \inwgt[n]{i}{j}) \}_{j = 1}^{\Nin}) \}_{i = 1}^{\Nis}$, $n \in \nset$, produced by the {\DBAIS} algorithm satisfies exponential deviation and asymptotic normality, where for $h \in \bmf{\stfd}$ and $n \in \nsetpos$,
    $$
    \begin{cases}
        \displaystyle \asvar[n](h) = \sum_{\ell = 0}^{n - 1} \left(1 + \sum_{p = \ell + 1}^{n - 1} \varepsilon_p \right)
        \dfrac{\targ[\ell] \prop[\ell] \{Ê\der_{\ell}^2 \op[\ell + 1] \cdots \op[n - 1](h - \targ[n] h)^2 \}}{(\targ[\ell] \op[\ell] \cdots \op[n - 1] \1{\stsp})^2} \eqsp, \\
        \displaystyle \asvarANone[n](h) \displaystyle = \sum_{\ell = 0}^{n - 1} \dfrac{\targ[\ell] \prop[\ell] \{ \der_{\ell}^2 \op[\ell + 1] \cdots \op[n - 1](h - \kac[n] h)^2 \}}{(\kac[\ell] \op[\ell] \cdots \op[n - 1] \1{\stsp})^2} \eqsp, \\
        \displaystyle \condmeas[n]{1} h \displaystyle = \kac[n] h
    \end{cases}
    $$
    (under the standard conventions that $\prod_{\ell = m}^n a_\ell = 1$, $\sum_{\ell = m}^n a_\ell = 0$, and $\op[m] \cdots \op[n] = \operatorname{id}$ if $m > n$), and $\{Ê\varepsilon_n \}_{n \in \nsetpos}$ is given in \autoref{theo:bootstrap:within:ESS}. In addition, $\condmeas[0]{1} = \kac[0]$ and
    $$
   	 \asvar[0](h) =  \asvarANone[0](h) = \kac[0] \{(h - \kac[0] h)^2\} \quad (h \in \bmf{\stfd}) \eqsp.
    $$
\end{corollary}

\begin{proof}
    The non-recursive expression above are verified using induction. More specifically, one assumes that the given expressions of $(\asvar[n], \asvarANone[n],  \condmeas[n]{1})$ hold true for some $n \in \nset$ (and for all $h \in \bmf{\stfd}$) and plug the same into the recursive expressions established in \autoref{theo:bootstrap:within:ESS} under repeated use of the identities
    \begin{multline*}
        \op[\ell] \cdots \op[n - 1]\{ \op[n](h - \kac[n + 1] h) - \kac[n] \op[n](h - \kac[n + 1] h) \} = \op[\ell] \cdots \op[n](h - \kac[n + 1] h) \\ (h \in \bmf{\stfd}, \ell \in \nset) \eqsp,
    \end{multline*}
    and
    \begin{equation*} \label{eq:factorize:Q:prod}
        \kac[\ell] \op[\ell] \cdots \op[n - 1] \1{\stsp} \times \kac[n] \op[n] \1{\stsp} = \kac[\ell] \op[\ell] \cdots \op[n] \1{\stsp} \quad (\ell \in \nset) \eqsp.
    \end{equation*}
We leave this to the reader. To verify the base case $n = 1$, note that the initial archipelago $\{(1, \{Ê(\epart[0]{i}{j}, 1) \}_{j = 1}^{\Nin}) \}_{i = 1}^{\Nis} \}$ is, by the standard CLT and law of large numbers of for independent random variables, asymptotically normal for $(\targ[0], \asvar[0], \asvar[0], \targ[0], \targ[0], \targ[0])$, where $\asvar[0](h) = \kac[0] \{(h - \kac[0] h)^2\}$, $h \in \bmf{\stfd}$, and satisfies, by Hoeffding's inequality, exponential deviation for $(\kac[0], 1, 2, 1/2)$. Now, by
    \autoref{thm:deviation:inequalities:mutation} and \autoref{thm:normality:mutation} also the weighted archipelago $\{Ê(\iswgt[1]{i}, \{Ê(\epart[1]{i}{j}, \inwgt[1]{i}{j}) \}_{j = 1}^{\Nin}) \}_{i = 1}^{\Nis}$, obtained by mutating the initial archipelago, satisfies exponential deviation and asymptotic normality for $\condmeas[1]{1} = \kac[1]$ and
    $$
    	\asvar[1](h) =  \asvarANone[1](h) = \dfrac{\kac[0] \prop[0] \{ \der_0^2 (h - \targ[1] h)^2 \}}{(\kac[0] \op[0] \1{\stsp})^2} \quad (h \in \bmf{\stfd}) \eqsp.
    $$
    Under the standard conventions, this is however in agreement with the formula in the statement of the theorem. This completes the proof.
\end{proof}

Of special interest is of course the special case where SIL is applied systematically at every iteration, corresponding to $\CVthres = 0$. This yields the standard {\DB} algorithm, in which case the asymptotic variance is given by the following corollary.

\begin{corollary} \label{coro:systematic:SIL}
    Assume \R{sub-geometric:growth}. Then all archipelagos $\{Ê(\iswgt[n]{i}, \{Ê(\epart[n]{i}{j}, \inwgt[n]{i}{j}) \}_{j = 1}^{\Nin}) \}_{i = 1}^{\Nis}$, $n \in \nset$, produced by the {\DB} algorithm satisfies exponential deviation and asymptotic normality, where for $h \in \bmf{\stfd}$ and $n \in \nsetpos$,
    \begin{equation} \label{eq:as:var:B2}
    	\asvar[n](h) = \sum_{\ell = 0}^{n - 1} (n - \ell)
            \dfrac{\targ[\ell] \prop[\ell] \{Ê\der_{\ell}^2 \op[\ell + 1] \cdots \op[n - 1](h - \targ[n] h)^2 \}}{(\targ[\ell] \op[\ell] \cdots \op[n - 1] \1{\stsp})^2} \eqsp,
    \end{equation}
    and
    $$
    	\asvar[0](h) = \kac[0] \{(h - \kac[0] h)^2\} \eqsp.
    $$
\end{corollary}

\begin{proof}
    The result is an immediate consequence of \autoref{coro:ESS:bootstrap}, as $\CVthres = 0$ implies that $\varepsilon_n = 1$ for all $n \in \nsetpos$.
\end{proof}

On the other hand, letting $\varepsilon_n = 0$ for all $n \in \nsetpos$ in \autoref{coro:ESS:bootstrap}, corresponding to the case where SIL is never applied, yields the variance
\begin{equation} \label{eq:as:var:bootstrap}
	\asvar[n](h) = \sum_{\ell = 0}^{n - 1} \dfrac{\targ[\ell] \prop[\ell] \{Ê\der_{\ell}^2 \op[\ell + 1] \cdots \op[n - 1](h - \targ[n] h)^2 \}}{(\targ[\ell] \op[\ell] \cdots \op[n - 1] \1{\stsp})^2} \eqsp,
\end{equation}
which we recognize as the well-known formula for the asymptotic variance of the standard SMC algorithm (more specifically, the \emph{sequential importance sampling with resampling}, SISR, algorithm). This is completely in line with our expectations, as such an algorithm would simply propagate $\Nis$ independent (non-interacting) islands, each island evolving as a standard SMC algorithm based on $\Nin$ particles.

\subsection{Long-term stability of the double bootstrap algorithm}
\label{long-term:stability:compact:case}
As a last part of our study we establish the long-term numerical stability of the {\DB} algorithm by providing a time uniform bound on the asymptotic variance of its output. Throughout this section we will, in the spirit of \autoref{ex:Feynman-Kac}, assume that each unnormalized transition kernel $\op[p]$, $p \in \nset$, can be decomposed into a normalized transition kernel $\Mk[p] : \stsp \times \stfd \rightarrow [0, 1]$Ê and a nonnegative potential potential function $\pot[p + 1] : \stsp \rightarrow \rset_+$, i.e., for all $h \in \bmf{\stfd}$ and $x \in \stsp$,
\begin{equation} \label{eq:def:Q:potential:setting}
	\op[p]h(x) = \Mk[p](\pot[p + 1] h)(x) \eqsp.
\end{equation}
In this setting, given a sequence $\{Ê\prop[p] \}_{p \in \nset}$ of proposal kernels such that $\Mk[p](x, \cdot) \ll \prop[p](x, \cdot)$ for all $x \in \stsp$ and $p \in \nset$,Ê the importance weight function is given by
$$
    \der_p(x, x') = \pot[p + 1](x') \frac{\rmd \Mk[p](x, \cdot)}{\rmd \prop[p](x, \cdot)} \quad (x, x') \in \stsp^2 \eqsp.
$$
\begin{remark}
    Instead of letting the Feynman-Kac distribution flow be generated by the unnormalized kernel \eqref{eq:def:Q:potential:setting}, one could, as in \cite{verge:dubarry:delmoral:moulines:2013}, consider an alternative model with a flow $\{ \kactd[p] \}_{p \in \nset}$ generated by
    \begin{equation} \label{eq:def:prediction:flow}
        \optd[p]h(x) = \pot[p](x) \Mk[p]h(x) \quad (h \in \bmf{\stfd}, x \in \stsp, p \in \nsetpos) \eqsp,
    \end{equation}
    with $\optd[0] = \Mk[0]$ and $\kactd[0] = \chi$. In \cite{delmoral:2004} the two models \eqref{eq:def:Q:potential:setting} and \eqref{eq:def:prediction:flow} are referred to as  \emph{updated} and \emph{prediction} Feynman-Kac models, respectively. For the prediction model, it is, in the case of the {\DB} algorithm, possible to achieve \emph{full adaptation} (borrowing the terminology of \cite{pitt:shephard:1999}) of the algorithm, i.e., to generate archipelagos $\{Ê(1, \{Ê(1, \epart[p]{i}{j}) \}_{j = 1}^{\Nin}) \}_{i = 1}^{\Nis}$, $p \in \nset$, with uniformly weighted islands and individuals targeting the distribution sequence of interest, by letting $\prop[p] = \Mk[p]$ for all $p \in \nset$ and decomposing the dynamics \eqref{eq:def:prediction:flow}Ê into the product
    \begin{equation} \label{eq:F-K:decomp}
    	\optd[p] = G_p \Mk[p] \eqsp,
    \end{equation}
    where $G_p h(x) = \pot[p](x) h(x)$, $(x, h) \in \stsp \times \bmf{\stfd}$, is the \emph{Boltzmann multiplicative operator} associated with the potential $\pot[p]$. Now \eqref{eq:F-K:decomp} allows also the Feynman-Kac transition according to $\optd[p]$ to be decomposed into \emph{two} subsequent Feynman-Kac sub-transitions, the first according to $G_p$ and the other according to $\Mk[p]$. The former corresponds to the Mutation operation
    \begin{equation} \label{eq:weighing:operation:prediction}
    	\{Ê(\iswgtbar[p]{i}, \{Ê(\inwgtbar[p]{i}{j}, \epart[p]{i}{j}) \}_{i = 1}^{\Nin}) \}_{i = 1}^{\Nis} \leftarrow \mutation{G_p}{} \left( \{Ê(1, \{Ê(1, \epart[p]{i}{j}) \}_{j = 1}^{\Nin}) \}_{i = 1}^{\Nis}, \operatorname{id} \right) \eqsp,
    \end{equation}
    which simply assigns each particle and island the weights $\inwgtbar{i}{j} = \pot[p](\epart[p]{i}{j})$ and $\iswgtbar[p]{i} = \sum_{j = 1}^{\Nin} \pot[p](\epart[p]{i}{j}) / \Nin$, respectively (where we assumed that we start with uniformly weighted islands and individuals). After this \emph{weighing operation}, the output \eqref{eq:weighing:operation:prediction} Êis, in accordance with Algorithm~\autoref{alg:ESS:Bootstrap} (with $\CVthres = 0$), subjected to the SIL and SiL operations followed by the Mutation operation
    $$
    	\{Ê(1, \{Ê(1, \epart[p + 1]{i}{j}) \}_{i = 1}^{\Nin}) \}_{i = 1}^{\Nis} \leftarrow \mutation{\Mk[p]}{} \left( \{Ê(1, \{Ê(1, \eparttd[p]{i}{j}) \}_{j = 1}^{\Nin}) \}_{i = 1}^{\Nis}, \Mk[p] \right) \eqsp,
    $$
    yielding an archipelago with perfectly uniform island and individual weights approximating the Feynman-Kac distribution $\kactd[p + 1]$ at the next time point. Also this algorithm may be analyzed straightforwardly using our results, and carrying through this analysis retrieves exactly the variance expression obtained in \cite[Equation~43]{verge:dubarry:delmoral:moulines:2013}. We leave this as an exercise to the interested reader.

    The previous way of obtaining an archipelago with uniformly weighted islands and individuals approximating the prediction Feynman-Kac distribution flow can be viewed as a special instance of a general \emph{auxiliary double bootstrap algorithm} (extending the so-called \emph{auxiliary particle filter} proposed in \cite{pitt:shephard:1999}) based on the decomposition
    $$
    	\op[p] = T_p \check{Q}_p \eqsp,
    $$
    where $T_p h(x) = t_p(x) h(x)$, $(x, h) \in \stsp \times \bmf{\stfd}$, is a Boltzmann multiplicative operator associated with some positive auxiliary importance weight function $t_p \in \bmf{\stfd}$, and
    $$
    	\check{Q}_p(x, h) \eqdef t_p^{-1}(x) \op[p] h(x) \quad (x \in \stsp, h \in \bmf{\stfd}) \eqsp.
    $$
    In analogy with the previous, we may thus construct an alternative algorithm approximating $\{ \kac[p] \}_{p \in \nset}$ by furnishing the main loop of the {\DB} algorithm with a prefatory weighing operation
    \begin{multline} \label{eq:auxiliary:wgt:step}
    \{Ê(\iswgtbar[p]{i}, \{Ê(\inwgtbar[p]{i}{j}, \epart[p]{i}{j}) \}_{i = 1}^{\Nin}) \}_{i = 1}^{\Nis} \\Ê
    \leftarrow \mutation{T_p}{} \left( \{Ê(\iswgttd[p]{i}, \{Ê(\inwgttd[p]{i}{j}, \epart[p]{i}{j}) \}_{i = 1}^{\Nin}) \}_{i = 1}^{\Nis}, \operatorname{id} \right) \eqsp,
    \end{multline}
    and, after intermediate SIL and SiL operations, a terminating Mutation operation
    \begin{multline*}
    	\{Ê(\iswgt[p + 1]{i}, \{Ê(\inwgt[p + 1]{i}{j}, \epart[p + 1]{i}{j}) \}_{i = 1}^{\Nin}) \}_{i = 1}^{\Nis} \\Ê
    	\leftarrow \mutation{\check{Q}_p}{} \left( \{Ê(1, \{Ê(1, \eparttd[p]{i}{j}) \}_{j = 1}^{\Nin}) \}_{i = 1}^{\Nis}, \prop[p] \right) \eqsp,
    \end{multline*}
    where, consequently, the all weights are given by the importance weight function
    $$
    	\check{w}_p(x, x') = t_p^{-1}(x) \frac{\rmd \op[p](x, \cdot)}{\rmd \prop[p](x, \cdot)}(x') \quad ((x, x') \in \stsp^2) \eqsp.
    $$
    Thus, choosing $t_p(x)$ as some prediction of the value of the derivative $\rmd \op[p](x, \cdot) / \rmd \prop[p](x, \cdot)$ in the support of $\prop[p](x, \cdot)$ yields close to uniformly weighted islands and individuals (i.e., a close to fully adapted algorithm); for instance, following \cite{pitt:shephard:1999}, a possible design is $t_p(x) = \rmd \op[p](x, \cdot) / \rmd \prop[p](x, \cdot)(\prop[p] \operatorname{id}(x))$. Of course, also this algorithm can be analyzed easily using our results (we refer to \cite{douc:moulines:olsson:2008} for such an analysis of the standard auxiliary particle filter).
\end{remark}

\subsection{Time uniform convergence under the strong mixing assumption}

When studying the numerical stability of the {\DB} algorithm we will first work under the following \emph{strong mixing condition}.

\begin{hypMG} \label{ass:strong:mixing:condition}
    \begin{enumerate}[(i)]
        \item There exist constants $0 < \lbmarkov < \ubmarkov  <\infty$ and $\minmeas \in \probmeas{\stfd}$ such that for all $p \in \nset$, $x \in \stsp$, and $\set{A} \in \stfd$,
        $$
        		\lbmarkov \minmeas(\set{A}) \leq \Mk[p](x, \set{A}) \leq \ubmarkov \minmeas(\set{A}) \eqsp.
        $$
        \item It holds that $w_+ \eqdef \sup_{p \in \nset} \supn{\der_p} < \infty$.
        \item It holds that $\lbop \eqdef \inf_{(p, x) \in \nset \times \stsp} \op[p] \1{\stsp}(x) > 0$.
    \end{enumerate}
\end{hypMG}

The assumption \MG{ass:strong:mixing:condition}(i), implying that each $\Mk[p]$ allows the whole state space $\stsp$ as a \emph{$1$-small set}, is rather restrictive and  requires typically the state space $\stsp$ to be a compact set. Still, it plays a vital role in the literature of SMC analysis \cite[see, e.g.,][]{delmoral:guionnet:2001,delmoral:2004,cappe:moulines:ryden:2005,douc:moulines:olsson:2008,olsson:cappe:douc:moulines:2006,douc:moulines:stoffer:2014}. On the other hand, the weaker assumption \MG{ass:strong:mixing:condition}(ii) is satisfied for most applications and \MG{ass:strong:mixing:condition}(iii) does not require the potential functions to be uniformly bounded from below; the latter is a condition that appears frequently in the literature. Under \MG{ass:strong:mixing:condition}, denote
\begin{equation} \label{eq:def:rho}
\rho \eqdef 1 - \frac{\lbmarkov}{\ubmarkov} \eqsp;
\end{equation}
then the previous assumptions allow the following explicit time uniform bound to be derived.

\begin{corollary} \label{cor:bound:of:the:variance}
    Suppose \MG{ass:strong:mixing:condition}. Then the sequence of asymptotic variances of the {\DB} algorithm (see \autoref{coro:systematic:SIL}) satisfies, for all $n \in \nset$ and $h \in \bmf{\stfd}$,
    \begin{equation} \label{def:time:uniform:bound}
        \asvar[n](h) \leq w_+ \dfrac{\osc[2]{h}}{(1 - \rho)^2 (1 - \rho^2)^2 \lbop} \eqsp,
    \end{equation}
    where $\rho$ is defined in \eqref{eq:def:rho}.
\end{corollary}

The proof is found in \autoref{proof_bound:var}. In addition, by comparing the formulas of the asymptotic variances of the {\DBAIS} and {\DB} algorithms provided by \autoref{coro:ESS:bootstrap} and \autoref{coro:systematic:SIL}, respectively, we conclude that at each time step, the asymptotic variance of the latter algorithm is always bounded by the that of the former (as the indicator variables $\{ \varepsilon_p \}_{p \in \nsetpos}$, determining the asymptotic island selection schedule of the {\DBAIS} algorithm, are either zero or one for all $p$). The following corollary is hence immediate.
\begin{corollary}
Suppose \MG{ass:strong:mixing:condition}. Then also the asymptotic variances of the {\DBAIS} algorithm (see \autoref{coro:ESS:bootstrap}) satisfy the bound \eqref{def:time:uniform:bound}.
\end{corollary}

\subsection{Time uniform convergence under a local Doeblin condition}
The explicitness and simplicity of the variance bound in \autoref{cor:bound:of:the:variance} are obtained at the cost of restrictive model assumptions that are rarely satisfied in real-world applications. Thus, in this section we will discuss how the assumptions of \MG{ass:strong:mixing:condition} can be lightened considerably and turned into easily verifiable conditions, satisfied for many models of interest, by considering assumptions under which the asymptotic variance is \emph{stochastically bounded} (tight) rather than bounded by a deterministic constant. Since the asymptotic variance \eqref{eq:as:var:B2} of the {\DB} algorithm differs only from that of the SISR algorithm (see \eqref{eq:as:var:bootstrap}) by the factors $n - \ell$, the results obtained in this section will rely heavily on similar results obtained in \cite{douc:moulines:olsson:2014} for the standard bootstrap particle filter. For this purpose, assume that each potential function depends on time through some random parameter only, i.e., for all $p \in \nsetpos$, $\pot[p] = \pot \langle \randpar{p} \rangle$, where $\{ \randpar{p} \}_{p \in \nset}$ is some stochastic process taking values in some state space $(\Zset, \Zfd)$ and $\pot \langle z \rangle \in \bmf{\stfd}$ for all $z \in \Zset$. Moreover, we assume that the normalized transition kernels of the model are time homogeneous, i.e., $\Mk[p] = \Mk$ for all $p \in \nset$, and that Mutation is based on the underlying dynamics of the model, i.e., $\prop[p] = \prop = \Mk$, and, consequently, $\der_p(x, x') = \der_p \langle \randpar{p + 1} \rangle (x, x') = \pot \langle \randpar{p + 1} \rangle(x')$ for all $p \in \nset$ and $(x, x') \in \stsp^2$. Thus, in this case the model generates a parameter dependent Feynman-Kac flow $\{Ê\kac[p] \langle \randpar{0:p} \rangle \}_{p \in \nset}$. (For instance, in the case of a hidden Markov model, the sequence $\{ \randpar{p} \}_{p \in \nset}$ plays the role of noisy observations of some Markov chain (the \emph{state process}) $\{ÊX_p \}_{p \in \nset}$Ê with transition kernel $\Mk$ on $(\stsp, \stfd)$. Conditionally on the state process, the observations are assumed to be independent and such that the conditional density $\Zset \ni z \mapsto \pot \langle z \rangle(x)$ of each $\randpar{p}$ depends on the corresponding state $X_p = x \in \stsp$ only. In this important framework, $\kac[p] \langle \randpar{0:p} \rangle$ is the so-called \emph{filter distribution} at time $p$, i.e., the conditional distribution of the latent state $X_p$ given the observations $\randpar{0:p}$.) In this case, the asymptotic variance $\asvar[n](h)$ of the {\DB} algorithm is a function of the random vector $\randpar{0:n}$, and we write $\asvar[n] \langle \randpar{0:n} \rangle(h)$ to emphasize this fact. We will replace the condition  \MG{ass:strong:mixing:condition}(i) by a considerably weaker condition of the following type.

\begin{definition} \label{def:local:Doeblin}
    A set $\set{C} \in \stfd$ is \emph{local Doeblin} with respect to $\Mk$ if there is $\minmeas_{\set{C}} \in \probmeas{\stfd}$ with $\minmeas_{\set{C}}(\set{C}) = 1$ and constants $0 < \sigma_{\set{C}}^- < \sigma_{\set{C}}^+$ such that for all $x \in \set{C}$Ê and $\set{A} \in \stfd$,
     $$
            	\sigma_{\set{C}}^- \minmeas_{\set{C}}(\set{A}) \leq \Mk(x, \set{A} \cap \set{C}) \leq  \sigma_{\set{C}}^+ \minmeas_{\set{C}}(\set{A}) \eqsp.
    $$
\end{definition}

Now, impose the following assumption.

\begin{hypLD} \label{ass:local:Doeblin}
    The process $\{ \randpar{p} \}_{p \in \nset}$ is strictly stationary and ergodic. Moreover, there exists a set $\set{K} \in \Zfd$ such that the following holds.
    \begin{enumerate}[(i)]
        \item $\prob\left( \randpar{0} \in \set{K} \right) > 2/3$.
        \item For all $\varepsilon > 0$ there exists a local Doeblin set $\set{C}$ such that for all $z \in \Zset$,
        $$
            \sup_{x \in \set{C}^c} \pot \langle z \rangle(x) \leq \varepsilon \supn{\pot \langle z \rangle} < \infty \eqsp.
        $$
        \item There exists a set $\set{D} \in \stfd$ such that $\inf_{x \in \set{D}} \Mk(x, \set{D}) > 0$ and
        $$
        	\E \left[Ê\ln^- \inf_{x \in \set{D}}  \pot \langle \randpar{0} \rangle(x) \right] < \infty \eqsp.
        $$
    \end{enumerate}
\end{hypLD}
In \LD{ass:local:Doeblin}(iii), $\ln^-$ denotes the negative part of the natural logarithm. The condition \LD{ass:local:Doeblin} can be checked easily for a large variety of models; see \cite[Section~4]{douc:moulines:olsson:2014} for examples.
\begin{remark}
    The condition \LD{ass:local:Doeblin} can be weakened further by requiring the local Doeblin condition to hold only for some iterate $\op \langle z_1 \rangle \cdots \op \langle z_r \rangle$, $z_{1:r} \in \Zset^r$, with a minorizing measure $\minmeas_{\set{C}}$ and constants $\sigma_{\set{C}}^-$,Ê $\sigma_{\set{C}}^+$ possibly depending on the block $z_{1:r}$; we refer to \cite{douc:moulines:olsson:2014} for details. In this paper we have however chosen to state the most basic version of the condition (corresponding to $r = 1$) for simplicity.
\end{remark}

Under~\LD{ass:local:Doeblin}, define $\probmeas{\stfd, \set{D}} \subset \probmeas{\stfd}$ as the set of all $\Xinit \in \probmeas{\stfd}$ for which there exists $\set{D}' \in \stfd$ such that (i) $\inf_{x \in \set{D}'} \Mk(x, \set{D}) > 0$, (ii) $\E[ \ln^- \inf_{x \in \set{D}'}  \pot \langle \randpar{0} \rangle(x)] < \infty$, and (iii) $\Xinit(\set{D}') > 0$. Then the following holds true.
\begin{corollary} \label{cor:non-compact:case}
    Assume {\LD{ass:local:Doeblin}} and suppose in addition that $\kac[0] \in \probmeas{\stfd, \set{D}}$. Then for all $h \in \bmf{\stfd}$, the sequence $\{ \asvar[n] \langle \randpar{0:n} \rangle(h) \}_{n \in \nset}$ of asymptotic variances of the output of the {\DB} algorithm is tight, i.e., it satisfies
    $$
        \lim_{\lambda \rightarrow \infty} \sup_{n \in \nset} \prob \left( \asvar[n] \langle \randpar{0:n} \rangle(h) \geq \lambda \right) = 0 \eqsp.
    $$
\end{corollary}

The previous result is obtained by inspecting the proof of \cite[Theorem~11]{douc:moulines:olsson:2014}; see \autoref{sec:proof:non-compact} for some details.


\section{Proofs}
\label{sec:proofs}
As mentioned in the introduction, our consistency and asymptotic normality proofs rely on limit theorems for triangular arrays developed in \cite{douc:moulines:2008}. More specifically, \cite{douc:moulines:2008} developed \autoref{thm:DM:A-1} and \autoref{thm:DM:A-3}, which are re-stated in \autoref{appenbdix:technical:lemmas} for completeness, for the purpose of proving consistency and asymptotic normality for weighted samples of \emph{particles}, whereas we use the same results for establishing these properties for weighted samples of particle \emph{islands}. However, whereas the elements of the arrays considered in \cite{douc:moulines:2008} correspond mainly to single particles, the arrays defined by us will be considerably more complex with each element being generally itself a weighted average of particles associated with an island. In this section we will use repeatedly the same notation $\{ \arr{i} \}_{i = 1}^{\Nis}$ to denote triangular arrays (in the sense of \autoref{thm:DM:A-1} and \autoref{thm:DM:A-3}), even though the roles of these arrays change throughout the proofs.


\subsection{Proof of \autoref{thm:consistency:island:selection}}
\label{sec:proof:consistency:island:selection}

We apply \autoref{thm:DM:A-1}. For this purpose, define the triangular array and filtration
\begin{equation} \label{eq:def:filt}
	\begin{split}
		\arr{i} &\eqdef \sum_{j = 1}^{\Nin} \frac{\inwgt{{\isind{i}}}{j}}{\Nis \sum_{j' = 1}^{\Nin} \inwgt{{\isind{i}}}{j'}} h(\epart{{\isind{i}}}{j}) \quad (i \in \intvect{1}{\Nis[\N]}, \N \in \nsetpos) \eqsp, \\
		\filt &\eqdef \sigma\left(\{Ê(\iswgt{i}, \{Ê(\epart{i}{j}, \inwgt{i}{j}) \}_{j = 1}^{\Nin}) \}_{i = 1}^{\Nis} \right) \quad (\N \in \nsetpos) \eqsp,
	\end{split}
\end{equation}
respectively. Now, since the island indices $\{ \isind{i} \}_{i = 1}^{\Nis}$ are, conditionally on $\filt$, i.i.d. with common distribution $\disc(\{ \iswgt{i'} \}_{i' = 1}^{\Nis})$ it holds, as the ancestor sample is assumed to be consistent,
\begin{align} \label{eq:unbiasedness:island:selection}
	\sum_{i = 1}^{\Nis} \cexp{\arr{i}}{\filt} &= \Nis \cexp{\arr{1}}{\filt} \\
	&= \sum_{i = 1}^{\Nis} \frac{\iswgt{i}}{\sum_{i' = 1}^{\Nis} \iswgt{i'}} \sum_{j = 1}^{\Nin} \frac{\inwgt{i}{j}}{\sum_{j' = 1}^{\Nin} \inwgt{i}{j'}} h(\epart{i}{j}) \plim \targ h \eqsp. \nonumber
\end{align}
Thus, since $|\arr{i}| \leq \supn{h} / \Nis < \infty$ for all $i \in \intvect{1}{\Nis}$, it is enough to check the conditions \A{ass:DM:cons:tightness} and \A{ass:DM:cons:Lindeberg} in \autoref{thm:DM:A-1}. The tightness condition \A{ass:DM:cons:tightness} is straightforwardly satisfied as sequences that converge in probability are tight.
Moreover, to check \A{ass:DM:cons:Lindeberg} we may apply \autoref{lem:tightness} with $\mbd = \supn{h} / \Nis$ and $\bd = \bdconst = 0$. Thus, the limits, in probability as $\N \rightarrow \infty$, of the series $\sum_{i = 1}^{\Nis} \arr{i}$ and $\sum_{i = 1}^{\Nis} \cexp{\arr{i}}{\filt}$ coincide, which completes the proof.

\subsection{Proof of \autoref{thm:deviation:inequalities:island:selection}}
\label{sec:proof:deviation:inequalities:island:selection}

Trivially, since $\{ \isind{i}Ê\}_{i = 1}^{\Nis} \subset \intvect{1}{\Nis}$ it holds for all $\varepsilon > 0$,
\begin{multline*}
	\prob \left( \max_{i \in \intvect{1}{\Nis}} \left| \frac{1}{\Nin} \sum_{j = 1}^{\Nin} \inwgt{\isind{i}}{j} h(\epart{\isind{i}}{j}) - \Hlw \times \targ h \right| \geq \varepsilon \right) \\
	\leq \prob \left( \max_{i \in \intvect{1}{\Nis}} \left| \frac{1}{\Nin} \sum_{j = 1}^{\Nin} \inwgt{i}{j} h(\epart{i}{j}) - \Hlw \times \targ h \right| \geq \varepsilon \right) \eqsp,
\end{multline*}
where the right hand side has an exponential bound by assumption. This completes the proof.

\subsection{Proof of \autoref{thm:normality:island:selection}}
\label{sec:proof:normality:island:selection}

In order to check \AN{ass:AN:CLT} using \autoref{thm:DM:A-3}, define the array
\begin{equation} \label{eq:island:selection:array:AN}
	\arr{i} \eqdef \sqrt{\frac{\Nin}{\Nis}} \sum_{j = 1}^{\Nin} \frac{\inwgt{{\isind{i}}}{j}}{\sum_{j' = 1}^{\Nin} \inwgt{{\isind{i}}}{j'}} \{ h(\epart{{\isind{i}}}{j}) -  \targ h \}\quad (i \in \intvect{1}{\Nis}, \N \in \nsetpos) \eqsp,
\end{equation}
and let $\{\filt \}_{\N \in \nsetpos}$ be the filtration \eqref{eq:def:filt}. We first note that $\cexp[txt]{\arr[2]{i}}{\filt} \leq 4 \Nin \supn{h}^2 / \Nis < \infty$ for all $i \in \intvect{1}{\Nis}$. Along the lines of \eqref{eq:unbiasedness:island:selection} (note however that the definition \eqref{eq:def:filt} of the triangular array in \eqref{eq:unbiasedness:island:selection} differs slightly from that of the array \eqref{eq:island:selection:array:AN} considered here),
\begin{multline} \label{eq:decomposition:CLT:island:selection}
	\sum_{i = 1}^{\Nis} \arr{i} = \sum_{i = 1}^{\Nis} \left\{ \arr{i} - \cexp{\arr{i}}{\filt} \right\} \\
	+ \sqrt{N} \sum_{i = 1}^{\Nis} \frac{\iswgt{i}}{\sum_{i' = 1}^{\Nis} \iswgt{i'}} \sum_{j = 1}^{\Nin} \frac{\inwgt{i}{j}}{\sum_{j' = 1}^{\Nin} \inwgt{i}{j'}} \{h(\epart{i}{j}) - \targ h \} \eqsp.
\end{multline}
By assumption, the second term on the right hand side of \eqref{eq:decomposition:CLT:island:selection} converges in distribution to a Gaussian random variable with zero mean and variance $\asvar(h)$. To treat the first term using \autoref{thm:DM:A-3} we first consider
\begin{multline} \label{eq:island:selection:second:moment}
    \sum_{i = 1}^{\Nis} \cexp{\arr[2]{i}}{\filt}
    = \Nin \sum_{i = 1}^{\Nis} \frac{\iswgt{i}}{\sum_{i' = 1}^{\Nis} \iswgt{i'}} \left( \sum_{j = 1}^{\Nin} \frac{\inwgt{i}{j}}{\sum_{j' = 1}^{\Nin} \inwgt{i}{j'}} \{ h(\epart{i}{j}) - \targ h \} \right)^2 \\
    \plim \asvarANone(h) \eqsp,
\end{multline}
where the limit holds as the ancestor archipelago is assumed to satisfy \AN{ass:AN:squared:islands}. In addition,
\begin{multline*}
    \sum_{i = 1}^{\Nis} \E^2 \left[Ê\arr{i} \mid \filt \right]
    = \frac{1}{\Nis} \left(\sqrt{N} \sum_{i = 1}^{\Nis} \frac{\iswgt{i}}{\sum_{i' = 1}^{\Nis} \iswgt{i'}} \sum_{j = 1}^{\Nin} \frac{\inwgt{i}{j}}{\sum_{j' = 1}^{\Nin} \inwgt{i}{j'}} \{h(\epart{i}{j}) - \targ h \} \right)^2 \\
    \plim 0 \eqsp,
\end{multline*}
as the right hand side tends, as the ancestor archipelago satisfies \AN{ass:AN:CLT}, in distribution to a scaled $\chi^2$-distributed random variable as $\N \rightarrow \infty$. Combining the two previous displays shows that the condition \B{ass:DM:norm:asymptotic:variance} in \autoref{thm:DM:A-3} holds with limit $\varsigma^2(h) = \asvarANone(h)$. To check the condition \B{ass:DM:norm:Lindeberg} in the same lemma, we note that $\max_{i \in \intvect{1}{\Nis}} | \arr{i} | \leq \mbd + \bdconst \bd^2$, with $\bdconst = \bd = 0$ and
\begin{equation} \label{eq:mbd_island_selection}
    \mbd = \max_{i \in \intvect{1}{\Nis}} \sqrt{\frac{\Nin}{\Nis}}  \left| \sum_{j = 1}^{\Nin} \frac{\inwgt{i}{j}}{\sum_{j' = 1}^{\Nin} \inwgt{i}{j'}} \{ h(\epart{i}{j}) - \targ h \} \right| \quad (\N \in \nsetpos) \eqsp.
\end{equation}
Note that the sequence $\{Ê\mbd \}_{\N \in \nsetpos}$ is $\filt$-adapted and vanishes in probability as $\N \rightarrow \infty$ since the ancestor {\WA} is assumed to satisfy~\D{def:deviation:inequality:individual} and thus \autoref{eq:deviation:normalized:estimator}.  We may then apply \autoref{lem:tightness} to check the condition~\B{ass:DM:norm:Lindeberg} in \autoref{thm:DM:A-3}, implying that for all $u \in \rset$,
$$
    \E \left[\exp \left( \operatorname{i} u \sum_{i = 1}^{\Nis} \left\{ \arr{i} - \cexp{\arr{i}}{\filt} \right\} \right) \mid \filt \right] \plim \exp(- u^2 \asvarANone(h) / 2) \eqsp.
$$
Now, using this limit, the decomposition \eqref{eq:decomposition:CLT:island:selection}, and the hypothesis that the ancestor archipelago satisfies \AN{ass:AN:CLT}, we conclude, via \autoref{lemma:asind}, that for all $u \in \rset$,
$$
    \E \left[\exp \left( \operatorname{i} u \sum_{i = 1}^{\Nis} \arr{i} \right) \right] \plim \exp(- u^2 \{ \asvar(h) + \asvarANone(h) \}/ 2) \eqsp,
$$
which concludes the proof of \AN{ass:AN:CLT}.

We establish Assumption~\AN{ass:AN:squared:islands} by applying \autoref{thm:DM:A-1}, this time to the array $\arr[']{i} \eqdef \arr[2]{i}$, $i \in \intvect{1}{\Nis}, \N \in \nsetpos$, where $\{ \arr{i} \}_{i = 1}^{\Nis}$ is defined by \autoref{eq:island:selection:array:AN}, and the filtration $\{Ê\filt \}_{\N \in \nsetpos}$ is defined by \autoref{eq:def:filt}. To prove that $\sum_{i = 1}^{\Nis} \arr[']{i}$ converges in probability, we first note that the sum $\sum_{i = 1}^{\Nis} \cexp[txt]{\arr[']{i}}{\filt}$ converges in probability to $\asvarANone(h)$ by \autoref{eq:island:selection:second:moment}. Moreover, the two conditions \A{ass:DM:cons:tightness} and \A{ass:DM:cons:Lindeberg} in \autoref{thm:DM:A-1} are straightforwardly satisfied for the array under consideration, the latter condition by \autoref{lem:tightness} as $\max_{i \in \intvect{1}{\Nis}} | \arr[']{i} | \leq \mbd[2]$, where $\mbd$ is defined in \autoref{eq:mbd_island_selection} and $\mbd[2]$ vanishes in probability by \autoref{lem:deviation:properly:normalized}. Consequently, \AN{ass:AN:squared:islands} is satisfied with $\asvartdANone = \asvarANone$.

In the case of multinomial island selection, \AN{ass:AN:squared:iswgts} coincides with the consistency property, which is implied by \AN{ass:AN:CLT}; thus, $\condmeastd{1} = \targ$.

We preface the proof of \AN{ass:AN:mut-1} by the following lemma.
 \begin{lemma} \label{lem:deviation:inequality:individual}
    Assume \D{def:deviation:inequality:individual} and \R{sub-geometric:growth}. Then
    \begin{equation*}
        \lim_{\lambda \rightarrow \infty} \sup_{\N \in \nsetpos} \prob \left( \max_{(i, j) \in \intvect{1}{\Nis} \times \intvect{1}{\Nin}} \Nin \frac{\inwgt{i}{j}}{\sum_{j' = 1}^{\Nin}\inwgt{i}{j'}} \geq \lambda\right) = 0 \eqsp.
    \end{equation*}
\end{lemma}
\begin{proof}
Using the boundedness of the particle weights,
$$
	\max_{(i, j) \in \intvect{1}{\Nis} \times \intvect{1}{\Nin}} \Nin \frac{\inwgt{i}{j}}{\sum_{j'=1}^{\Nin}\inwgt{i}{j'}} \leq \inwgtbd \max_{i \in \intvect{1}{\Nis}} \left( \frac{1}{\Nin}\sum_{j = 1}^{\Nin}\inwgt{i}{j} \right)^{-1} \eqsp,
$$
where the quantity on right hand side is tight as it converges in probability to $\inwgtbd/\Hlw$ by  \D{def:deviation:inequality:individual}, \R{sub-geometric:growth} and \autoref{lem:inverse:proba}.
\end{proof}
Now, to check \AN{ass:AN:mut-1} we apply \autoref{thm:DM:A-1} to the array
$$
	\arr{i} \eqdef \frac{\Nin}{\Nis} \sum_{j = 1}^{\Nin} \left( \dfrac{\inwgt{\isind{i}}{j}}{\sum_{j' = 1}^{\Nin}\inwgt{\isind{i}}{j'}} \right)^2 h(\epart{\isind{i}}{j}) \quad (i \in \intvect{1}{\Nis}, \N \in \nsetpos) \eqsp.
$$
associated with the same filtration as previously. Since the ancestor archipelago satisfies \AN{ass:AN:mut-2},
\begin{equation} \label{eq:AN4:cond:limit}
    \sum_{i = 1}^{\Nis} \cexp{\arr{i}}{\filt} = \Nin \sum_{i = 1}^{\Nis} \dfrac{\iswgt{i}}{\sum_{i' = 1}^{\Nis} \iswgt{i'}} \sum_{j = 1}^{\Nin} \left( \dfrac{\inwgt{i}{j}}{\sum_{j' = 1}^{\Nin}\inwgt{i}{j'}}
    \right)^2 h(\epart{i}{j}) \plim \condmeas{3} h \eqsp,
\end{equation}
and we show that $\sum_{i = 1}^{\Nis} \arr{i}$ tends to the same limit by using \autoref{thm:DM:A-1}.
First, $\arr{i} \leq {\Nin}^2 \supn{h} / \Nis < \infty$ for all $i \in \intvect{1}{\Nis}$ and $\N \in \nsetpos$; moreover, be reusing \eqref{eq:AN4:cond:limit} for $|h|$ we check \A{ass:DM:cons:tightness}. To check the Lindeberg condition \A{ass:DM:cons:Lindeberg}, we bound, using $\{\isind{i} \}_{i = 1}^{\Nis} \subset \intvect{1}{\Nis}$,
\begin{multline*}
    \max_{i \in \intvect{1}{\Nis}} |\arr{i}| \leq \max_{(i, j) \in \intvect{1}{\Nis} \times \intvect{1}{\Nin} } \Nin \dfrac{\inwgt{i}{j}}{\sum_{j' = 1}^{\Nin} \inwgt{i}{j'}} \\
     \times \max_{i \in \intvect{1}{\Nis}} \frac{1}{\Nis} \left| \sum_{j = 1}^{\Nin} \frac{\inwgt{i}{j}}{\sum_{j' = 1}^{\Nin} \inwgt{i}{j'}} h(\epart{i}{j}) \right| \eqsp,
\end{multline*}
where the first factor on the right hand side is tight by \autoref{lem:deviation:inequality:individual} and the second term is bounded by $\supn{h} / \Nis$, which tends to zero when $\N$ tends to infinity. Thus, \autoref{lem:inverse:proba} can be applied for checking \A{ass:DM:cons:Lindeberg}, which establishes that $\condmeastd{2} = \condmeas{3}$.

Finally, in the case of selection on the island level, \AN{ass:AN:mut-2} coincides with \AN{ass:AN:mut-1} and \AN{def:tightness:island} is trivially satisfied.


\subsection{Proof of \autoref{thm:consistency:individual:selection}}
\label{sec:proof:consistency:individual:selection}

We first note that \C{def:consistency:max} is trivially satisfied. In order to check \C{def:consistency} we apply \autoref{thm:DM:A-1} to the array
$$
	\arr{i} \eqdef \frac{\iswgt{i}}{\Nin \sum_{i' = 1}^{\Nis} \iswgt{i'}} \sum_{j = 1}^{\Nin} h(\epart{i}{\inind{i}{j}}) \quad (i \in \intvect{1}{\Nis}, \N \in \nsetpos)
$$
associated with $\{ \filt \}_{\N \in \nn}$ given by \eqref{eq:def:filt}. Note that all indices $\{Ê\inind{i}{j} \in \intvect{1}{\Nis} \times \intvect{1}{\Nin} \}$ are conditionally independent given $\filt$. Moreover, for all $i \in \intvect{1}{\Nis}$ it holds that $\{ \inind{i}{j} \}_{j = 1}^{\Nin} \sim \disc(\{ \inwgt{i}{j'} \}_{j' = 1}^{\Nin})^{\varotimes \Nin}$. Hence,
$$
    \sum_{i = 1}^{\Nis} \cexp{\arr{i}}{\filt} =  \sum_{i = 1}^{\Nis} \frac{\iswgt{i}}{\sum_{i' = 1}^{\Nis} \iswgt{i'}} \sum_{j = 1}^{\Nin} \frac{\inwgt{i}{j}}{\sum_{j' = 1}^{\Nin} \inwgt{i}{j'}} h(\epart{i}{j}) \plim \targ h \eqsp,
$$
where convergence holds by assumption. First, note that $\arr{i} \leq \supn{h} < \infty$ for all $i \in \intvect{1}{\Nis}$ and $\N \in \nsetpos$. Moreover, \A{ass:DM:cons:tightness} is trivially satisfied. Thus, consistency is established by showing that \A{ass:DM:cons:Lindeberg} is satisfied, which is an immediate implication of \autoref{lem:tightness} with $X_N = Y_N = 0$ and $V_N = \supn{h} \max_{i \in \intvect{1}{\Nis}} \iswgt{i} / \sum_{i' = 1}^{\Nis} \iswgt{i'}$, which is $\filt$-adapted and tends to zero in probability thanks to \C{def:consistency:max}. Hence, by \autoref{thm:DM:A-1}, the series $\sum_{i = 1}^{\Nis} \arr{i}$ and $\sum_{i = 1}^{\Nis} \cexp{\arr{i}}{\filt}$ have the same limit $\targ h$ in probability. This completes the proof.

\subsection{Proof of \autoref{thm:deviation:inequalities:individual:selection}}
\label{sec:proof:deviation:inequalities:individual:selection}

We may bound the quantity of interest according to
\begin{multline*} \label{eq:max:decomposition}
    \max_{i \in \intvect{1}{\Nis}} \left|  \frac{1}{\Nin} \sum_{j = 1}^{\Nin} h(\epart{i}{\inind{i}{j}}) - \targ h \right| \\
    \leq \max_{i \in \intvect{1}{\Nis}} \frac{1}{\Nin} \left| \sum_{j = 1}^{\Nin} \diff{i}{j} \right|Ê+ \max_{i \in \intvect{1}{\Nis}} \left| \sum_{j = 1}^{\Nin} \frac{\inwgt{i}{j}}{\sum_{j' = 1}^{\Nin} \inwgt{i}{j'}} h(\epart{i}{j}) - \targ h \right| \eqsp,
\end{multline*}
where we have set
\begin{equation} \label{def:diff:ind:selection}
    \diff{i}{j} \eqdef h(\epart{i}{\inind{i}{j}}) - \sum_{j' = 1}^{\Nin} \frac{\inwgt{i}{j'}}{\sum_{j'' = 1}^{\Nin} \inwgt{i}{j''}} h(\epart{i}{j'}) \quad ((i, j) \in \intvect{1}{\Nis} \times \intvect{1}{\Nin}) \eqsp.
\end{equation}
By \autoref{lem:deviation:properly:normalized}, the second term on the right hand side satisfies
$$
    \prob \left( \max_{i \in \intvect{1}{\Nis}} \left| \sum_{j = 1}^{\Nin} \frac{\inwgt{i}{j}}{\sum_{j' = 1}^{\Nin} \inwgt{i}{j'}} h(\epart{i}{j}) - \targ h \right| \geq \varepsilon/2 \right)
    \leq 2 \Hc \Nis \exp{\left( - \He \Nin \frac{\varepsilon^2 \Hlw^2}{16 \supn{h}^2} \right)} \eqsp.
 $$
For each $i \in \intvect{1}{\Nis}$, the variables $\{ \diff{i}{j} \}_{j = 1}^{\Nin}$ are, conditionally on $\filt$, independent and identically distributed with zero mean; moreover, as $|\diff{i}{j}| \leq 2\supn{h}$ for all $(i, j) \in \intvect{1}{\Nis} \times \intvect{1}{\Nin}$, Hoeffding's inequality implies that for all $\varepsilon > 0$,
\begin{equation} \label{eq:conditional:hoeffding:individual:selection}
    \prob \left( \max_{i \in \intvect{1}{\Nis}} \frac{1}{\Nin} \left| \sum_{j = 1}^{\Nin} \diff{i}{j} \right| \geq \varepsilon/2 \mid \filt \right) \leq 2 \Nis \exp \left( -  \Nin \frac{ \varepsilon^2}{8 \supn{h}^2} \right) \eqsp.
\end{equation}
Combining the previous two displays show that \D{def:deviation:inequality:individual} is satisfied with the choice of $\Hconstd{1}$ and $\Hconstd{2}$ given in the theorem.

\subsection{Proof of \autoref{thm:normality:individual:selection}}
\label{sec:proof:normality:individual:selection}

We start with \AN{ass:AN:CLT}. In order to apply \autoref{thm:DM:A-3}, define the array
\begin{equation} \label{def:arr:AN:PS}
    \arr{i} \eqdef \sqrt{\frac{\Nis}{\Nin}} \frac{\iswgt{i}}{\sum_{i' = 1}^{\Nis} \iswgt{i'}} \sum_{j=1}^{\Nin} \{h(\epart{i}{\inind{i}{j}}) - \targ h\} \quad ((i, j) \in \intvect{1}{\Nis} \times \intvect{1}{\Nin}) \eqsp,
\end{equation}
equipped with the usual filtration $\{\filt \}_{\N \in \nn}$ given by \autoref{eq:def:filt}. We first note that $\arr{i} \leq 2 \sqrt{\N} \supn{h} < \infty$ for all $i \in \intvect{1}{\Nis}$ and $\N \in \nsetpos$. In order to check \B{ass:DM:norm:asymptotic:variance}\!, write, following the arguments of the proof of \autoref{thm:consistency:individual:selection},
\begin{multline*}
    \sum_{i = 1}^{\Nis} \cexp{\arr[2]{i}}{\filt} = \Nis \sum_{i = 1}^{\Nis} \left( \frac{\iswgt{i}}{\sum_{i' = 1}^{\Nis} \iswgt{i'}} \right)^2 \sum_{j = 1}^{\Nin} \frac{\inwgt{i}{j}}{\sum_{j' = 1}^{\Nin} \inwgt{i}{j'}} \{h(\epart{i}{j}) - \targ h \}^2 \\
    + \Nis (\Nin - 1) \sum_{i = 1}^{\Nis} \left( \frac{\iswgt{i}}{\sum_{i' = 1}^{\Nis} \iswgt{i'}} \right)^2 \left( \sum_{j = 1}^{\Nin} \frac{\inwgt{i}{j}}{\sum_{j' = 1}^{\Nin} \inwgt{i}{j'}} \{h(\epart{i}{j}) - \targ h \} \right)^2 \eqsp.
\end{multline*}
Moreover, since
\begin{multline*}
    \sum_{i = 1}^{\Nis} \E^2\left[\arr{i} \mid \filt \right] \\
    = \N \sum_{i = 1}^{\Nis} \left( \frac{\iswgt{i}}{\sum_{i' = 1}^{\Nis} \iswgt{i'}} \right)^2 \left( \sum_{j = 1}^{\Nin} \frac{\inwgt{i}{j}}{\sum_{j' = 1}^{\Nin} \inwgt{i}{j'}} \{h(\epart{i}{j})- \targ h \} \right)^2 \eqsp,
\end{multline*}
we obtain
\begin{align} \label{eq:ass:B1:indS}
    \lefteqn{\sum_{i = 1}^{\Nis} \left\{Ê\cexp{\arr[2]{i}}{\filt} - \E^2\left[\arr{i} \mid \filt \right] \right\}} \\
    &= \Nis \sum_{i = 1}^{\Nis} \left( \frac{\iswgt{i}}{\sum_{i' = 1}^{\Nis} \iswgt{i'}} \right)^2 \sum_{j = 1}^{\Nin} \frac{\inwgt{i}{j}}{\sum_{j' = 1}^{\Nin} \inwgt{i}{j'}} \{ h(\epart{i}{j}) - \targ h \}^2 \nonumber \\
    &- \Nis \sum_{i = 1}^{\Nis} \left( \frac{\iswgt{i}}{\sum_{i' = 1}^{\Nis} \iswgt{i'}} \right)^2 \left( \sum_{j = 1}^{\Nin} \frac{\inwgt{i}{j}}{\sum_{j' = 1}^{\Nin} \inwgt{i}{j'}} \{h(\epart{i}{j}) -  \targ h \} \right)^2 \nonumber \eqsp.
\end{align}
Since the ancestor archipelago satisfies \eqref{eq:deviation:normalized:estimator} and is consistent for $\targ$, we deduce that
\begin{equation*} \label{ass:AN:squared:iswgts:islands}
    \Nis \sum_{i = 1}^{\Nis} \left( \frac{\iswgt{i}}{\sum_{i' = 1}^{\Nis} \iswgt{i'}} \right)^2 \left( \sum_{j = 1}^{\Nin} \frac{\inwgt{i}{j}}{\sum_{j' = 1}^{\Nin} \inwgt{i}{j'}} \{h(\epart{i}{j}) - \targ h \} \right)^2 \plim[N] 0 \eqsp.
\end{equation*}
Then, since the ancestor archipelago also satisfies \AN{ass:AN:squared:iswgts} we conclude that the variance \eqref{eq:ass:B1:indS} tends in probability to $\condmeas{1}\{ (h - \targ h)^2\}$. Consequently, the triangular array satisfies Assumption~\B{ass:DM:norm:asymptotic:variance} with limit $\condmeas{1}\{ (h - \targ h)^2\}$. In order to check Assumption~\B{ass:DM:norm:Lindeberg} we may apply \autoref{lem:tightness} by bounding
\begin{equation*} \label{eq:ass:B2:indS}
	\max_{i \in \intvect{1}{\Nis}} |\arr{i}| \leq \mbd + \bdconst \bd^2 \quad (\N \in \nsetpos) \eqsp,
\end{equation*}
with, for $\N \in \nsetpos$,
$$
\begin{cases}
    \displaystyle \mbd = \sqrt{\N} \max_{i \in \intvect{1}{\Nis}} \frac{\iswgt{i}}{\sum_{i' = 1}^{\Nis} \iswgt{i'}} \max_{i \in \intvect{1}{\Nis}} \left| \sum_{j = 1}^{\Nin} \dfrac{\inwgt{i}{j}}{\sum_{j' = 1}^{\Nin} \inwgt{i}{j'}} \{ h(\epart{i}{j}) - \targ h \} \right| \eqsp, \\
    \displaystyle \bdconst = \Nis \max_{i \in \intvect{1}{\Nis}}  \frac{\iswgt{i}}{\sum_{i' = 1}^{\Nis} \iswgt{i'}} \eqsp, \\
    \displaystyle \bd^2 = \sqrt{\frac{\Nin}{\Nis}} \max_{i \in \intvect{1}{\Nis}} \bigg| \frac{1}{\Nin} \sum_{j = 1}^{\Nin} \diff{i}{j} \bigg| \eqsp,
\end{cases}
$$
where the $\delta_{\N}$s are defined in \eqref{def:diff:ind:selection}. Here $\{Ê\mbd \}_{\N \in \nsetpos}$ is $\filt$-adapted and tends to zero in probability by
\AN{def:tightness:island} and \autoref{lem:deviation:properly:normalized}. In addition, $\{Ê\bdconst \}_{\N \in \nsetpos}$ is $\filt$-adapted and, by \AN{def:tightness:island}\!, tight. Moreover, for all $\N \in \nsetpos$, $\bd$ has, by \eqref{eq:conditional:hoeffding:individual:selection}, a tail of the type
$$
	\prob(Y_N \geq \varepsilon \mid \filt) \leq  2 \Nis \exp\left( -\Nis \frac{\varepsilon^4}{2 \supn{h}^2} \right) \eqsp.
$$
Thus, \autoref{lem:tightness} applies, which establishes \B{ass:DM:norm:Lindeberg}. Finally, we may conclude, by using \autoref{lemma:asind}, the proof of \AN{ass:AN:CLT}  to obtain that $\asvartd(h) = \asvar(h) + \condmeas{1}\{(h - \targ h)^2 \}$.

To check \AN{ass:AN:squared:islands} and prove that the series $\sum_{i = 1}^{\Nis} \arr{i}$, where
$$
    \arr{i} \eqdef \frac{\iswgt{i}}{\Nin \sum_{i' = 1}^{\Nis} \iswgt{i'}} \left( \sum_{j = 1}^{\Nin} \{h(\epart{i}{\inind{i}{j}}) - \targ h \} \right)^2 \quad (i \in \intvect{1}{\Nis}, \N \in \nsetpos) \eqsp,
$$
converges in probability as $\N \rightarrow \infty$, let $\{ \filt \}_{\N \in \nn}$ be defined as in \eqref{eq:def:filt} and consider the sum
\begin{multline} \label{eq:ind:sel:development:AN2}
    \sum_{i = 1}^{\Nis} \cexp{\arr{i}}{\filt} = \sum_{i = 1}^{\Nis} \frac{\iswgt{i}}{\sum_{i' = 1}^{\Nis} \iswgt{i'}} \sum_{j = 1}^{\Nin} \frac{\inwgt{i}{j}}{\sum_{j' = 1}^{\Nin} \inwgt{i}{j'}} \{h(\epart{i}{j}) - \targ h \}^2 \\
    + (\Nin - 1) \sum_{i = 1}^{\Nis} \frac{\iswgt{i}}{\sum_{i' = 1}^{\Nis} \iswgt{i'}} \left( \sum_{j = 1}^{\Nin} \frac{\inwgt{i}{j}}{\sum_{j' = 1}^{\Nin} \inwgt{i}{j'}} \{h(\epart{i}{j}) - \targ h \} \right)^2 \eqsp,
\end{multline}
where we used, as previously, that for each $i \in \intvect{1}{\Nis}$, the variables $\{ h(\epart{i}{\inind{i}{j}}) \}_{j = 1}^{\Nin}$ are, conditionally on $\filt$,
independent and identically distributed with common mean
$\sum_{j = 1}^{\Nin} \inwgt{i}{j}  h(\epart{i}{j})/\sum_{j' = 1}^{\Nin} \inwgt{i}{j'}$.
The first term of the right hand side of \eqref{eq:ind:sel:development:AN2} tends in probability to $\targ \{(h- \targ h)^2\}$ by consistency, while the second term tends in probability to $\asvarANone(h)$ by \AN{ass:AN:squared:islands}. Since this establishes the condition \A{ass:DM:cons:tightness} in \autoref{thm:DM:A-1}, the series $\sum_{i = 1}^{\Nis} \arr{i}$ and $\sum_{i = 1}^{\Nis} \cexp{\arr{i}}{\filt}$ have the same limit $\targ \{(h- \targ h)^2\} + \asvarANone(h)$ in probability as soon as the condition \A{ass:DM:cons:Lindeberg} in the same theorem can be checked for the array in question. However, write
$$
	\max_{i \in \intvect{1}{\Nis}} |\arr{i}| \leq \mbd + \bdconst \bd^2 \quad (\N \in \nsetpos) \eqsp,
$$
where, for $\N \in \nsetpos$,
$$
        \begin{cases}
                \mbd = \displaystyle 2 \Nin \max_{i \in \intvect{1}{\Nis}} \frac{\iswgt{i}}{\sum_{i' = 1}^{\Nis} \iswgt{i'}} \left( \max_{i \in \intvect{1}{\Nis}} \left| \sum_{j = 1}^{\Nin} \frac{\inwgt{i}{j}}{\sum_{j' = 1}^{\Nin} \inwgt{i}{j'}} h(\epart{i}{j}) \right| \right)^2 \eqsp, \\
                \bdconst = \displaystyle 2 \max_{i \in \intvect{1}{\Nis}} \Nis \frac{\iswgt{i}}{\sum_{i' = 1}^{\Nis} \iswgt{i'}} \eqsp, \\
                \bd = \displaystyle \max_{i \in \intvect{1}{\Nis}} \frac{1}{\sqrt{\N}} \bigg| \sum_{j = 1}^{\Nin} \diff{i}{j} \bigg| \eqsp,
        \end{cases}
$$
and the $\delta_\N$s  are defined in \autoref{def:diff:ind:selection}; then, since $\mbd$ tends to zero in probability (by \AN{def:tightness:island} and \autoref{lem:deviation:properly:normalized}), $\bdconst$ is tight, and $\bd$ has an exponential tail (by \autoref{eq:conditional:hoeffding:individual:selection}), \autoref{lem:tightness} applies, establishing that the array satisfies Assumption~\A{ass:DM:cons:Lindeberg}. Consequently, we obtain that $\asvartdANone(h) = \targ \{(h - \targ h)^2 \} + \asvarANone(h)$.

To verify \AN{ass:AN:squared:iswgts} we retain to the previous machinery and study the array
$$
    \arr{i} \eqdef \frac{\Nis}{\Nin} \left( \frac{\iswgt{i}}{\sum_{i' = 1}^{\Nis} \iswgt{i'}} \right)^2 \sum_{j = 1}^{\Nin} h(\epart{i}{\inind{i}{j}}) \quad (i \in \intvect{1}{\Nis}, \N \in \nsetpos)
$$
associated with the filtration $\{\filt \}_{\N \in \nsetpos}$ defined in \autoref{eq:def:filt}. To establish the convergence of $\sum_{i = 1}^{\Nis} \arr{i}$ we reapply \autoref{thm:DM:A-1} and consider
\begin{equation} \label{eq:AN3:cond:limit}
    \sum_{i = 1}^{\Nis} \cexp{\arr{i}}{\filt} = \Nis \sum_{i = 1}^{\Nis} \left( \frac{\iswgt{i}}{\sum_{i' = 1}^{\Nis} \iswgt{i'}} \right)^2 \sum_{j = 1}^{\Nin} \frac{\inwgt{i}{j}}{\sum_{j' = 1}^{\Nin} \inwgt{i}{j'}} h(\epart{i}{j}) \plim \condmeas{1} h \eqsp,
\end{equation}
where convergence follows since the ancestor archipelago satisfies \AN{ass:AN:squared:iswgts}. By reusing \eqref{eq:AN3:cond:limit} for $|h|$ we check that the condition \A{ass:DM:cons:tightness} in \autoref{thm:DM:A-1} is satisfied. Moreover, since
$$
    \max_{i \in \intvect{1}{\Nis}} |\arr{i}| \leq  \supn{h} \left( \sqrt{\Nis} \frac{\iswgt{i}}{\sum_{i' = 1}^{\Nis} \iswgt{i'}} \right)^2 \eqsp,
$$
where the right hand side vanishes in probability by \AN{def:tightness:island}, \autoref{lem:tightness} implies that the array satisfies \A{ass:DM:cons:Lindeberg} as well. Thus, \AN{ass:AN:squared:iswgts} holds true with $\condmeastd{1} = \condmeas{1}$.

In addition, since Assumption~\AN{ass:AN:mut-1} coincides with \AN{ass:AN:squared:iswgts} in the case of uniform particle weights, we obtain immediately that $\condmeastd{2} = \condmeas{1}$. Moreover, \AN{ass:AN:mut-2} coincides precisely with \C{def:consistency}, 
which is satisfied as the output satisfies the stronger condition \AN{ass:AN:CLT}, and we obtain $\condmeastd{3} = \targ$. Finally, \AN{def:tightness:island} holds trivially true.


\subsection{Proof of \autoref{thm:consistency:mutation}}
\label{sec:proof:consistency:mutation}


First, note that
\begin{multline} \label{decomposition:mutation}
    \sum_{i = 1}^{\Nis} \frac{\iswgttd{i}}{\sum_{i' = 1}^{\Nis} \iswgttd{i'}} \sum_{j = 1}^{\Nin} \frac{\inwgttd{i}{j}}{\sum_{j' = 1}^{\Nin} \inwgttd{i}{j'}} h(\eparttd{i}{j}) \\
    = \dfrac{\sum_{i' = 1}^{\Nis} \iswgt{i'}}{\sum_{i'' = 1}^{\Nis} \iswgttd{i''}} \sum_{i = 1}^{\Nis} \frac{\iswgt{i}}{\sum_{i' = 1}^{\Nis} \iswgt{i'}} \sum_{j = 1}^{\Nin} \frac{\inwgttd{i}{j}}{\sum_{j' = 1}^{\Nin} \inwgt{i}{j'}} h(\eparttd{i}{j}) \eqsp,
\end{multline}
using the definition of $\{ \iswgttd{i} \}_{i = 1}^{\Nis}$ in Algorithm~\ref{alg:mutation}. In order to determine the limit in probability of this quantity we apply \autoref{thm:DM:A-1} to the array
$$
    \arr{i} \eqdef \frac{\iswgt{i}}{\sum_{i' = 1}^{\Nis} \iswgt{i'}} \sum_{j = 1}^{\Nin} \frac{\inwgttd{i}{j}}{\sum_{j' = 1}^{\Nin} \inwgt{i}{j'}} h(\eparttd{i}{j}) \quad (i \in \intvect{1}{\Nis}, \N \in \nsetpos)
$$
associated with the filtration $\{Ê\filt \}_{\N \in \nsetpos}$ given in \eqref{eq:def:filt}. For each $(i, j) \in \intvect{1}{\Nis} \times \intvect{1}{\Nin}$, the conditional distribution of $\eparttd{i}{j}$ given $\filt$ is $\prop(\epart{i}{j}, \cdot)$; thus,
\begin{equation} \label{eq:cond:exp:mutation:wh}
    \begin{split}
        \cexp{\inwgttd{i}{j}h(\eparttd{i}{j})}{\filt} &= \inwgt{i}{j} \int \der(\epart{i}{j}, \tilde{x})h(\tilde{x}) \, \prop(\epart{i}{j}, \rmd \tilde{x}) \\
        &= \inwgt{i}{j} \op h (\epart{i}{j}) \eqsp,
    \end{split}
\end{equation}
implying that
\begin{multline} \label{eq:unbiasedness:mutation}
    \sum_{i = 1}^{\Nis} \cexp{\arr{i}}{\filt} = \sum_{i = 1}^{\Nis} \frac{\iswgt{i}}{\sum_{i' = 1}^{\Nis} \iswgt{i'}} \sum_{j = 1}^{\Nin} \frac{\inwgt{i}{j}}{\sum_{j' = 1}^{\Nin} \inwgt{i}{j'}} \op h (\epart{i}{j}) \plim \targ \op h \eqsp,
\end{multline}
where convergence holds since the ancestor archipelago satisfies Assumption~\C{def:consistency}. This implies~\A{ass:DM:cons:tightness}. To check also the condition \A{ass:DM:cons:Lindeberg} we apply \autoref{lem:tightness} with $X_N = Y_N = 0$ and $V_N = \supn{\der} \supn{h} \max_{i \in \intvect{1}{\Nis}} \iswgt{i} / \sum_{i' = 1}^{\Nis} \iswgt{i'}$, where $V_N$ is $\filt$-adapted and tends to zero in probability by the assumption~\C{def:consistency:max}\!. Hence, \autoref{thm:DM:A-1} ensures that the two series $\sum_{i = 1}^{\Nis} \arr{i}$ and $\sum_{i = 1}^{\Nis} \cexp{\arr{i}}{\filt}$ have the same limit $\targ \op h$ in probability. Moreover, by setting $h$ is equal to the constant function $\1{\stsptd}$ we deduce that
\begin{equation} \label{consistency:mutation:1}
	\frac{\sum_{i = 1}^{\Nis} \iswgt{i}}{\sum_{i' = 1}^{\Nis} \iswgttd{i'}} \plim \dfrac{1}{\targ \op \1{\stsptd}} \eqsp,
\end{equation}
which allows us to complete the proof of \C{def:consistency} using Slutsky's lemma.


Finally, Assumption~\C{def:consistency:max} is checked straightforwardly by just noting that
$$
    \max_{i \in \intvect{1}{\Nis}} \dfrac{\iswgttd{i}}{\sum_{i' = 1}^{\Nis} \iswgttd{i'}} \leq  \supn{\der} \max_{i \in \intvect{1}{\Nis}} \dfrac{\iswgt{i}}{\sum_{i' = 1}^{\Nis} \iswgt{i'}} \dfrac{\sum_{i' = 1}^{\Nis} \iswgt{i'}}{\sum_{i'' = 1}^{\Nis} \iswgttd{i''}} \eqsp,
$$
where the right hand side tends to zero in probability by \eqref{consistency:mutation:1} and the fact that the ancestor archipelago satisfies~\C{def:consistency}.

\subsection{Proof of \autoref{thm:deviation:inequalities:mutation}}
\label{sec:proof:deviation:inequalities:mutation}
Note that $\Hlwtd \times \targtd h = \Hlw \times \targ \op h$ 
and bound the quantity of interest according to
\begin{multline} \label{eq:max:decomposition}
    \max_{i \in \intvect{1}{\Nis}} \left| \frac{1}{\Nin} \sum_{j = 1}^{\Nin} \inwgttd{i}{j} h(\eparttd{i}{j}) - \Hlw \times \targ \op h \right| \\
    \leq \max_{i \in \intvect{1}{\Nis}} \frac{1}{\Nin} \left| \sum_{j = 1}^{\Nin} \difftd{i}{j} \right| + \max_{i \in \intvect{1}{\Nis}} \frac{1}{\Nin} \left| \sum_{j = 1}^{\Nin} \inwgt{i}{j} \op h(\epart{i}{j}) - \Hlw \times \targ \op h \right| \eqsp,
\end{multline}
where
\begin{equation} \label{eq:diff:def}
    \difftd{i}{j} \eqdef \inwgttd{i}{j} h(\eparttd{i}{j}) - \inwgt{i}{j} \op h(\epart{i}{j}) \quad ((i, j) \in \intvect{1}{\Nis} \times \intvect{1}{\Nin}) \eqsp.
\end{equation}
Since the input archipelago satisfies~\D{def:deviation:inequality:individual} it holds that
\begin{multline*}
    \prob \left(  \max_{i \in \intvect{1}{\Nis}} \frac{1}{\Nin} \left| \sum_{j = 1}^{\Nin} \{ \inwgt{i}{j} \op h(\epart{i}{j}) - \Hlw \times \targ \op h \} \right|  \geq \varepsilon/2 \right) \\
    \leq \Nis \Hc \exp{\left( - \He \Nin \frac{\varepsilon^2}{4 \supn{\op \1{\stsp}}^2 \supn{h}^2} \right)} \eqsp.
\end{multline*}
For each $i \in \intvect{1}{\Nis}$, the random variables $\{ \difftd{i}{j} \}_{j = 1}^{\Nin}$ are, conditionally on $\filt$, independent and, by \eqref{eq:cond:exp:mutation:wh}, zero mean. Moreover, since for all $(i, j) \in \intvect{1}{\Nis} \times \intvect{1}{\Nin}$, $|\difftd{i}{j}| \leq \delta \supn{h}$, where $\delta$ is defined in the statement of theorem, Hoeffding's inequality implies that for all $\varepsilon > 0$,
\begin{equation} \label{eq:conditional:hoeffding:mutation}
    \prob \left( \max_{i \in \intvect{1}{\Nis}} \frac{1}{\Nin} \left| \sum_{j = 1}^{\Nin} \difftd{i}{j} \right| \geq \varepsilon/2 \mid \filt \right) \leq 2 \Nis \exp \left( - \Nin \frac{ \varepsilon^2}{2 \delta^2 \supn{h}^2} \right) \eqsp.
\end{equation}
By combining the two previous displays we may conclude that \D{def:deviation:inequality:individual} is satisfied with $\Hconstd{1}$ and $\Hconstd{2}$ defined as in the theorem statement.

\subsection{Proof of \autoref{thm:normality:mutation}}
\label{sec:proof:normality:mutation}
We preface the proof by the following auxiliary result, which is obtained as a straightforward extension of the generalized hoeffding inequality in \cite[Lemma~4]{douc:garivier:moulines:olsson:2010}.
\begin{lemma} \label{lem:Hoeffding:ratio:island:weights}
    Let the assumptions of \autoref{thm:deviation:inequalities:mutation} hold. Then for all $\Nis \in \nsetpos$, $\Nin \in \nsetpos$, and $\varepsilon > 0$, \begin{equation} \label{convergence:proba:inverse}
        \prob \left( \max_{i \in \intvect{1}{\Nis}} \left| \frac{\iswgttd{i}}{\iswgt{i}} - \targ \op \1{\stsptd} \right| \geq \varepsilon \right) \leq \Nis \Hconstck{1} \exp \left( - \Hconstck{2} \Nin \varepsilon^2 \right) \eqsp,
    \end{equation}
    where $\Hconstck{1} \eqdef 2 (\Hc \vee \Hconstd{1})$ and $\Hconstck{2} \eqdef \{Ê(\He / \supn{\der}^2) \wedge \Hconstd{2} \} \Hlw^2/4$.
\end{lemma}


To check \AN{ass:AN:CLT}\!, take $h \in \bmf{\stfd}$ and assume without loss of generality that $\targtd h = 0$ and, consequently, $\targ \op h = 0$. We again rewrite the estimator according to \eqref{decomposition:mutation} and apply \autoref{thm:DM:A-3} to the second factor. For this purpose, define the array
\begin{equation} \label{eq:mutation:array:AN}
    \arr{i} \eqdef \sqrt{\N} \frac{\iswgt{i}}{\sum_{i' = 1}^{\Nis} \iswgt{i'}} \sum_{j = 1}^{\Nin} \frac{\inwgttd{i}{j}}{\sum_{j' = 1}^{\Nin} \inwgt{i}{j'}} h(\eparttd{i}{j}) \quad (i \in \intvect{1}{\Nis}, \N \in \nsetpos) \eqsp,
\end{equation}
and furnish the same with the filtration $\{Ê\filt \}_{\N \in \nsetpos}$ defined in \eqref{eq:def:filt}. We may now write
\begin{multline} \label{eq:decomposition:CLT:mutation}
    \sum_{i = 1}^{\Nis} \arr{i} = \sum_{i = 1}^{\Nis} \{ \arr{i} - \cexp{\arr{i}}{\filt} \} \\
    + \sqrt{\N} \sum_{i = 1}^{\Nis} \frac{\iswgt{i}}{\sum_{i' = 1}^{\Nis} \iswgt{i'}} \sum_{j = 1}^{\Nin} \frac{\inwgt{i}{j}}{\sum_{j' = 1}^{\Nin} \inwgt{i}{j'}} \op h(\epart{i}{j}) \eqsp,
\end{multline}
where, by assumption, since $\supn{\op h} \leq \supn{h} \supn[txt]{\op \1{\stsptd}}$, the second term on the right hand side satisfies the CLT
$$
    \sqrt{\N} \sum_{i = 1}^{\Nis} \frac{\iswgt{i}}{\sum_{i' = 1}^{\Nis} \iswgt{i'}} \sum_{j = 1}^{\Nin} \frac{\inwgt{i}{j}}{\sum_{j' = 1}^{\Nin} \inwgt{i}{j'}} \op h(\epart{i}{j}) \dlim \normdist(0, \asvar(\op h)) \eqsp.
$$
Our main challenge will be to handle the first term on the right hand side of \eqref{eq:decomposition:CLT:mutation}. Since all individuals of the mutated archipelago are conditionally independent given $\filt$, we notice that for all $i \in \intvect{1}{\Nis}$,
\begin{multline*}
    \cexp{\left( \sum_{j = 1}^{\Nin} \frac{\inwgttd{i}{j}}{\sum_{j' = 1}^{\Nin} \inwgt{i}{j'}} h(\eparttd{i}{j}) \right)^2}{\filt}
    = \sum_{j = 1}^{\Nin} \left( \frac{\inwgt{i}{j}}{\sum_{j' = 1}^{\Nin} \inwgt{i}{j'}}\right)^2 \prop(\der^2 h^2)(\epart{i}{j}) \\
    + \left( \sum_{j = 1}^{\Nin} \frac{\inwgt{i}{j}}{\sum_{j' = 1}^{\Nin} \inwgt{i}{j'}} \op h(\epart{i}{j}) \right)^2
    - \sum_{j = 1}^{\Nin} \left( \frac{\inwgt{i}{j}}{\sum_{j' = 1}^{\Nin} \inwgt{i}{j'}} \right)^2  (\op h)^2(\epart{i}{j}) \eqsp.
\end{multline*}
Using this, we turn to the variance and deduce the expression
\begin{multline} \label{eq:ass:B1:mut}
    \sum_{i = 1}^{\Nis} \left\{Ê\cexp{\arr[2]{i}}{\filt} - \E^2 \left[\arr{i} \mid \filt \right] \right\} \\
    = \N \sum_{i = 1}^{\Nis} \left( \frac{\iswgt{i}}{\sum_{i' = 1}^{\Nis} \iswgt{i'}} \right)^2 \sum_{j = 1}^{\Nin} \left( \frac{\inwgt{i}{j}}{\sum_{j' = 1}^{\Nin} \inwgt{i}{j'}}\right)^2 \prop(\der^2 h^2)(\epart{i}{j}) \nonumber \\
    - \N \sum_{i = 1}^{\Nis} \left( \frac{\iswgt{i}}{\sum_{i' = 1}^{\Nis} \iswgt{i'}} \right)^2 \sum_{j = 1}^{\Nin} \left( \frac{\inwgt{i}{j}}{\sum_{j' = 1}^{\Nin} \inwgt{i}{j'}} \right)^2  (\op h)^2(\epart{i}{j}) \eqsp, \nonumber
\end{multline}
which tends in probability to $\condmeas{2} \prop(\der^2 h^2) - \condmeas{2}(\op h)^2$ as the input archipelago satisfies Assumption~\AN{ass:AN:mut-1}. This implies that Assumption~\B{ass:DM:norm:asymptotic:variance} in \autoref{thm:DM:A-3} holds with the same limit. To verify the Lindeberg condition \B{ass:DM:norm:Lindeberg} in \autoref{thm:DM:A-3}, note that proceeding as in \eqref{eq:max:decomposition} yields
\begin{equation} \label{eq:ass:B2:mut}
    \max_{i \in \intvect{1}{\Nis}} |\arr{i}|
    \leq \mbd + \bdconst \bd^2 \quad (\N \in \nsetpos) \eqsp,
\end{equation}
where, for $\N \in \nsetpos$,
\[
    \begin{cases}
        \mbd = \displaystyle \sqrt{\N} \max_{i \in \intvect{1}{\Nis}} \frac{\iswgt{i}}{\sum_{i' = 1}^{\Nis} \iswgt{i'}} \max_{i \in \intvect{1}{\Nis}}  \left| \sum_{j = 1}^{\Nin} \frac{\inwgt{i}{j}}{\sum_{j' = 1}^{\Nin} \inwgt{i}{j'}} \op h(\epart{i}{j}) \right| \eqsp, \\
        \bdconst = \displaystyle \Nis \max_{i \in \intvect{1}{\Nis}} \frac{\iswgt{i}}{\sum_{i = 1}^{\Nis} \iswgt{i'}} \eqsp, \\
        \bd^2 = \displaystyle \sqrt{\frac{\Nin}{\Nis}}
        \max_{i \in \intvect{1}{\Nis}} \left| \frac{\sum_{j = 1}^{\Nin}\difftd{i}{j}}{\sum_{j' = 1}^{\Nin} \inwgt{i}{j'}} \right| \eqsp.
    \end{cases}
\]
Here $\{ÊV_N \}_{\N \in \nsetpos}$ is $\filt$-adapted and tends to zero in probability by \AN{def:tightness:island} and \autoref{lem:deviation:properly:normalized}, $\{ÊX_N \}_{\N \in \nsetpos}$ is $\filt$-adapted and tight by \AN{def:tightness:island}\!, and $\bd$ has, by \eqref{eq:conditional:hoeffding:mutation}, \D{def:deviation:inequality:individual}, and the extension of Hoeffding's inequality in \cite[Lemma~4]{douc:garivier:moulines:olsson:2010}, a tail of the form \eqref{eq:lemma:exp:tail} (with $\alpha = 2$). Thus, by \autoref{lem:tightness}, \B{ass:DM:norm:Lindeberg} holds true, and we may conclude the proof of \AN{ass:AN:CLT} using first \autoref{lemma:asind} and then Slutsky's lemma.


We turn to \AN{ass:AN:squared:islands} and decompose the quantity under consideration according to
\begin{multline} \label{eq:mutation:decomposition:AN2}
    \Nin \sum_{i = 1}^{\Nis} \frac{\iswgttd{i}}{\sum_{i' = 1}^{\Nis} \iswgttd{i'}} \left( \sum_{j = 1}^{\Nin} \frac{\inwgttd{i}{j}}{\sum_{j' = 1}^{\Nin} \inwgttd{i}{j'}} h(\eparttd{i}{j}) \right)^2 \\
    = \Nin \dfrac{\sum_{i'' = 1}^{\Nis} \iswgt{i''}}{\sum_{i' = 1}^{\Nis} \iswgttd{i'}} \sum_{i = 1}^{\Nis} \frac{\iswgt{i}}{\sum_{i'' = 1}^{\Nis} \iswgt{i''}} \left(\frac{\iswgt{i}}{\iswgttd{i}} - \frac{1}{\targ \op \1{\stsptd}}\right) \left( \sum_{j = 1}^{\Nin} \frac{\inwgttd{i}{j}}{\sum_{j' = 1}^{\Nin} \inwgt{i}{j'}} h(\eparttd{i}{j}) \right)^2 \\
    + \Nin \frac{1}{\targ \op \1{\stsptd}} \dfrac{\sum_{i'' = 1}^{\Nis} \iswgt{i''}}{\sum_{i' = 1}^{\Nis} \iswgttd{i'}} \sum_{i = 1}^{\Nis} \frac{\iswgt{i}}{\sum_{i'' = 1}^{\Nis} \iswgt{i''}} \left( \sum_{j = 1}^{\Nin} \frac{\inwgttd{i}{j}}{\sum_{j' = 1}^{\Nin} \inwgt{i}{j'}} h(\eparttd{i}{j}) \right)^2 \eqsp.
\end{multline}
The convergence in probability of the second term on the right hand side will now to be established using \autoref{thm:DM:A-1}. For this purpose, define the triangular array
\begin{equation} \label{eq:def:array:AN2}
    \arr{i} \eqdef \Nin \frac{\iswgt{i}}{\sum_{i' = 1}^{\Nis} \iswgt{i'}} \left( \sum_{j = 1}^{\Nin} \frac{\inwgttd{i}{j}}{\sum_{j' = 1}^{\Nin} \inwgt{i}{j'}} h(\eparttd{i}{j}) \right)^2 \quad (i \in \intvect{1}{\Nis}, \N \in \nsetpos) \eqsp,
\end{equation}
and associate the same with the $\sigma$-field $\filt$ defined in \eqref{eq:def:filt}. We now apply the previous machinery and study the convergence of the series
\begin{multline*}
    \sum_{i = 1}^{\Nis} \cexp{\arr{i}}{\filt} = \Nin \sum_{i = 1}^{\Nis} \frac{\iswgt{i}}{\sum_{i' = 1}^{\Nis} \iswgt{i'}}
    \sum_{j = 1}^{\Nin} \left(\frac{\inwgt{i}{j}}{\sum_{j' = 1}^{\Nin} \inwgt{i}{j'}} \right)^2 \prop(\der^2 h^2) (\epart{i}{j}) \\
    + \Nin \sum_{i = 1}^{\Nis} \frac{\iswgt{i}}{\sum_{i' = 1}^{\Nis} \iswgt{i'}}
    \left( \sum_{j = 1}^{\Nin} \frac{\inwgt{i}{j}}{\sum_{j' = 1}^{\Nin} \inwgt{i}{j'}} \op h (\epart{i}{j}) \right)^2 \\
    - \Nin \sum_{i = 1}^{\Nis} \frac{\iswgt{i}}{\sum_{i' = 1}^{\Nis} \iswgt{i'}}
    \sum_{j = 1}^{\Nin} \left( \frac{\inwgt{i}{j}}{\sum_{j' = 1}^{\Nin} \inwgt{i}{j'}} \right)^2 (\op h)^2 (\epart{i}{j}) \eqsp,
\end{multline*}
which tends in probability to $\asvarANone(\op h) + \condmeas{3} \prop(\der^2h^2) -  \condmeas{3}(\op^2 h)$ as the ancestor archipelago satisfies \AN{ass:AN:squared:islands} and \AN{ass:AN:mut-2}. Thus, the condition  \A{ass:DM:cons:tightness} in \autoref{thm:DM:A-1} is checked. In addition, \A{ass:DM:cons:Lindeberg} is checked using \autoref{lem:tightness}, as
\begin{equation*}
	\max_{i \in \intvect{1}{\Nis}} |\arr{i}| \leq \mbd + \bdconst \bd^2 \quad (\N \in \nsetpos) \eqsp,
\end{equation*}
where for $\N \in \nsetpos$, $\mbd = 0$ and
$$
    \begin{cases}
        \bdconst = \displaystyle \Nis \max_{i \in \intvect{1}{\Nis}} \frac{\iswgt{i}}{\sum_{i = 1}^{\Nis} \iswgt{i'}} \eqsp, \\
        \bd = \displaystyle \sqrt{\frac{\Nin}{\Nis}} \max_{i \in \intvect{1}{\Nis}} \left| \sum_{j = 1}^{\Nin} \frac{\inwgttd{i}{j}}{\sum_{j' = 1}^{\Nin} \inwgt{i}{j'}} h(\eparttd{i}{j}) \right| \eqsp,
    \end{cases}
$$
where $\{Ê\bdconst \}_{\N \in \nsetpos}$ is $\filt$-adapted and tight by \AN{def:tightness:island} and each
$Y_N$ has, by \cite[Lemma~4]{douc:garivier:moulines:olsson:2010}, since the input and output archipelagos satisfy \D{def:deviation:inequality:individual}\!, a tail of the form \eqref{eq:lemma:exp:tail} (with $\alpha = 1$). Thus, \A{ass:DM:cons:tightness} holds true, and we may conclude that the series $\sum_{i = 1}^{\Nis} \arr{i}$ and $\sum_{i = 1}^{\Nis} \cexp{\arr{i}}{\filt}$ tend to the same limit in probability.

We turn to the first term of \eqref{eq:mutation:decomposition:AN2} and show that this tends to zero in probability. Indeed, note that the absolute value of the same is, up to the factor $\sum_{i' = 1}^{\Nis} \iswgt{i'} / \sum_{i'' = 1}^{\Nis} \iswgttd{i''}$, which converges in probability by \eqref{consistency:mutation:1}, bounded by
$$
    \max_{i' \in \intvect{1}{\Nis}} \left| \frac{\iswgt{i'}}{\iswgttd{i'}} - \frac{1}{\targ \op \1{\stsptd}} \right| \Nin \sum_{i = 1}^{\Nis} \frac{\iswgt{i}}{\sum_{i'' = 1}^{\Nis} \iswgt{i''}} \left( \sum_{j = 1}^{\Nin} \frac{\inwgttd{i}{j}}{\sum_{j' = 1}^{\Nin} \inwgt{i}{j'}} h(\eparttd{i}{j}) \right)^2 \eqsp,
$$
where the first factor vanishes in probability by \autoref{lem:Hoeffding:ratio:island:weights} and \autoref{lem:inverse:proba}, and the convergence of the second factor was established above. This establishes \AN{ass:AN:squared:islands}.

To check Assumption~\AN{def:tightness:island}, consider the bound
\begin{multline} \label{mutation:deviation:decompo}
    \Nis \max_{i \in \intvect{1}{\Nis}} \dfrac{\iswgttd{i}}{\sum_{i' = 1}^{\Nis} \iswgttd{i'}} \\Ê\leq \supn{\der} \Nis \max_{i \in \intvect{1}{\Nis}} \dfrac{\iswgt{i}}{\sum_{i' = 1}^{\Nis} \iswgt{i'}} \left( \left|Ê\dfrac{\sum_{i' = 1}^{\Nis} \iswgt{i'}}{\sum_{i'' = 1}^{\Nis} \iswgttd{i''}} - \frac{1}{\targ \op \1{\stsptd}} \right|Ê+\frac{1}{\targ \op \1{\stsptd}} \right) \eqsp,
\end{multline}
where the second factor on the right hand side is tight as the ancestor archipelago is assumed to satisfy \AN{def:tightness:island}. Moreover, as the third factor tends to $1 / \targ \op \1{\stsptd}$ in probability by \eqref{consistency:mutation:1} we conclude that \AN{def:tightness:island} holds true also for the output.

In order to check \AN{ass:AN:squared:iswgts}, pick $h \in \bmf{\stfd}$ and decompose the quantity of interest according to
\begin{multline} \label{eq:AN3:key:decomposition}
    \Nis \sum_{i = 1}^{\Nis} \left( \frac{\iswgttd{i}}{\sum_{i' = 1}^{\Nis} \iswgttd{i'}} \right)^2 \sum_{j = 1}^{\Nin} \frac{\inwgttd{i}{j}}{\sum_{j' = 1}^{\Nin} \inwgttd{i}{j'}} h(\eparttd{i}{j}) \\
    = \Nis \left(\dfrac{\sum_{i = 1}^{\Nis} \iswgt{i}}{\sum_{i' = 1}^{\Nis} \iswgttd{i'}} \right)^2
    \sum_{i = 1}^{\Nis} \left( \frac{\iswgt{i}}{\sum_{i' = 1}^{\Nis} \iswgt{i'}} \right)^2 \left( \dfrac{\iswgttd{i}}{\iswgt{i}} - \targ \op \1{\stsptd}\right) \sum_{j = 1}^{\Nin} \frac{\inwgttd{i}{j}}{\sum_{j' = 1}^{\Nin} \inwgt{i}{j'}} h(\eparttd{i}{j}) \\
    + \Nis \targ \op \1{\stsptd} \left(\dfrac{\sum_{i = 1}^{\Nis} \iswgt{i}}{\sum_{i' = 1}^{\Nis} \iswgttd{i'}} \right)^2 \sum_{i = 1}^{\Nis} \left( \frac{\iswgt{i}}{\sum_{i' = 1}^{\Nis} \iswgt{i'}} \right)^2 \sum_{j = 1}^{\Nin} \frac{\inwgttd{i}{j}}{\sum_{j' = 1}^{\Nin} \inwgt{i}{j'}} h(\eparttd{i}{j}) \eqsp.
\end{multline}
In order to handle the second term of this decomposition, we apply \autoref{thm:DM:A-1} to the array
$$
    \arr{i} \eqdef \Nis \left( \frac{\iswgt{i}}{\sum_{i' = 1}^{\Nis} \iswgt{i'}} \right)^2 \sum_{j = 1}^{\Nin} \frac{\inwgttd{i}{j}}{\sum_{j' = 1}^{\Nin} \inwgt{i}{j'}} h(\eparttd{i}{j}) \quad (i \in \intvect{1}{\Nis}, \N \in \nsetpos)
$$
furnished with the filtration $\{Ê\filt \}_{\N \in \nsetpos}$ given by \autoref{eq:def:filt}. First, we observe that
$$
    \sum_{i = 1}^{\Nis} \cexp{\arr{i}}{\filt} = \Nis \sum_{i = 1}^{\Nis} \left( \frac{\iswgt{i}}{\sum_{i' = 1}^{\Nis} \iswgt{i'}} \right)^2 \sum_{j = 1}^{\Nin} \frac{\inwgt{i}{j}}{\sum_{j' = 1}^{\Nin} \inwgt{i}{j'}} \op h(\epart{i}{j}) \plim \condmeas{1} \op h \eqsp,
$$
as the ancestor archipelago satisfies Assumption~\AN{ass:AN:squared:iswgts}. Thus, the condition \A{ass:DM:cons:tightness} in \autoref{thm:DM:A-1} holds true. In addition, as
$$
    \max_{i \in \intvect{1}{\Nis}} |\arr{i}| \leq \supn{\der} \supn{h} \left( \sqrt{\Nis} \max_{i \in \intvect{1}{\Nis}} \dfrac{\iswgt{i}}{\sum_{i' = 1}^{\Nis} \iswgt{i'}} \right)^2 \eqsp,
$$
also \A{ass:DM:cons:Lindeberg} is verified by \autoref{lem:tightness} (applied with $X_N = Y_N = 0$) and the fact that the input archipelago satisfies \AN{def:tightness:island}. Consequently, the also series $ \sum_{i = 1}^{\Nis} \arr{i}$ tends in probability to the limit $\condmeas{1} \op h$, which, by \autoref{consistency:mutation:1}, implies that the second term of \eqref{eq:AN3:key:decomposition} tends to $\condmeas{1} \op h / \targ \op \1{\stsptd}$. To treat the first term of \eqref{eq:AN3:key:decomposition}, note that this is, up to the factor $\sum_{i' = 1}^{\Nis} \iswgt{i'} / \sum_{i'' = 1}^{\Nis} \iswgttd{i''}$, which converges in probability by \autoref{consistency:mutation:1}, bounded by
$$
    \Nis \max_{i \in \intvect{1}{\Nis}} \left| \dfrac{\iswgttd{i}}{\iswgt{i}} - \targ \op \1{\stsptd} \right|
    \sum_{i = 1}^{\Nis} \left( \frac{\iswgt{i}}{\sum_{i' = 1}^{\Nis} \iswgt{i'}} \right)^2 \sum_{j = 1}^{\Nin} \frac{\inwgttd{i}{j}}{\sum_{j' = 1}^{\Nin} \inwgt{i}{j'}} |h|(\eparttd{i}{j}) \eqsp,
$$
which tends to zero in probability by the previous computation and \autoref{lem:Hoeffding:ratio:island:weights}. This completes the proof of \AN{ass:AN:squared:iswgts}.


In order to prove \AN{ass:AN:mut-1}, introduce the array
$$
    \arr{i} \eqdef \N \left( \frac{\iswgt{i}}{\sum_{i' = 1}^{\Nis} \iswgt{i'}} \right)^2 \sum_{j = 1}^{\Nin} \left( \frac{\inwgttd{i}{j}}{\sum_{j' = 1}^{\Nin} \inwgt{i}{j'}} \right)^2 h(\eparttd{i}{j}) \quad (i \in \intvect{1}{\Nis}, \N \in \nsetpos)
$$
and equip the same with usual filtration $\{Ê\filt \}_{\N \in \nsetpos}$. With this notation, the quantity of interest in \AN{ass:AN:mut-1} can be written as
$$
    \left( \dfrac{\sum_{i = 1}^{\Nis} \iswgt{i}}{\sum_{i' = 1}^{\Nis} \iswgttd{i'}} \right)^2 \sum_{i = 1}^{\Nis} \arr{i} \eqsp,
$$
where the first factor tends to $1 / (\targ \op \1{\stsptd})^2$ by \autoref{lem:Hoeffding:ratio:island:weights}. Thus, it is enough to show that the second factor tends to $\condmeas{2} \prop(\der^2 h)$ in probability, and for this purpose we use \autoref{thm:DM:A-1}. As the ancestor archipelago satisfies \AN{ass:AN:mut-1}, the quantity
$$
    \sum_{i = 1}^{\Nis} \cexp{\arr{i}}{\filt} = \N \sum_{i = 1}^{\Nis} \left( \frac{\iswgt{i}}{\sum_{i' = 1}^{\Nis} \iswgt{i'}} \right)^2 \sum_{j = 1}^{\Nin} \left( \frac{\inwgt{i}{j}}{\sum_{j' = 1}^{\Nin} \inwgt{i}{j'}} \right)^2 R(\der^2 h)(\epart{i}{j})
$$
tends in probability to the desired limit $\condmeas{2} \prop(\der^2 h)$.
This implies the condition \A{ass:DM:cons:tightness} in \autoref{thm:DM:A-1}. In addition, \A{ass:DM:cons:Lindeberg} is checked using \autoref{lem:tightness}; indeed,
$$
    \max_{i \in \intvect{1}{\Nis}} |\arr{i}| \leq \inwgtbd \supn{\der}^2 \supn{h} \left( \sqrt{\Nis} \max_{i \in \intvect{1}{\Nis}} \dfrac{\iswgt{i}}{\sum_{i' = 1}^{\Nis} \iswgt{i'}} \right)^2 \max_{i \in \intvect{1}{\Nis}} \dfrac{\Nin}{\sum_{j = 1}^{\Nin} \inwgt{i}{j}} \eqsp,
$$
where the the right hand side is adapted to $\{Ê\filt \}_{\N \in \nsetpos}$ and vanishes in probability by \autoref{lem:inverse:proba}, as the ancestor archipelago satisfies \AN{def:tightness:island} and \D{def:deviation:inequality:individual}. This shows \AN{ass:AN:mut-1}.


Finally, in order to prove \AN{ass:AN:mut-2} we decompose the quantity of interest according to
\begin{multline} \label{eq:AN5:key:decomposition}
    \Nin \sum_{i = 1}^{\Nis} \frac{\iswgttd{i}}{\sum_{i' = 1}^{\Nis} \iswgttd{i'}} \sum_{j = 1}^{\Nin} \left( \frac{\inwgttd{i}{j}}{\sum_{j' = 1}^{\Nin} \inwgttd{i}{j'}} \right)^2 h(\eparttd{i}{j}) = \\
    \Nin \left( \dfrac{\sum_{i' = 1}^{\Nis} \iswgt{i'}}{\sum_{i'' = 1}^{\Nis} \iswgttd{i''}} \right)
    \sum_{i = 1}^{\Nis} \frac{\iswgt{i}}{\sum_{i' = 1}^{\Nis} \iswgt{i'}} \left( \frac{\iswgt{i}}{\iswgttd{i}} - \frac{1}{\targ \op \1{\stsptd}} \right) \sum_{j = 1}^{\Nin} \left( \frac{\inwgttd{i}{j}}{\sum_{j' = 1}^{\Nin} \inwgt{i}{j'}} \right)^2 h(\eparttd{i}{j}) \\
    +  \Nin \frac{1}{\targ \op \1{\stsptd}} \left( \dfrac{\sum_{i' = 1}^{\Nis} \iswgt{i'}}{\sum_{i'' = 1}^{\Nis} \iswgttd{i''}} \right) \sum_{i = 1}^{\Nis} \frac{\iswgt{i}}{\sum_{i' = 1}^{\Nis} \iswgt{i'}}  \sum_{j = 1}^{\Nin} \left( \frac{\inwgttd{i}{j}}{\sum_{j' = 1}^{\Nin} \inwgt{i}{j'}} \right)^2 h(\eparttd{i}{j}) \eqsp.
\end{multline}
To deal with the second term we reapply \autoref{thm:DM:A-1}, this time to the array
$$
    \arr{i} = \Nin \frac{\iswgt{i}}{\sum_{i' = 1}^{\Nis} \iswgt{i'}}  \sum_{j = 1}^{\Nin} \left( \frac{\inwgttd{i}{j}}{\sum_{j' = 1}^{\Nin} \inwgt{i}{j'}} \right)^2 h(\eparttd{i}{j}) \quad (i \in \intvect{1}{\Nis}, \N \in \nsetpos) \eqsp.
$$
As usual, we study first the series
\begin{multline*}
    \sum_{i = 1}^{\Nis} \cexp{\arr{i}}{\filt} = \Nin \sum_{i = 1}^{\Nis} \frac{\iswgt{i}}{\sum_{i' = 1}^{\Nis} \iswgt{i'}}  \sum_{j = 1}^{\Nin} \left( \frac{\inwgt{i}{j}}{\sum_{j' = 1}^{\Nin} \inwgt{i}{j'}} \right)^2 \prop(\der^2 h)(\epart{i}{j}) \\
    \plim \condmeas{3} \prop(\der^2 h) \eqsp,
\end{multline*}
where the limit is a consequence of the fact that the ancestor archipelago satisfies \AN{ass:AN:mut-2}. \\
This establishes \A{ass:DM:cons:tightness} in \autoref{thm:DM:A-1}. To check also \A{ass:DM:cons:Lindeberg}, consider the upper bound
$$
    \max_{i \in \intvect{1}{\Nis}} |\arr{i}| \leq \inwgtbd \supn{\der}^2 \supn{h} \max_{i \in \intvect{1}{\Nis}} \dfrac{\iswgt{i}}{\sum_{i' = 1}^{\Nis} \iswgt{i'}} \max_{i \in \intvect{1}{\Nis}} \dfrac{\Nin}{\sum_{j = 1}^{\Nin} \inwgt{i}{j}} \eqsp,
$$
which is $\{Ê\filt \}_{\N \in \nsetpos}$-adapted and tends to zero in probability by \autoref{lem:inverse:proba}, as the ancestor archipelago satisfies \C{def:consistency:max} and \D{def:deviation:inequality:individual}. Now, \autoref{thm:DM:A-1} guarantees that $\sum_{i = 1}^{\Nis} \arr{i}$ Êand $\sum_{i = 1}^{\Nis} \cexp{\arr{i}}{\filt}$Ê have the same limit $\condmeas{3} \prop(\der^2 h)$ in probability. Moreover, note that the second term of \eqref{eq:AN5:key:decomposition} is, up to the factor $\sum_{i' = 1}^{\Nis} \iswgt{i'} / \sum_{i'' = 1}^{\Nis} \iswgttd{i''}$, bounded by
\begin{multline*}
    \Nin \max_{i \in \intvect{1}{\Nis}} \left| \frac{\iswgt{i}}{\iswgttd{i}} - \frac{1}{\targ \op \1{\stsptd}} \right| \sum_{i = 1}^{\Nis} \frac{\iswgt{i}}{\sum_{i' = 1}^{\Nis} \iswgt{i'}} \sum_{j = 1}^{\Nin} \left( \frac{\inwgttd{i}{j}}{\sum_{j' = 1}^{\Nin} \inwgt{i}{j'}} \right)^2 |h|(\eparttd{i}{j}) \eqsp,
\end{multline*}
which tends to zero in probability by \eqref{convergence:proba:inverse} and \autoref{lem:inverse:proba}. Thus, also  \AN{ass:AN:mut-2} holds true.


\subsection{Proof of \autoref{cor:bound:of:the:variance}}
\label{proof_bound:var}
First, a prefatory lemma.
\begin{lemma} \label{lem:prefatory:stability}
Assume \MG{ass:strong:mixing:condition}. Then for all $(\ell, n) \in \nset^2$ such that $\ell \leq n$ and $h \in \bmf{\stfd}$,
\begin{equation} \label{proof:stability:1}
\supn{\dfrac{ \op_{\ell} \cdots \op_{n-1}(h-\targ[n]h)}{\targ[\ell] \op_{\ell} \cdots  \op_{n-1} \1{\stsp}}} \leq \dfrac{\rho^{n-\ell}}{1-\rho} \osc{h} \eqsp,
\end{equation}
where $\rho$ is defined in \eqref{eq:def:rho}.
\end{lemma}
\begin{proof}
For $x \in \stsp$, write
\begin{align}
    \lefteqn{\dfrac{ \op_{\ell} \cdots \op_{n-1}(h-\targ[n]h)(x)}{\targ[\ell] \op_{\ell} \cdots  \op_{n-1} \1{\stsp}}} \nonumber \\
    &= \dfrac{ \op_{\ell} \cdots \op_{n-1} h(x) }{\targ[\ell] \op_{\ell} \cdots \op_{n-1} \1{\stsp}}
    - \dfrac{\op_{\ell} \cdots \op_{n-1} \1{\stsp}(x) }{\targ[\ell] \op[\ell] \cdots \op_{n-1} \1{\stsp}} \targ[n]h \nonumber \\
    &= \dfrac{\op_{\ell} \cdots \op_{n-1} \1{\stsp}(x) }{\targ[\ell] \op_{\ell} \cdots \op_{n-1} \1{\stsp}} \label{eq:stab:comp-1}
    \left[ \dfrac{\op_{\ell} \cdots \op_{n-1} h(x)}{ \op_{\ell} \cdots \op_{n-1} \1{\stsp}(x)}
    - \dfrac{\targ[\ell] \op_{\ell} \cdots \op_{n-1}h}{\targ[\ell] \op_{\ell} \cdots \op_{n-1} \1{\stsp}}
    \right] \eqsp.
\end{align}
Note that since $\op_{\ell} \dots \op_{n-1} h(x) = \delta_x \op_{\ell} \dots \op_{n-1} h$ (where $\delta_x$ denotes the Dirac mass located at $x$) we may, under \MG{ass:strong:mixing:condition}, apply \cite[Proposition~10.20]{douc:moulines:stoffer:2014}, yielding the uniform bound
\begin{equation} \label{eq:stab:comp-2}
    \left| \dfrac{\delta_x \op[\ell] \cdots \op[n] h}{\delta_x \op[\ell] \cdots \op[n] \1{\stsp}}
    - \dfrac{\targ[\ell] \op[\ell] \cdots \op[n] h}{\targ[\ell] \op[\ell] \cdots \op[n] \1{\stsp}} \right|
    \leq \rho^{n - \ell} \osc{h} \quad (x \in \stsp) \eqsp.
\end{equation}
Combining \eqref{eq:stab:comp-1} and \eqref{eq:stab:comp-2} with the uniform bound
$$
\dfrac{\op_{\ell} \cdots \op_{n - 1} \1{\stsp}(x) }{\targ[\ell] \op_{\ell} \cdots \op_{n - 1} \1{\stsp}} \leq \dfrac{\ubmarkov}{\lbmarkov} = \dfrac{1}{1 - \rho} \quad (x \in \stsp)
$$
yields (\ref{proof:stability:1}).
\end{proof}

For arbitrary $(\ell, n) \in \nset^2$ with $\ell \leq n$, combining the identity
$$
    \targ[\ell] \op[\ell] \cdots \op[n - 1] \1{\stsp} = \targ[\ell] \op[\ell] \1{\stsp} \times \targ[\ell + 1] \op[\ell + 1] \cdots \op[n - 1] \1{\stsp}
$$
with the bound $\targ[\ell] \op[\ell - 1] \1{\stsp} \geq \lbop$Ê(the latter implied by  \MG{ass:strong:mixing:condition}(iii)) yields
\begin{equation*}
    \dfrac{\targ[\ell] \prop[\ell] \{Ê\der_{\ell}^2 \op[\ell + 1] \cdots \op[n - 1](h - \targ[n] h)^2 \}}{(\targ[\ell] \op[\ell] \cdots \op[n - 1] \1{\stsp})^2} \leq \lbop^{-1} \supn{\der_{\ell}} \supn{\dfrac{\op[\ell + 1] \cdots \op[n - 1](h - \targ[n] h)}{\targ[\ell + 1] \op[\ell + 1] \cdots \op[n - 1] \1{\stsp}}}^2 \eqsp.
\end{equation*}
Now, using \autoref{lem:prefatory:stability} we obtain
$$
    \dfrac{\targ[\ell] \prop[\ell] \{Ê\der_{\ell}^2 \op[\ell + 1] \cdots \op[n - 1](h - \targ[n] h)^2 \}}{(\targ[\ell] \op[\ell] \cdots \op[n - 1] \1{\stsp})^2} \leq w_+ \dfrac{ \rho^{2(n - \ell - 1)} }{(1 - \rho)^2 \lbop} \osc[2]{h} \eqsp.
$$
Finally, the proof of \autoref{cor:bound:of:the:variance} is concluded by summing up the terms.

\subsection{Proof of \autoref{cor:non-compact:case}}
\label{sec:proof:non-compact}

Since the $\ell$th terms of the asymptotic variances \eqref{eq:as:var:B2} and \eqref{eq:as:var:bootstrap} differ only by the multiplicative constant $n - \ell$, the proof follows straightforwardly by direct inspection of the proof of the analogous result for the standard bootstrap particle filter given in \cite[Theorem~11]{douc:moulines:olsson:2014} (which in turn is an adaptation of the proof of Theorem~10 in the same paper, providing the analogous result for the particle predictor). More specifically, the result is obtained by
\begin{itemize}
    \item embedding, using a trivial extension of Kolmogorov's extension theorem, the stationary sequence $\{ \randpar{p} \}_{p \in \nset}$ into a stationary process $\{ \randpar{p} \}_{p \in \zset}$ with doubly infinite time.
    \item bounding, for a given $n \in \nset$, using \cite[Equations~34--35]{douc:moulines:olsson:2014}, $\asvar[n] \langle \randpar{0:n} \rangle(h)$ by a quantity of form $c \sum_{\ell = 0}^n (n - \ell) \Delta_{n - \ell} \langle h \rangle (\randpar{-\infty:\ell - 1}, \randpar{\ell:n})$, where $c$ is a $\prob$-a.s. finite random variable and Êeach function $\Delta_m\langle h \rangle : \Zset^\infty \rightarrow \rset_+$, $m \in \nset$, is of the same type as the terms of the sum in \cite[Equation~35]{douc:moulines:olsson:2014}.
    \item using the stationarity to conclude that $\sum_{\ell = 0}^n (n - \ell) \Delta_{n - \ell} \langle h \rangle (\randpar{-\infty:\ell - 1}, \randpar{\ell:n})$ has the same distribution as $\sum_{\ell = 0}^n  \ell \Delta_\ell \langle h \rangle (\randpar{-\infty:- \ell - 1}, \randpar{- \ell:0})$.
    \item bounding, using \cite[Equation~39]{douc:moulines:olsson:2014}, each term of the sum as $\Delta_\ell \langle h \rangle (\randpar{-\infty:- \ell - 1}, \randpar{- \ell:0}) \leq d \beta^\ell$, $\prob$-a.s., where $d$ is a $\prob$-a.s. finite random variable and $\beta < 1$ is a constant. This shows that $\asvar[n] \langle \randpar{0:n} \rangle(h) \leq cd\sum_{\ell = 0}^\infty \ell \beta^\ell < \infty$, $\prob$-a.s., which concludes the proof.
\end{itemize}

\section*{Acknowledgment} The authors thank the editor and the anonymous referee for insightful comments that improved the presentation of the paper.


\appendix

\section{Technical results}
\label{appenbdix:technical:lemmas}

We first recall two results, obtained in \cite{douc:moulines:2008}, which are essential for the developments of the present paper.

\begin{theorem}[\cite{douc:moulines:2008}] \label{thm:DM:A-1}
    Let $(\Omega, \mathcal{A}, \{\genfd[\N] \}_{\N \in \nn}, \prob)$ be a filtered probability space. In addition, let, for a given sequence $\{\M[\N] \}_{\N \in \nn}$ of integers such that $\M[\N] \rightarrow \infty$ as $\N \rightarrow \infty$, $\{\arr{i} \}_{i = 1}^{\M[\N]}$, $\N \in \nn$, be a triangular array of random variables on $(\Omega, \mathcal{A}, \prob)$ such that for all $\N \in \nn$, the variables $\{Ê\arr{i} \}_{i = 1}^{\M[\N]}$ are conditionally independent given $\genfd[\N]$ with $\cexp[txt]{|\arr{i}|}{\genfd[\N]} < \infty$, $\prob\mbox{-a.s.}$, for all $i \in \intvect{1}{\M[\N]}$. Moreover, assume that
    \begin{hypA} \label{ass:DM:cons:tightness}
        $
        		\displaystyle \quad \lim_{\lambda \rightarrow \infty} \sup_{\N \in \nn} \prob \left( \sum_{i = 1}^{\M[\N]} \cexp[txt]{|\arr{i}|}{\genfd[\N]} \geq \lambda \right) = 0 \eqsp.
        $
    \end{hypA}
    \begin{hypA} \label{ass:DM:cons:Lindeberg}
        For all $\varepsilon > 0$, as $\N \rightarrow \infty$,
        $$
        		\sum_{i = 1}^{\M[\N]} \cexp{|\arr{i}| \1{\{Ê|\arr{i}| \geq \varepsilon \}}}{\genfd[\N]} \plim 0 \eqsp.
        $$
    \end{hypA}
    Then, as $\N \rightarrow \infty$,
    $$
        \max_{m \in \intvect{1}{\M[\N]}} \left|Ê\sum_{i = 1}^m \arr{i} - \sum_{i = 1}^m \cexp{\arr{i}}{\genfd[\N]} \right| \plim 0 \eqsp.
    $$
\end{theorem}

\begin{theorem}[{\cite{douc:moulines:2008}}] \label{thm:DM:A-3}
    Let the assumptions of Theorem~\ref{thm:DM:A-1} hold with $\cexp[txt]{\arr[2]{i}}{\genfd[\N]} < \infty$, $\prob\mbox{-a.s.}$, for all $i \in \intvect{1}{\M[\N]}$, and
    \A{ass:DM:cons:tightness} and \A{ass:DM:cons:Lindeberg} replaced by:
    \begin{hypB} \label{ass:DM:norm:asymptotic:variance}
        For some constant $\varsigma^2 > 0$, as $\N \rightarrow \infty$,
        $$
        		\sum_{i = 1}^{\M[\N]} \left( \cexp[txt]{\arr[2]{i}}{\genfd[\N]} - \E^2 \left[ \arr{i} \mid \genfd[\N] \right] \right) \plim \varsigma^2 \eqsp.
        $$
    \end{hypB}
    \begin{hypB} \label{ass:DM:norm:Lindeberg}
        For all $\varepsilon > 0$, as $\N \rightarrow \infty$,
        $$
        		\sum_{i = 1}^{\M[\N]} \cexp{\arr[2]{i} \1{\{Ê|\arr{i}| \geq \varepsilon \}}}{\genfd[\N] } \plim 0 \eqsp.
        $$
    \end{hypB}
    Then, for all $u \in \rset$, as $\N \rightarrow \infty$,
    $$
        \cexp{ \exp \left( \operatorname{i} u \sum_{i = 1}^{\M[\N]} \left\{ \arr{i} - \cexp[txt]{\arr{i}}{\genfd[\N]} \right\} \right) }{\genfd[\N]}  \plim \exp \left( -u^2 \varsigma^2 /2 \right) \eqsp.
    $$
\end{theorem}

The following lemma is useful when verifying the tightness conditions~\A{ass:DM:cons:Lindeberg} and \B{ass:DM:norm:Lindeberg}.


\begin{lemma} \label{lem:tightness}
    Let the $(\{Ê\M[\N] \}_{\N \in \nn}, \{ \arr{i} \}_{i = 1}^{\M[\N]}, \{Ê\filt \}_{\N \in \nn})$ be the triangular array given in Theorem~\ref{thm:DM:A-1}. Assume that there exist sequences $\{ \mbd \}_{\N \in \nn}$, $\{\bdconst \}_{\N \in \nn}$, and $\{ \bd \}_{\N \in \nn}$ of nonnegative random variables such that
    \begin{enumerate}[(i)]
        \item for all $\N \in \nn$, $\prob\mbox{-a.s.}$,
        $$
            \max_{i \in \intvect{1}{\M[\N]}} |\arr{i}| \leq \mbd + \bdconst \bd^2 \eqsp,
        $$
        \item $\{Ê\mbd \}_{\N \in \nn}$ and $\{Ê\bdconst \}_{\N \in \nn}$ are $\{Ê\genfd[\N] \}_{\N \in \nn}$-adapted and such that $\mbd \plim 0$ as $\N \rightarrow \infty$ and
        $$
            \lim_{\lambda \rightarrow \infty} \sup_{\N \in \nn} \prob \left( \bdconst \geq \lambda \right) = 0 \eqsp,
        $$
        \item for some $\alpha \in \{1,2\} $, $\genexp > 0$, $\Hcons > 0$, and $C > 0$, $\prob\mbox{-a.s.}$,
        \begin{equation} \label{eq:lemma:exp:tail}
            \prob \left( \bd \geq y \mid \genfd[\N] \right) \leq C \M[\N] \exp \left(- c \M[\N]^\nu y^{2\alpha} \right) \eqsp.
        \end{equation}
    \end{enumerate}
    Then for $p \in \{ 1, 2 \}$,
     \begin{multline*}
        \lim_{\lambda \rightarrow \infty} \sup_{\N \in \nn} \prob \left( \sum_{i = 1}^{\M[\N]} \cexp{|\arr[p]{i}|}{\genfd[\N]} \geq \lambda \right) = 0 \\
         \Rightarrow \sum_{i = 1}^{\M[\N]} \cexp{|\arr[p]{i}| \1{\{Ê|\arr{i}| \geq \varepsilon \}}}{\genfd[\N]} \plim 0 \eqsp, \quad \mbox{$\forall \varepsilon > 0,$ as } \N \rightarrow \infty.
    \end{multline*}
\end{lemma}

\begin{proof}
    We start with the case $p = 1$. First, note that for all $\upsilon > 0$,
    \begin{equation*}
        \begin{split}
            \cexp{\bd^2 \1{\{\bd^2 \geq \upsilon \}}}{\genfd[\N]}
            &= \int_\upsilon^\infty \prob \left( \bd \geq \sqrt{y} \mid \genfd[\N] \right) \, \rmd y
            + \upsilon \prob \left( \bd \geq \sqrt{\upsilon} \mid \genfd[\N] \right) \\
            &\leq C \M[\N] \int_{\upsilon}^{\infty} \exp \left(- \Hcons \M[\N]^{\genexp} y^{\alpha} \right) \rmd y + C \M[\N] \upsilon \exp \left(- \Hcons \M[\N]^{\genexp} \upsilon^{\alpha} \right) \eqsp,
        \end{split}
    \end{equation*}
    where we used the condition~(iii) in the second step. Thus,
    \begin{equation} \label{eq:cases:alpha}
        \cexp{\bd^2 \1{\{Ê\bd^2 \geq \upsilon \}}}{\genfd[\N]} \leq
        \begin{cases}
            \left( \upsilon + \M[\N]^{- \genexp} / c \right)
            C \M[\N] \exp \left(- \Hcons \M[\N]^{\genexp} \upsilon \right) & \mbox{for } \alpha = 1 \eqsp, \\
            \left(\upsilon + \M[\N]^{-\genexp} /(2 \Hcons \upsilon) \right) C \M[\N] \exp \left(- \Hcons \M[\N]^{\genexp} \upsilon^2 \right) & \mbox{for } \alpha = 2 \eqsp,
        \end{cases}
    \end{equation}
    using the standard upper tail bound for Gaussian distributions.
    In any case,
    \begin{equation} \label{eq:bd:bd}
        \M[\N] \cexp{\bd^2 \1{\{Ê\bd^2 \geq \upsilon \}}}{\genfd[\N]} \plim[N] 0 \eqsp.
    \end{equation}
    In addition, note that~(ii) implies, for all $\varepsilon' > 0$ and all $\delta > 0$, the existence of a constant $\lambda_\delta < \infty$ such that for all  $\lambda \geq \lambda_\delta$,
    \begin{equation} \label{eq:prob:bound:X}
        \sup_{\N \in \nn} \prob \left(  \1{\{Ê\bdconst \geq \lambda \}} \sum_{i = 1}^{\M[\N]} \cexp{|\arr{i}|}{\genfd[\N]} \geq \varepsilon' \right) \leq \sup_{\N \in \nn} \prob \left(  \bdconst \geq \lambda \right) \leq \delta \eqsp.
    \end{equation}
    Now, for any $\varepsilon > 0$ and $\lambda > 0$ the quantity of interest may be bounded as
    \begin{multline*}
        \sum_{i = 1}^{\M[\N]} \cexp{|\arr{i}| \1{\{Ê|\arr{i}| \geq \varepsilon \}}}{\genfd[\N]}
        \leq \M[\N] \mbd \prob \left( \bd \geq \sqrt{\frac{\varepsilon}{2 \lambda}} \mid \genfd[\N] \right) \\
        + \left( \1{\{Ê\bdconst \geq \lambda \}} + \1{\{Ê\mbd \geq \varepsilon / 2  \}} \right) \sum_{i = 1}^{\M[\N]} \cexp{|\arr{i}|}{\genfd[\N]}
        + \lambda \M[\N] \cexp{\bd^2 \1{\{Ê\bd^2 \geq \varepsilon / (2 \lambda) \}}}{\genfd[\N]} \eqsp,
    \end{multline*}
    where the upper bound may, by~(iii), \eqref{eq:bd:bd}, and \eqref{eq:prob:bound:X}, be made arbitrarily small in probability by increasing first $\lambda$ and then $\N$. This completes the proof in the case $p = 1$.

    We turn to the case $p = 2$. However, by letting $\tilde{U}_\N(i) \eqdef \arr[2]{i}$, $i \in \intvect{1}{\M[\N]}$, $\N \in \nsetpos$, and noting that $\max_{i \in \intvect{1}{\M[\N]}} \tilde{U}_\N(i) \leq \tilde{V}_\N  + \tilde{X}_\N \tilde{Y}_\N^2$, where $\tilde{V}_\N \eqdef 2 \mbd^2$, $\tilde{X}_\N \eqdef 2 \bdconst^2$, and $\tilde{Y}_\N \eqdef \bd^2$, we thus realise that the proof of the case $p = 1$ goes through if we can verify that \eqref{eq:bd:bd} and \eqref{eq:prob:bound:X} hold true when $\bdconst$, $\bd$, and $\{\arr{i} \}_{i = 1}^{\M[\N]}$ are replaced by $\tilde{X}_\N$, $\tilde{Y}_\N$, and $\{\tilde{U}_\N(i) \}_{i = 1}^{\M[\N]}$, respectively. Nevertheless, \eqref{eq:prob:bound:X} holds straightforwardly as tightness of $\{Ê\bdconst \}_{\N \in \nn}$ implies tightness of $\{ \tilde{X}_\NÊ\}_{\N \in \nn}$. Moreover, using condition (iii) one shows, along previous lines, that
    \begin{equation*}
        \begin{split}
            \cexp{\tilde{Y}_\N^2 \1{\{Ê\tilde{Y}_\N^2 \geq \upsilon \}}}{\genfd[\N]} &= \int_\upsilon^\infty \prob \left( \bd \geq \sqrt[4]{y} \mid \genfd[\N] \right) \, \rmd y + \upsilon \prob \left( \bd \geq \sqrt[4]{\upsilon} \mid \genfd[\N] \right) \\
            &\leq C \M[\N] \int_{\upsilon}^{\infty} \exp \left(- \Hcons \M[\N]^{\genexp} y^{\alpha/2} \right) \rmd y + \upsilon C \M[\N] \exp \left(- \Hcons \M[\N]^{\genexp} \upsilon^{\alpha/2} \right) \eqsp.
        \end{split}
    \end{equation*}
    For $\alpha = 1$,
    $$
        \cexp{\tilde{Y}_\N^2 \1{\{Ê\tilde{Y}_\N^2 \geq \upsilon \}}}{\genfd[\N]}
        \leq \left( 2\sqrt{\upsilon} \M[\N]^{- \genexp} / c + 2  \M[\N]^{- 2 \genexp} / c^2 +  \upsilon \right) C \M[\N] \exp \left(- \Hcons \M[\N]^{\genexp} \sqrt{\upsilon} \right) \eqsp.
    $$
    while the case $\alpha = 2$ corresponds to the first case of \eqref{eq:cases:alpha}. Consequently, as $\N \rightarrow \infty$,
    $$
        \M[\N] \cexp{\tilde{Y}_\N^2 \1{\{\tilde{Y}_\N^2 \geq \upsilon \}}}{\genfd[\N]} \plim 0 \eqsp,
    $$
    which completes the proof.
\end{proof}


\begin{lemma} \label{lem:inverse:proba}
    Let $a \in \rset$ be nonzero, $a \neq 0$, and let $\{\X{i} \}_{i = 1}^{\Nis}$, $\Nis \in \nsetpos$, be random variables such that $\X{i} \neq 0$ for all $i \in \intvect{1}{\Nis}$. Assume that $\max_{i \in \intvect{1}{\Nis}} | \X{i} - a | \plim 0$ as $\N \rightarrow \infty$. Then
    $$
        \max_{i \in \intvect{1}{\Nis}} \left| \X[-1]{i} - a^{-1} \right| \plim 0.
    $$
\end{lemma}
\begin{proof} Pick $\e > 0 $; then we may write for all $\et > 0$,
    \begin{multline*}
        \prob\left( \max_{i \in \intvect{1}{\Nis}} \left| \X[-1]{i} - a^{-1} \right| \geq \e \right) \leq \prob \left( \max_{i \in \intvect{1}{\Nis}} | \X{i} - a | \geq \et \right) \\
        + \prob \left( \max_{i \in \intvect{1}{\Nis}} \left| \X[-1]{i} - a^{-1} \right| \geq \e, \max_{i \in \intvect{1}{\Nis}} | \X{i} - a | < \et \right) \eqsp,
    \end{multline*}
    where the first term tends to zero as $\N$ tends to infinity for any $\eta$ by assumption.
    For all $i \in \intvect{1}{\Nis}$, there exists, by Taylor's formula, $\zeta_\N(i) \in (\X{i} \wedge a, \X{i} \vee a)$ such that $| \X[-1]{i} - a^{-1}| = \zeta_\N^{-2}(i) |\X{i} - a|$. Thus, if $a > 0$ and $0 < \et < a$,
    \begin{multline*}
        \prob \left( \max_{i \in \intvect{1}{\Nis}} \left| \X[-1]{i} - a^{-1} \right| \geq \e, \max_{i \in \intvect{1}{\Nis}} | \X{i} - a | < \et \right) \\
        \leq \PP\left(\max_{i \in \intvect{1}{\Nis}}  | \X{i} - a | > \e \{a - \et\}^2 \right) \eqsp,
    \end{multline*}
    where the right hand side tends, by assumption, to zero as $\N$ tends to infinity. ÊOn the other hand, if $a < 0$ and $0 < \et < -a$,
    \begin{multline*}
        \prob \left(\max_{i \in \intvect{1}{\Nis}} \left| \X[-1]{i} - a^{-1} \right| \geq \e, \max_{i \in \intvect{1}{\Nis}} | \X{i} - a | < \et \right) \\
        \leq \prob \left( \max_{i \in \intvect{1}{\Nis}}  | \X{i} - a | > \e\{a + \et\}^2 \right) \eqsp,
    \end{multline*}
    where again the right hand side tends to zero. This concludes the proof.
\end{proof}


\begin{lemma}  \label{lemma:asind}
    Let $\{Ê\evar[\N] \}_{\N \in \nset}$ be a sequence of random variables such that for some constant $z_\infty \in \RR$, as $\N \rightarrow \infty$,
    $$
        \E \left[ \evar[\N] \right] \rightarrow z_\infty
    $$
    and is uniformly bounded by some constant $z^+ \in \RR$. Let $\{\pvar[\N] \}$ be a sequence of random variables that (i) converges in probability to some constant $x_\infty \in \RR$ and (ii) is dominated by some integrable random variable. Then,
    as $\N \rightarrow \infty$,
    $$
        \E \left[ \pvar[\N] \evar[\N] \right] \rightarrow x_\infty z_\infty \eqsp.
    $$
\end{lemma}

\begin{proof}
    The result is obtained straightforwardly by writing
    $$
        \left|\E \left[ \pvar[\N] \evar[\N] \right]- x_\infty z_\infty \right|
        \leq z^+ \E \left[ |\pvar[\N] - x_\infty|\right]+ | x_\infty | \left| \E\left[ \evar[\N] \right] - z_\infty \right| \eqsp,
    $$
    where the right hand side tends to zero as $\N$ tends to infinity by assumption and dominated convergence.
\end{proof}

\bibliographystyle{plain}
\bibliography{biblio}

\begin{thebibliography}{10}

\bibitem{balasingam:bolic:djuric:miguez:2011}
B.~Balasingam, M.~Bolić, P.~M. Djurić, and J.~Míguez.
\newblock Efficient distributed resampling for particle filters.
\newblock In {\em 2011 IEEE International Conference on Acoustics, Speech and
  Signal Processing (ICASSP)}, pages 3772--3775, May 2011.

\bibitem{bolic:djuric:hong:2005}
M.~Bolic, P.~M. Djuric, and Sangjin Hong.
\newblock Resampling algorithms and architectures for distributed particle
  filters.
\newblock {\em IEEE Transactions on Signal Processing}, 53(7):2442--2450, July
  2005.

\bibitem{cappe:moulines:2005}
O.~Capp\'{e} and E.~Moulines.
\newblock On the use of particle filtering for maximum likelihood parameter
  estimation.
\newblock In {\em European Signal Processing Conference (EUSIPCO)}, Antalya,
  Turkey, September 2005.

\bibitem{cappe:moulines:ryden:2005}
O.~Capp\'{e}, E.~Moulines, and T.~Ryd\'{e}n.
\newblock {\em Inference in Hidden {M}arkov Models}.
\newblock Springer, 2005.

\bibitem{cerou:delmoral:furon:guyader:2012}
F.~C{\'e}rou, P.~Del~Moral, T.~Furon, and A.~Guyader.
\newblock Sequential {M}onte {C}arlo for rare event estimation.
\newblock {\em Stat. Comput.}, 22(3):795--808, 2012.

\bibitem{chopin:2002}
N.~Chopin.
\newblock A sequential particle filter method for static models.
\newblock {\em Biometrika}, 89:539--552, 2002.

\bibitem{chopin:2004}
N.~Chopin.
\newblock Central limit theorem for sequential {M}onte {C}arlo methods and its
  application to {B}ayesian inference.
\newblock {\em Ann. Statist.}, 32(6):2385--2411, 2004.

\bibitem{chopin:jacob:papaspiliopoulos:2012}
N.~Chopin, P.~Jacob, and O.~Papaspiliopoulos.
\newblock {SMC2}: A sequential {M}onte {C}arlo algorithm with particle {M}arkov
  chain {M}onte {C}arlo updates.
\newblock {\em J. Roy. Statist. Soc. B}, 75(3):397--426, 2013.

\bibitem{crisan:rozovskii:2011}
D.~Crisan and B.~L. Rozovskii, editors.
\newblock {\em The Oxford handbook of nonlinear filtering}.
\newblock Oxford N.Y. Oxford University Press, 2011.

\bibitem{delmoral:2004}
P.~{Del Moral}.
\newblock {\em {F}eynman-Kac {F}ormulae. {G}enealogical and Interacting
  Particle Systems with Applications}.
\newblock Springer, 2004.

\bibitem{delmoral:garnier:2005}
P.~{Del Moral} and J.~Garnier.
\newblock Genealogical particle analysis of rare events.
\newblock {\em Ann. Appl. Probab.}, 15(4):2496--2534, 2005.

\bibitem{delmoral:guionnet:1999}
P.~Del~Moral and A.~Guionnet.
\newblock Central limit theorem for nonlinear filtering and interacting
  particle systems.
\newblock {\em Ann. Appl. Probab.}, 9(2):275--297, 1999.

\bibitem{delmoral:guionnet:2001}
P.~{Del Moral} and A.~Guionnet.
\newblock On the stability of interacting processes with applications to
  filtering and genetic algorithms.
\newblock {\em Annales de l'Institut Henri Poincar\'e}, 37:155--194, 2001.

\bibitem{verge:delmoral:moulines:olsson:2014:supp}
P.~Del~Moral, C.~Verg\'{e}, , E.~Moulines, and J.~Olsson.
\newblock Supplement to ``{C}onvergence properties of weighted particle islands
  with application to the double bootstrap algorithm''.
\newblock Supplementary material, 2017.

\bibitem{douc:garivier:moulines:olsson:2010}
R.~Douc, A.~Garivier, E.~Moulines, and J.~Olsson.
\newblock Sequential {M}onte {C}arlo smoothing for general state space hidden
  {M}arkov models.
\newblock {\em Ann. Appl. Probab.}, 21(6):2109--2145, 2011.

\bibitem{douc:moulines:2008}
R.~Douc and E.~Moulines.
\newblock Limit theorems for weighted samples with applications to sequential
  {M}onte {C}arlo methods.
\newblock {\em Ann. Statist.}, 36(5):2344--2376, 2008.

\bibitem{douc:moulines:olsson:2008}
R.~Douc, \'E. Moulines, and J.~Olsson.
\newblock Optimality of the auxiliary particle filter.
\newblock {\em Probab. Math. Statist.}, 29(1):1--28, 2009.

\bibitem{douc:moulines:olsson:2014}
R.~Douc, E.~Moulines, and J.~Olsson.
\newblock Long-term stability of sequential {M}onte {C}arlo methods under
  verifiable conditions.
\newblock {\em Ann. Appl. Probab.}, 24(5):1767--1802, 2014.

\bibitem{douc:moulines:stoffer:2014}
R.~Douc, E.~Moulines, and D.~Stoffer.
\newblock {\em Nonlinear Time Series: Theory, Methods and Applications with R
  Examples}.
\newblock Chapman \& Hall/CRC Texts in Statistical Science, 2014.

\bibitem{doucet:defreitas:gordon:2001}
A.~Doucet, N.~{De Freitas}, and N.~Gordon, editors.
\newblock {\em Sequential {M}onte {C}arlo Methods in Practice}.
\newblock Springer, New York, 2001.

\bibitem{heine:whiteley:2015}
K.~{Heine} and N.~{Whiteley}.
\newblock {Fluctuations, stability and instability of a distributed particle
  filter with local exchange}.
\newblock {\em ArXiv e-prints}, May 2015.

\bibitem{kong:liu:wong:1994}
A.~Kong, J.~S. Liu, and W.~Wong.
\newblock Sequential imputation and {B}ayesian missing data problems.
\newblock {\em J. Am. Statist. Assoc.}, 89(278-288):590--599, 1994.

\bibitem{kuensch:2005}
H.~R. K\"{u}nsch.
\newblock Recursive {M}onte-{C}arlo filters: algorithms and theoretical
  analysis.
\newblock {\em Ann. Statist.}, 33(5):1983--2021, 2005.

\bibitem{liu:1996}
J.~S. Liu.
\newblock Metropolized independent sampling with comparisons to rejection
  sampling and importance sampling.
\newblock {\em Stat. Comput.}, 6:113--119, 1996.

\bibitem{liu:2001}
J.S. Liu.
\newblock {\em {M}onte {C}arlo Strategies in Scientific Computing}.
\newblock Springer, New York, 2001.

\bibitem{olsson:cappe:douc:moulines:2006}
J.~Olsson, O.~Capp\'e, R.~Douc, and E.~Moulines.
\newblock Sequential {M}onte {C}arlo smoothing with application to parameter
  estimation in non-linear state space models.
\newblock {\em Bernoulli}, 14(1):155--179, 2008.
\newblock arXiv:math.ST/0609514.

\bibitem{pitt:shephard:1999}
M.~K. Pitt and N.~Shephard.
\newblock Filtering via simulation: Auxiliary particle filters.
\newblock {\em J. Am. Statist. Assoc.}, 94(446):590--599, 1999.

\bibitem{ristic:arulampalam:gordon:2004}
B.~Ristic, M.~Arulampalam, and A.~Gordon.
\newblock {\em Beyond Kalman Filters: Particle Filters for Target Tracking}.
\newblock Artech House, 2004.

\bibitem{rosen:medvedev:2013}
O.~Rosen and A.~Medvedev.
\newblock Efficient parallel implementation of state estimation algorithms on
  multicore platforms.
\newblock {\em IEEE Transactions on Control Systems Technology},
  21(1):107--120, Jan 2013.

\bibitem{sankaranarayanan:srivastava:chellappa:2008}
A.~C. Sankaranarayanan, A.~Srivastava, and R.~Chellappa.
\newblock Algorithmic and architectural optimizations for computationally
  efficient particle filtering.
\newblock {\em IEEE Transactions on Image Processing}, 17(5):737--748, May
  2008.

\bibitem{sutharsan:kirubaraja:lang:2012}
S.~Sutharsan, T.~Kirubarajan, T.~Lang, and M.~Mcdonald.
\newblock An optimization-based parallel particle filter for multitarget
  tracking.
\newblock {\em IEEE Transactions on Aerospace and Electronic Systems},
  48(2):1601--1618, APRIL 2012.

\bibitem{verge:dubarry:delmoral:moulines:2013}
C.~Verg\'{e}, C.~Dubarry, P.~Del~Moral, and E.~Moulines.
\newblock On parallel implementation of sequential {M}onte {C}arlo methods: the
  island particle model.
\newblock {\em Statistics and Computing}, 23, 2013.

\end{thebibliography}

\end{document}